\documentclass[smallcondensed]{svjour3}
\usepackage[english]{babel}

\usepackage{amsmath}
\usepackage{amsfonts}
\usepackage{amssymb}
\usepackage{calc}
\usepackage{graphicx}
\usepackage{algorithm}
\usepackage{algpseudocode}
\usepackage{array}
\usepackage{booktabs}
\usepackage{hyperref}
\usepackage{listings}
\usepackage{pgfplots}
\usepackage{tikz}
\usetikzlibrary{positioning, shapes, decorations.markings, arrows, calc, patterns,backgrounds}
\usepackage{todonotes}
\usepackage{subcaption} 
\usepackage{booktabs}
\usepackage{ifthen}
\usepackage{xstring}
\usepackage{formal}
\usepackage{macros}
\usepackage{verbatimbox} % used for figo.tex
\usepackage[numbers]{natbib}

\newcounter{lst}[section]

\lstdefinestyle{easslisting}{basicstyle=\footnotesize\sffamily, mathescape=true, frame=tb,  numbers=right, numberstyle=\footnotesize, stepnumber=1, numbersep=-5pt, captionpos=b}

\lstdefinestyle{listingl}{basicstyle=\footnotesize\sffamily, mathescape=true, numberstyle=\footnotesize, stepnumber=1, numbersep=-5pt, captionpos=b}

\lstdefinestyle{eass}{basicstyle=\sffamily, mathescape=true}

\lstnewenvironment{listing}[3]{
  \noindent        
  \refstepcounter{lst}         
  \label{code:#1}  
\begin{tabular}{p{.97\columnwidth}} \\ \hline  {\normalsize \textbf{Code  fragment \arabic{section}.\arabic{lst}} #2} \\  \end{tabular} 
\lstset{language=#3,          
  basicstyle=\footnotesize\sffamily,
  xleftmargin=10pt,    
  mathescape=true,    
  frame=tb,    
  numbers=right,    
  numberstyle=\footnotesize,     
  stepnumber=1,     
  numbersep=-5pt}}{}

\lstdefinelanguage{Pseudocode}{%
    morekeywords={mod,or,and,if,end,then,else,loop,while},
    morecomment=[l]{//},
 literate= {<-}{{$\leftarrow$}{$\:$}}2
           {.B}{{${\cal B}$}}2
           {.G}{{${\cal G}$}}2
           {lnot}{{$\sim$}}2
           {assert_shared}{{$+_{\Sigma}$}}2
           {remove_shared}{{$-_{\Sigma}$}}2
           {True}{{$\top$}}2
           {(perform)}{}0
}

\lstdefinelanguage{Gwendolen}{%
    morekeywords={Plans,Initial,Beliefs,Goals,name,fof-parse,Rules,Belief},
    morecomment=[l]{//},
 literate= {<-}{{$\leftarrow$}{$\:$}}2
           {.B}{{${\cal B}$}}2
           {.G}{{${\cal G}$}}2
           {lnot}{{$\sim$}}2
           {assert_shared}{{$+_{\Sigma}$}}2
           {remove_shared}{{$-_{\Sigma}$}}2
           {True}{{$\top$}}2
           {(perform)}{}0
}

\title{Multi-Scale Verification of Distributed Synchronisation\thanks{
	This work was supported by
    the Sir Joseph Rotblat Alumni Scholarship at Liverpool
and 
    the Engineering and Physical Sciences Research Council, under grants EP/N007565/1 (S4: Science of Sensor Systems Software),  EP/L024845/1 (Verifiable Autonomy), 
and the 
 FAIR-SPACE (EP/R026092/1), RAIN (EP/R026084/1).
and ORCA (EP/R026173/1) RAI Hubs.
}
}

\author{ Paul Gainer \and
Sven Linker \and
Clare Dixon \and
Ullrich Hustadt \and
Michael Fisher
}
\begin{document}
\maketitle

\begin{abstract}
Algorithms for the synchronisation of clocks across networks are both
common and important within distributed systems. We here address not
only the formal modelling of these algorithms, but also the formal
verification of their behaviour. Of particular importance is the
strong link between the very different levels of abstraction at which
the algorithms may be verified. Our contribution is primarily the
formalisation of this connection between individual models and
population-based models, and the subsequent verification that
is then possible.
While the technique is applicable across a range of synchronisation
algorithms, we particularly focus on the synchronisation
of (biologically-inspired) pulse-coupled oscillators, a widely used
approach in practical distributed systems. For this application
domain, different levels of abstraction are crucial: models based
on the  behaviour of an individual process
are able to capture the details of distinguished nodes in possibly 
heterogenous networks, where each node may exhibit different behaviour.
On the other hand, collective models assume homogeneous sets
of processes, and allow the behaviour of the network to be analysed at the
global level. System-wide parameters may be easily adjusted,
for example environmental factors
inhibiting the reliability of the shared communication medium.
This work provides a formal bridge across the ``abstraction gap''
separating the individual models and the population-based models for
this important class of synchronisation algorithms.
\end{abstract}

\section{Introduction}
\label{sec:intro}

 Small computing devices comprising networks, be it commercial wireless sensor networks, or
communicating devices in the Internet of Things, become increasingly common.
However, to enable these devices to communicate efficiently, they have to employ
methods to use the shared communication medium without too many conflicts, e.g., 
in the form of collisions. Several protocols to organise shared medium access have
been developed and analysed~\cite{Akyildiz2002,yick2008survey}. These
protocols typically identify a 
common time frame and divide this frame into slots associated to each node. Thus
every node has an allocated time slot that it may use to communicate its messages
onto the shared medium. 

Such an approach introduces the need for a common clock 
between the nodes, i.e., they need to synchronise. A valuable approach 
to achieve synchrony of nodes is the implementation of biologically-inspired 
\emph{pulse-coupled oscillators} (PCOs)~\cite{mirollo1990synchronization}. A network 
of PCOs synchronises in the following way: all oscillators have a similar
\emph{clock cycle} at the end of which they \emph{fire}. That is, they transmit
a broadcast message which is received by all oscillators in their
communication range. These oscillators then adjust their own position within
their clock cycle according to a \emph{phase response function}.
Depending on the concrete implementation, they may move their current
position within the clock cycle closer to its end, or closer to its start. 

Most analyses of the synchronisation behaviour of PCOs are concerned with
continouous clock cycles, i.e., where clocks take real values
from the interval \([0,1]\). However, the smaller
devices get, the more important it is to save memory and computing time for 
such a low-level functionality. Even a floating point number may need too much memory, 
compared to an implementation with, for example, a four-bit vector. Hence, in previous
work, we chose to analyse the behaviour of \emph{discrete time PCOs} \cite{gainer2017investigating}. 

In contrast to continuous time PCOs, networks of discrete time PCOs are not always
guaranteed to synchronise. Instead, whether they synchronise or not depends on the
type of coupling between the oscillators and their common phase-response function.
We analysed the behaviour of such networks for different parameters via
model-checking, to check both qualitatively for which parameters the networks
synchronise, as well as quantitatively for how long they need to achieve a
synchronised state and how
much energy  is used to achieve this~\cite{gainer2017power}.
In the context of large numbers of single oscillators, for example in the context
of wireless sensor networks, the well-known state-space explosion problem of
the model-checking approach is extremely important ~\cite{chen13review}.
 We formalised a network of oscillators as \emph{population models}~\cite{donaldson2006symmetry} which 
exploit the behavioural homogeneity of the nodes  to encode the global state efficiently.
This allows the network size to be increased above what would be feasible when distinguishing
each node.
%allowed
%for a condensation of the state-space, and thus to increase the network size and
%other parameters above
%the numbers typically achieved by explicit formalisation of single oscillators. 
But the
construction of a population model from a given oscillator specification is not straightforward,
and in particular, it is not obvious whether the constructed population model correctly 
reflects the behaviour of the oscillators. This results in an `abstraction gap': after abstracting into populations, how can we be sure that the abstraction process was correct and that the results of verification of population models actually hold for the concrete models on which they are based? 

%% population models exploit the behavioural homogeneity of nodes in a network and efficiently encode the global state of the whole network as a vector. Not distinguishing between nodes sharing the same state allows much larger networks to be analysed.
%
%

In this paper, we remedy this lack of certainty, by proving the correspondence of 
our population model with an explicit formalisation of the oscillators. To that
end, we present the concrete oscillator model as well as its formalisation as 
a discrete-time Markov chain. Subsequently we describe the corresponding population model,
and show how we can, in addition to the abstraction created by the populations,
 reduce the state space even further to facilitate the analysis. Finally, we prove
that the behaviour of a network of concrete oscillators can be simulated by the population model.
We cannot prove a one-to-one correspondence, since the concrete model 
implicitly includes the possibility of identifying individual oscillators, which is exactly what the population model abstracts from. However, by providing a formal 
notion of abstraction, we prove that population models are a truthful abstraction of concrete models.

The paper is structured as follows. In Sect.~\ref{sec:related}, we review a selection of related work,
both for models of pulse-coupled oscillators, as well as approaches for their verification. After an
introduction of preliminary notions in Sect.~\ref{sec:preliminaries},
we present the concrete model of single oscillators, both as an algorithm and as a discrete-time Markov
chain derived from this algorithm, in Sect.~\ref{sec:concrete}. The abstract model in terms 
of population models and proofs about their properties are contained in Sect.~\ref{sec:population}.
In Sect.~\ref{sec:connection}, we prove the correspondence between these two types of models, and conclude
 in Sect.~\ref{sec:conclusion}. 

% PCO-based synchronisation algorithms can been verified using traditional methods like experimentation and simulation~\cite{breza13phd}, as well as formal methods like model checking~\cite{gainer2017investigating}. The latter can provide exhaustive proof of the properties of such algorithms, but at the cost of a state space explosion
% \mwcomment{We've mentioned this earlier as well. Re-write?} due to the large number of states in \emph{concrete models}, in which nodes and clocks are not abstracted~\cite{chen13review}. 
% This explosion in the number of states can be mitigated through the use of \emph{population models}, in which the nodes and the clocks which run within them are abstracted into populations of nodes, permitting verification of larger numbers of nodes. However, there remains an `abstraction gap': after abstracting into populations, how can we be sure that the abstraction process was correct and that the results of verification of population models actually hold for the concrete models on which they are based? In the following sections we aim to remedy the abstraction gap between population models and concrete models by formalising the definitions of population models, concrete models, and the abstraction between them, thus proving that population models are a truthful abstraction of concrete models. 

%%% Local Variables: 
%%% mode: latex
%%% TeX-master: "sync_journal"
%%% End: 
\section{Related Work}
\label{sec:related}

The canonical model of pulse-coupled oscillators, and their synchronisation, 
was formulated by Mirollo
and Strogatz~\cite{mirollo1990synchronization}, and based on Peskin's
model of a cardiac pacemaker~\cite{peskin1975mathematical}. 
Here the progression of an oscillator through its oscillation cycle is given
by a real value in the interval $[0,1]$. Mirollo
and Strogatz proved that with a convex phase
response function, a network of mutually coupled oscillators always converges, i.e.,
their position within the oscillation cycle eventually coincides. 
Such a model has been shown to be applicable to the clock
synchronisation of wireless sensor nodes~\cite{tyrrell2006fireflies} and
swarms of robots~\cite{perez2015firefly}.
 
Synchronisation algorithms based on pulse-coupled oscillators are often benefitial in unreliable, decentralised networks, where other synchronisation 
algorithms are not appropriate. For example, the Flooding Time Synchronisation Protocol (FTSP)~\cite{Maroti:2004:FTS:1031495.1031501} requires the 
use of an arbitrary root node. In situations where the root becomes unavailable due to communication failure or power outage, FTSP will have to 
assign another root node. When implemented on unreliable, decentralised networks, FTSP may spend considerable resources on repeatedly assigning root nodes, which may slow down or prevent synchronisation~\cite{breza13phd}. 
Other algorithms such as the Berkeley algorithm~\cite{berkeley-alg} and Cristian's algorithm~\cite{Cristian1989} require the use of centralised time servers, which is problematic for unreliable, decentralised networks. 

%which will not be appropriate for networks with unreliable connections. FTSP also requires t

Several decentralised network algorithms for synchronisation are based on pulse-coupled oscillators~\cite{tyrrell2006fireflies,werner2005firefly}. For example, the Gradient Time Synchronisation Protocol (GTSP) by Sommer and 
Wattenhofer~\cite{sommer09} achieves synchronisation by having nodes send their current clock value to their neighbours. Each node then calculates the average of the clock values received and its own clock value. This process is then repeated to maintain synchronisation. Another approach to synchronisation, the Pulse-Coupled Oscillator Protocol~\cite{Pagliari:2007:DIP:1322263.1322308}, makes use of refractory periods after sending messages containing time information. During the refractory period, no more messages are sent, which reduces network bandwidth and energy usage. A similar approach is used in the FiGo protocol~\cite{breza13phd}, which combines biologically inspired synchronisation with information distribution via gossiping. All of these approaches use  different phase response functions. 

In general, synchronisation algorithms based on PCOs are more robust for unreliable networks, as they do not require centralised nodes and can work with only partial network connectivity~\cite{breza13phd}. They are particularly useful for battery-powered nodes in wireless networks, as the node can be placed in a low-power node during the refractory period, thus reducing energy usage. (The clock keeps ticking even in low-power mode, thanks to the design of  microcontrollers such as the `Atmel ATmega128L'~\cite{atmel}.)

Synchronisation of clocks for networks of nodes  has been investigated
from different perspectives.  Heidarian et al.~\cite{heidarian2012analysis}
analysed the behaviour of a synchronisation protocol based on time allocation
slots for up to four nodes and different topologies, from fully connected
networks to line topologies. They modelled the 
protocol as timed automata~\cite{alur1994theory}, and used the model-checker
UPPAAL~\cite{Behrmann2006} to examine its worst-case behaviour.
Their model is based on continuous time, and in particular, they
did not model pulse-coupled oscillators.

 Bartocci et al.~\cite{bartocci2010detecting} described pulse-coupled oscillators
as extended timed automata with suitable semantics to model their peculiarities. 
They defined a dedicated logic to analyse
the behaviour of a network of such automata along traces, and used 
a pacemaker as a case study to verify the eventual synchronisation and
the time needed to achieve this. 

Our models and methods are slightly different to all of these approaches. 
This is, of course, evident for all the mentioned work that is not concerned with
pulse-coupled oscillators. 
However, we also define the oscillation cycle  to consist of discrete steps. To
the best of our knowledge, with the exception the paper by Webster et al.~\cite{Webster2018} and  our previous work~\cite{gainer2017investigating,gainer2017power}, 
there is no other work concerned with PCOs with discrete oscillation cycles.
Furthermore, all of these approaches distinguish between single oscillators in the network, 
% each oscillator is explicitly modelled 
%as a single entity, 
while the properties of interest relate to global behaviour.
%analysis is concerned with the global
%behaviour of the network.
% This prevents
This discrepancy between local modelling and global analysis % \slcomment{better terms for ``local'' and ``global'' are welcome}{}
 restricts the size of networks that can be analysed,
% model-checking of networks
%for sizes of more than four nodes,
 due to the state-space explosion. 
To extend the size of analysable networks, we employ
\emph{population models},
a counting-abstraction of such networks~\cite{Del03}.
Instead of identifying each oscillator on its own, we record how many
oscillators are in each step of the oscillation cycle. 
This reduces the state-space quite tremendously by exploiting the symmetries in the model~\cite{donaldson2006symmetry},  and we are hence able to extend the 
size of networks.

The notion of population models should not be confused with \emph{population protocols} \cite{Angluin2006}, a formalism
to express distributed algorithms. %The main differences are as follows. 
In contrast to our setting, communication in population protocols is always between two agents, where 
one agent initiates the communication and the other responds. Furthermore, even though the agents cannot
identify the other agents in the network, within the global model each agent is uniquely associated with
a state. In our model, we cannot distinguish between two different agents sharing the same state, even at the
global level. Finally, our oscillators may change their state without interacting with other oscillators,
while the agents in a population protocol must communicate with another agent to change their internal state.

We will present a relation between the concrete models, where each oscillator can
be identified, and corresponding population models, and show that these two models are in a 
\emph{simulation relation} \cite{Milner1971}. More precisely,  the concrete model
\emph{weakly} simulates its abstraction, since
%. This is due to the
%fact that in the concrete model, 
the oscillators have to take transitions independently,
while in the population model, all oscillators evolve in a single step.   

Similarly to typical definitions
of counter abstractions \cite{emerson1999asymmetry,Basler2009}, we use counters to model concurrent entities
that are  indistinguishable for our purposes. For example, to analyse the probability of
eventually reaching a synchronised state, we are not interested in an order of oscillators, which
would be artificial anyway. However, in contrast to these approaches, we do not include means
to introduce new entities into a model. That is, the values within our population models are naturally
bounded by the number of oscillators within the network.

%Furthermore, the concrete models can be thought of as a \emph{refinement} of the abstract
%model, since they contain more structure in the form of an order on the oscillators. 

% \pgcomment{Will not need all of these.}{}
% % TAKEN FROM OLD NOTES 
% \cite{winfree1967biological}
% \cite{kuramoto1987statistical}
% %\cite{bojic2012self} Currently can't see a reason to cite this
% \cite{lucarelli2004decentralized}
% \cite{werner2005firefly}
% \cite{wokoma2002weakly}
% \cite{taniguchi2005distributed}
% \cite{mallada2013synchronization}
% \cite{ernst1995synchronization}
% \cite{ernst1998delay}
% \cite{werner2005firefly}
% \cite{lipa2015firefly}
% \cite{christensen2009fireflies}
% \cite{trianni2009self}
% \cite{wischmann2006synchronization}
% \cite{hartbauer2007novel}
% \cite{smeal2010phase}
% \cite{izhikevich1999weakly}:
% %\cite{anglea2017synchronization} Currently can't see a reason to cite this

% % EXTRAS TAKEN FROM QEST PAPER
% \cite{buck1988synchronous}
% \cite{kuramoto1991collective}.
% \cite{amin2017formal}
% \cite{gainer2016probabilistic}
% \cite{konur2012analysing}
% \cite{bucur2011software}
% ,

% % EXTRAS TAKEN FROM PERCOM PAPER

% \cite{rhee2004techniques}
% \cite{albers2010energy}.
% \cite{yick2008survey}
% \cite{rhee2009clock}.
% \cite{taniguchi2005distributed}
% \cite{bojic2014scalability}
% \cite{baier2010performability}
% \cite{baier2014trade}.
% \cite{gainer2016probabilistic}

%%% Local Variables: 
%%% mode: latex
%%% TeX-master: "sync_journal"
%%% End: 
%\input{pco}
\section{Preliminaries}
\label{sec:preliminaries}

%A \emph{probability space} is a tuple $(\sspace, \events, \measure)$
%where the \emph{sample space} $\sspace$ is a non-empty set, the set of
%events $\events \subseteq \mathbb{P}(\sspace)$ is a \salgebra over $\sspace$,
%and $\measure : \events \to [0, 1]$ is the probability measure such that
%$\measure(\sspace) = 1$ and
%$\measure(\bigcup_{i=1}^{k} \event_i) = \sum_{i=1}^k \measure(\event_i)$
%for any countable collection of pairwise disjoint sets $E_1, \ldots, E_k \in \events$.
%
%\begin{definition}
%Given a probability space $(\sspace, \events, \measure)$ a \emph{random
%variable} is a function $\rvar : \sspace \to \mathbb{R}^{+}$. If
%$V$ takes finitely many values then the expectation
%of $\rvar$ with respect to measure $\measure$ is given by
%\[
%\expect[\rvar] = \int_{\event \in \sspace} \rvar(\event) \ d \measure
%=
%\sum_{\event \in \sspace} \rvar(\event) \measure(\event).
%\]
%\end{definition}
In this section we define \emph{discrete-time Markov chains} (DTMCs),
stochastic processes with discrete state space and discrete time, and
introduce Probabilistic Computation Tree Logic (PCTL), a logic that can be
used to reason about probabilistic reachability and rewards in these
processes.

Throughout this paper, we use the notation
\(\override{f}{x}{y}\), where \(f\) is a function,  to express \emph{updating \(f\) at \(x\) by \(y\)}. 
That is, the function that coincides with \(f\), except for \(x\), where it 
takes the value \(y\).

\subsection{Discrete-Time Markov Chains}

DTMCs can be used to model systems where the
evolution of the system at any moment in time can be represented by
a discrete probabilistic choice over several outcomes.

\begin{definition}
A discrete-time Markov chain $D$ is a tuple $(Q, \init, \prmatrix, L)$
where $Q$ is a finite set of states. $\init$ is the initial state,
and $L : Q \to \mathbb{P}(\labels)$ is a labelling function that assigns 
properties of interest from a set of labels $\labels$ to states.
$\prmatrix : Q \times Q \to [0, 1]$ is the \emph{transition probability matrix}
subject to $\sum_{\state^\prime \in Q} \prmatrix(\state, \state^\prime) = 1$
for all $\state \in Q$,
where $\prmatrix(\state, \state^\prime)$ gives the probability of
transitioning from $\state$ to $\state^\prime$.
We say that there is a transition between two states $\state, \state^\prime \in Q$
if $\prmatrix(\state, \state^\prime) > 0$.
\end{definition}
Intuitively, a DTMC is a state transition system where
transitions between states
are labelled with probabilities greater than $0$ 
% labelled with probability $0$ are omitted from the
%underlying graph, 
and where states are labelled with properties of interest.
An execution \emph{path} $\path$ of a DTMC $D = (Q, \init, \prmatrix, L)$
is a non-empty finite, or infinite, sequence
$\state_0 \state_1 \state_2 \cdots$
where $\state_i \in Q$ and $\prmatrix(\state_i, \state_{i+1}) > 0$ for $i \geqslant 0$.
We denote %the $i^{th}$ state of $\path$ by $\path[i]$, 
the set of all paths starting in state $\state$ by $\paths^D(\state)$, and the
set of all finite paths starting in $\state$ by $\paths_f^D(\state)$. For paths
where the first state along that path is the initial state $\init$ we will
simply use $\paths^D$ and $\paths_f^D$, and we will simply use $\paths$ and
$\paths_f$ if $D$ is
clear from the context. For a finite
path $\path_f \in \paths_f$ the \emph{cylinder set} of $\path_f$
is the set of all infinite paths in $\paths$ that share
prefix $\path_f$.
The probability of taking a finite
path $\state_0 \state_1 \cdots \state_n \in \paths_f$
is given by $\prod_{i=1}^n \prmatrix(\state_{i-1}, \state_i)$.
This measure over finite paths can be extended to a probability measure
$\pmeasure$ over the set of infinite paths $\paths$, where
the smallest $\sigma$-algebra over $\paths$ is the smallest
set containing all cylinder sets for paths in $\paths_f$. For a detailed
description of the construction of the probability measure we refer the
reader to~\cite{kemeny2012denumerable}.

\subsection{Probabilistic Computation Tree Logic}

\emph{Probabilistic Computation Tree Logic~\cite{hansson1994logic}}
(PCTL) is a probabilistic extension of the temporal logic CTL. Properties for
DTMCs can be formulated in PCTL and then checked against the DTMCs using
\emph{model checking}.

\begin{definition}
The syntax of PCTL is given by:
\begin{align*}
	\Phi & = p \mid \lnot \Phi \mid \Phi \land \Phi \mid \pctlp_{\bowtie \lambda}[\Psi] \\
	\Psi & = \pctlnext \ \Phi \mid \Phi \ \pctluntil^{\leqslant k} \ \Phi
\end{align*}
where $p$ is an atomic proposition,
$\bowtie \ \in \{<,\leqslant,\geqslant,>\}$, $\lambda \in [0, 1]$,
and $k \in \mathbb{N} \cup \{\infty\}$.
\end{definition}
Formulas denoted by $\Phi$ are
\emph{state formulas} and formulas denoted by $\Psi$ are \emph{path formulas}.
A PCTL formula is always a state formula, and a path formula can only occur
inside the $\pctlp$ operator. We now give the semantics of PCTL over a DTMC.

\begin{definition}
Given a DTMC $D = (Q, \init, \prmatrix, L)$, we inductively define the
satisfaction relation $\models$ for any state $\state \in Q$ as follows:
\begin{align*}
%\state & \models \ltrue &
%	\qquad \Leftrightarrow \qquad &
%	\text{ for all } \state \\		
\state & \models p &
	\qquad \Leftrightarrow \qquad &
		p \in L(\state) \\
\state & \models \lnot \Phi &
	\qquad \Leftrightarrow \qquad &
		\state \not \models \Phi \\
\state & \models \Phi \land \Phi^\prime &
	\qquad \Leftrightarrow \qquad &
		\state \models \Phi \text{ and } \state \models \Phi^\prime \\
\state & \models \pctlp_{\bowtie \lambda}[\Psi] &
	\qquad \Leftrightarrow \qquad &
		\pmeasure\{\path \in \paths(\state) \mid \path \models \Psi\}
			\bowtie \lambda \\
\intertext{where $v \in \labels$, and for any path
$\path = \init \state_1 \state_2 \cdots$ of  $D$  as follows:}	
\path & \models \pctlnext \ \Phi &
	\qquad \Leftrightarrow \qquad &
		\state_1 \models \Phi \\
\path & \models \Phi \ \pctluntil^{\leqslant k} \ \Phi^\prime &
	\qquad \Leftrightarrow \qquad &
		\exists i \in \nat (i \leqslant k \text{ and } \state_i \models \Phi^\prime
			\text{ and } \forall j < i. \state_j \models \Phi).
\end{align*}
\end{definition}
\noindent{Disjunction}, $\ltrue$, $\lfalse$,  and implication are derived as usual, and we define
eventuality as
$\pctlsometime^{\leqslant k} \ \Phi \equiv \ltrue \ \pctluntil^{\leqslant k} \ \Phi$.
We simply use $\pctlsometime \ \Phi$ and $\Phi\ \pctluntil \ \Phi^\prime$ when $k = \infty$.

\section{Concrete Model of a Network of Pulse-Coupled Oscillators}
\label{sec:concrete}

In this section we give a brief introduction to the formal model of a single pulse-coupled oscillator,
as originally presented in previous work~\cite{gainer2017investigating}. Subsequently,
we encode fully-coupled networks of such oscillators  as discrete
time Markov chains.

\subsection{Pulse-Coupled Oscillator Model}
\label{sec:oscillator_model}
We consider a fully-coupled network of  pulse-coupled oscillators
with identical dynamics over discrete time.
%We denote the set of these oscillators by $\network = \{1, 2, \ldots\}$,
%where each $i \in \network$ corresponds to a single pulse-coupled
%oscillator.
The \emph{phase} of an oscillator $u$ 
%$\network$
at time $t$ is denoted by $\phase_u(t)$.
The phase of an oscillator progresses through a sequence of discrete
integer values bounded by some $T \ge 1$.
%\pgcomment{Is everything still well defined when $T=1$?}{}
The phase progression over time of a single uncoupled oscillator
is determined by the successor function, where the phase
increases over time until it equals $T$, at which point the
oscillator will fire in the next moment in time and the phase
will reset to one. The phase progression of an uncoupled
oscillator is therefore cyclic with period $T$, and we refer to
one cycle as an \emph{oscillation cycle}.

When an oscillator fires, it may happen that its firing is not perceived by
any of the other oscillators coupled to it. We call this a
\emph{broadcast failure} and denote its probability by $\mu\in [0,1]$.
Note that $\mu$ is a global parameter, hence the chance of broadcast
failure is identical for all oscillators.
When an oscillator fires, and a broadcast failure does not occur,
it perturbs the phase of all oscillators to which it is coupled;
we use $\alpha_u(t)$ to denote the number of all other
oscillators %in $\network$
that are coupled to $u$ and will fire
at time $t$.
\begin{definition}
The \emph{phase response function} is a positive increasing function
$\pert : \{1, \ldots, T\} \times \mathbb{N} \times \mathbb{R}^+
\to \mathbb{N}$
%with $\pert(\Phi, \alpha, \epsilon) \ge 0$ for all
%$\Phi, \alpha, \epsilon$,
%$\langle \Phi, \alpha, \epsilon \rangle \in \{1, \ldots, T\} \times \mathbb{N} \times \mathbb{R}^+$,
that maps the phase of an oscillator $u$,
the number of other oscillators perceived to be firing by $u$,
and a real value defining the strength of the coupling between oscillators, to an integer value corresponding to the perturbation to phase
induced by the firing of oscillators where broadcast failures did
not occur. We require \(\pert(\Phase, 0, \epsilon)=0\) for all possible
phase response functions, that is, oscillators are only perturbed if they
perceive at least one firing oscillator.
\label{def:pert}
\end{definition}

We can introduce a \emph{refractory period} into the oscillation cycle
of each oscillator. 
A refractory period is an interval of discrete values
$[1, R] \subseteq [1, T]$ where $ R \le T$ is the size of the
refractory period, such that if $\phase_u(t)$ is inside the interval, for some
oscillator $u$ at time $t$, then $u$ cannot be perturbed by other
oscillators to which it is coupled.
If $R = 0$ then we set $[1, R] = \emptyset$, and there is no refractory period at all.
\begin{definition}
The \emph{refractory function}
$\refr : \{1, \ldots, T\} \times \mathbb{N} \to \mathbb{N}$
is defined as $\refr(\Phase, \delta) = \Phase$ if $\Phase \in [1, R]$,
or $\refr(\Phase,\delta) = \Phase + \delta$ otherwise, and
takes as parameters $\delta$, the degree of perturbance to the phase of
an oscillator, and $\Phase$, the phase,
and returns \(\Phase\) if it is in the refractory period, % defined by $R$,
or $\Phase + \delta$ otherwise.
\label{def:ref}
\end{definition}

The phase evolution of an oscillator $u$ over time is then defined
as follows, where the \emph{update function}
and \emph{firing predicate}, respectively denote the updated
phase of oscillator $u$ at time $t$ in the
next moment in time, and the firing of oscillator $u$ at time $t$,
\begin{align*}
	\update_u(t) & = 1 + \refr(\phase_u(t), \pert(\phi_u(t), \alpha_u(t), \epsilon)), \\
	\fire_u(t) & = \update_u(t) > T, \\
	\phase_u(t+1) & = 	
		\begin{cases}
			1 & \text{if } \fire_u(t) \\
			\update_u(t) & \text{otherwise}.
		\end{cases}
\end{align*}

%\begin{figure}[htb]
% \centering
% \includegraphics[width=.6\columnwidth]{figures/figocode/algorithm-crop.pdf}
% \caption{An abstract PCO algorithm.}
% \label{f:pco}
%\end{figure}

%\subsection{Perceiving flashes of more than one oscillator.}
\subsection{Modelling the Network Using a DTMC}
\label{sec:concrete_absorbtion}
%The behaviour and dynamics of all oscillators are determined by the
%same  global parameters. 
%All oscillators are similar, i.e., internally work according to the same 
%principle and with the same parameters.
%We assume a global \emph{cycle length}
%\(T \in \N\), a refractory period \(0 \leq R \leq T\), and broadcast failure
%probability \(\failureprop \in [0,1]\). 
%
We  model the whole network of oscillators
as a single DTMC \(D = (Q, \concstate_0, \prmatrix, L)\), where each state \(\concstate \in Q\)  denotes a
global state of the network. 
More precisely, the labelling function uniquely  maps each state \(\concstate\) 
 to an combined   encoding of 
 the individual  state of each oscillator. For simplicity, we identify the label of a state
with the state itself, and hence we omit \(L\) from the DTMC, but describe each member
of \(Q\) via its internal state.   
%the different properties of the oscillators (e.g., their current
%phase) are captured within each state of the DTMC. 

We model each
transition of an oscillator as a single transition within the DTMC. However, since 
the oscillators may influence each other within a single time step (that is,
when they are firing), we cannot simply allow for arbitrary sequences of transitions.
For instance,   to model that all the oscillators progress on a similar time-scale,
we need to prevent a single oscillator from taking a transition and thus progressing
its phase without giving the other oscillators a chance to do the same.
We achieve this by the following means: 
\begin{itemize}
\item we divide the internal computation of each oscillator
into two modes: \emph{start} and \emph{update}, and
\item we add a counter to the model, containing the number of oscillators that fire.
\end{itemize}
The counter also possesses both modes, and  resets at the start of each ``round'' of computation. 
First, in the \emph{start} mode, each oscillator checks whether it would fire, 
according to its phase response function and the current number of oscillators that
already fired, as given by the counter. If it does, it
increases the counter and updates its mode to \emph{update}, otherwise it just
updates its mode. If all oscillators are in the update mode, they compute their 
new phases in a single step, according to the phase response function and the current
state of the environment counter.  
%We achieve this by imposing
Furthermore, we impose
 an order on the evaluation 
on the oscillators in the start mode if at least one oscillator fires, starting from the
highest phase to the lowest. This ensures that firing  oscillators are perceived by 
the other nodes, and thus may lead to the firing of the latter. This way
of modelling the nodes implies the assumption that the time window during which each oscillator 
listens on the 
shared medium is long enough to perceive the firing of any other oscillator. 

The general idea of the progress of the network of oscillators is visualised
in Fig.~\ref{fig:concrete_example}. In the figure, each rounded rectangle shows a state of a network
of four oscillators. The circles represents the nodes, where we inscribe its current phase and
an abbreviation of its mode.
%A node that is about to fire is
%drawn red, while nodes within the refractory period are indicated by blue color.
A node that is about to fire is indicated by a starred circle, while a shaded circle
indicates a node that is within the refractory period.
The rectangle denotes the environment counter, with its corresponding value
and mode. The phase response function is arbitrarily chosen, and of minor importance
for the example.

In the first state, all outgoing transitions only check whether
to increase the counter. Since no oscillator is in the firing phase, all oscillators
just update their mode (observe that the single arrow actually denotes four transitions).
In the next step, all oscillators increase their phase by one, and reset their mode 
to \emph{start}. In the next four transitions, oscillator \(2\) fires and increases
the counter, which in turn is sufficient for oscillator \(3\) to fire as well. Hence
they both increased the counter by one, while oscillators \(1\) and \(4\) did not. 
During the last transition of the example, oscillator \(2\) and \(3\) reset their
phase to one, while oscillator \(1\) is perturbed and increases its phase by
two steps at once.  
Oscillator
\(4\)  is within its refractory period, which means that it is not perturbed, and simply 
increments its phase. In addition to these transitions, we also need some bookkeeping
transitions, to ensure that the counter is reset before the oscillators check their phase
response. Furthermore, observe that in the example, it is crucial that oscillator \(3\) checks
its response after oscillator \(2\) increased the counter, since otherwise \(3\) would not have 
been perturbed to fire.   

 Formally, we conflate the states of the oscillators
and the environment into a single state of the DTMC.
Each oscillator %(module, resp.)
can be described by a tuple consisting of the current \emph{phase} \(\Phase\) of the 
oscillator and the \emph{mode} \(\location\) within this phase. The phase ranges
from \(1\) to \(T\), while the mode takes values from \(\{\start, \update\}\). 
Furthermore, we use a single counter to keep track of the number of oscillators that
fired successfully within a single phase computation. 

For a fixed sequence of \(N\) oscillators, a state of the concrete model
consists of a function $\netstate$ that associates a phase and mode 
with each oscillator,
\begin{align*}
  \netstate \colon \{1,\dots,N\} \to (\{1, \dots, T\}\times\{\start, \update\}),
\end{align*}
and the state of the environment $\envstate$ that counts the number of
oscillators that fired,
\begin{align*}
  \envstate \in \{\start, \update\}\times\{0,\dots, N\}.
\end{align*}
A state is therefore a tuple \(\concstate = (\envstate, \netstate)\), where
\(\envstate\) is the state of the environment, and \(\netstate\) is the state of
the network.
We denote the set of all concrete
system states by \(\concstates\).
For simplicity, we use the notation \(\projSing{\phase}\) (\(\projSing{\location}\), respectively) for
the corresponding projection function of the network states, i.e., if \(\netstate(u) = (\Phase_u, \location_u)\), then
%\pgcomment{Shouldn't this be \(\proj{\phase}{\netstate(i)} = \phase_i\) and
%\(\proj{\location}{\netstate(i)} = \location_i\)?}{}
\(\proj{\phase}{\netstate(u)} = \Phase_u\) and \(\proj{\location}{\netstate(u)} = \location_u\).
Similarly, for an environment state \(\envstate = (\location, c)\), we will refer to \(\location\)
by \(\proj{\location}{\envstate} \) and to \(c\) by \(\proj{c}{\envstate}\). 
%Furthermore,
We use the notation
%\(\outref{\concstate}\) for the set \(\{n \mid 1 \leq n \leq N \land \proj{\phase}{\concstate(n)} > R\}\) of all oscillators whose phase is outside the interval
%defined by the refractory period $R$,  
%and
\(\initloc{\concstate}{\Phase} = \{u \mid \proj{\location}{\netstate(u)} = \start \land \proj{\phase}{\netstate(u)} = \Phase\}\) for the set of all oscillators sharing phase
$\Phase$ and mode $\start$ in the state \( \concstate = (\envstate, \netstate)\).
%\pgcomment{This is confusing, should it instead be \(\initloc{\concstate}{} = \{n \mid \proj{\location}{\concstate(n)} = \start\}\) i.e. no subscript $\Phase$?}{}
Furthermore, we simply use the notation \(\initloc{\concstate}{} = \{u \mid \proj{\location}{\netstate(u)} = \start\}\).

\tikzset{
  osc/.style={ circle, thick, draw},
  ref/.style={draw=lightgreyblue,  fill=lightgreyblue!35},
%  fire/.style={draw=lightgreyred, fill=lightgreyred!35},
  fire/.style={draw=lightgreyred},
  firestar/.style={fill=\firenodecolour, star, star points=5, star point ratio=2.2, scale=1.15, opacity=\fireandrefracopacity, rotate=0},  
  counter/.style={draw=black!35, rectangle, thick, rounded corners=0pt}
}

\begin{figure}
  \centering  
  \begin{tikzpicture}[align=center, >=stealth, remember picture]
    \node (state1) { 
      \begin{tikzpicture}[framed, rounded corners=2pt]
        
    \node[osc] (1) at (0,0) {3,s};
    \node[osc] (2) at (2,0) {8,s};  
    \node[osc] (3) at (4,0) {5,s};  
    \node[osc,ref] (4) at (6,0) {1,s};  
    \node[counter] (count) at (8,0) {0,u};  
      \end{tikzpicture}
    };
    \node[above =.25cm of 1] (n1) {1};
    \node[above =.25cm of 2] (n2) {2};
    \node[above =.25cm of 3] (n3) {3};
    \node[above =.25cm of 4] (n4) {4};

    \node[below=of state1] (state2) { 
      \begin{tikzpicture}[framed, rounded corners=2pt]
    \node[osc] (1) at (0,0) {3,u};
    \node[osc] (2) at (2,0) {8,u};  
    \node[osc] (3) at (4,0) {5,u};  
    \node[osc,ref] (4)  at (6,0){1,u};  
    \node[counter] (count) at (8,-.250)  {0,u};  
      \end{tikzpicture}
};
    \node[below=of state2] (state3) { 
      \begin{tikzpicture}[framed, rounded corners=2pt]
    \node[osc] (1) at (0,0) {4,s};
    \node[osc,fire] (2) at (2,0) {9,s}; 
    \node[firestar] (f) at (2,0) {};
    \node[osc] (3) at (4,0) {6,s};  
    \node[osc,ref] (4) at (6,0) {2,s};  
    \node[counter] (count) at (8,-.250) {0,s};  
      \end{tikzpicture}
};

    \node[below=of state3] (state4) { 
      \begin{tikzpicture}[framed, rounded corners=2pt]
    \node[osc] (1) at (0,0) {4,s};
    \node[osc,fire] (2) at (2,0) {9,s}; 
    \node[firestar] (f) at (2,0) {};    
    \node[osc] (3) at (4,0) {6,s};  
    \node[osc,ref] (4) at (6,0) {2,s};  
    \node[counter] (count) at (8,-.250) {0,u};  
      \end{tikzpicture}
};

    \node[below=of state4] (state5) { 
      \begin{tikzpicture}[framed, rounded corners=2pt]
    \node[osc] (1) at (0,0) {4,u};
    \node[osc,fire] (2) at (2,0) {9,u}; 
    \node[firestar] (f) at (2,0) {};    
    \node[osc] (3) at (4,0) {6,u};  
    \node[osc,ref] (4) at (6,0) {2,u};  
    \node[counter] (count) at (8,-.250) {2,u};  
      \end{tikzpicture}
};

    \node[below=of state5] (state6) { 
      \begin{tikzpicture}[framed, rounded corners=2pt]
    \node[osc] (1) at (0,0) {6,s};
    \node[osc,ref] (2) at (2,0) {1,s};  
    \node[osc,ref] (3) at (4,0) {1,s};  
    \node[osc] (4) at (6,0) {3,s};  
    \node[counter] (count) at (8,-.250){2,s};  
      \end{tikzpicture}
};

\node(int1) at ($(state1)!.5!(state2)$) {\dots};
\node[right =.25cm of int1] { check if osc. fire};

\node(int2) at ($(state4)!.5!(state5)$) {\dots};
\node[right =.25cm of int2] { osc. 2 fires, perturbs 3 to fire};

\draw[-] (state1) to  (int1);
\draw[->] (int1) to (state2);
\draw[->] (state2) to node[right] {update phases } (state3);
\draw[->] (state3) to node[right] {reset counter } (state4);
\draw[-] (state4) to  (int2);
\draw[->] (int2) to  (state5);
\draw[->] (state5) to node[right] {update phases  } (state6);

  \end{tikzpicture}
  \caption{Transitions in the Concrete Oscillator Model (\(N=4, T = 9, R = 2\))}
  \label{fig:concrete_example}
\end{figure}
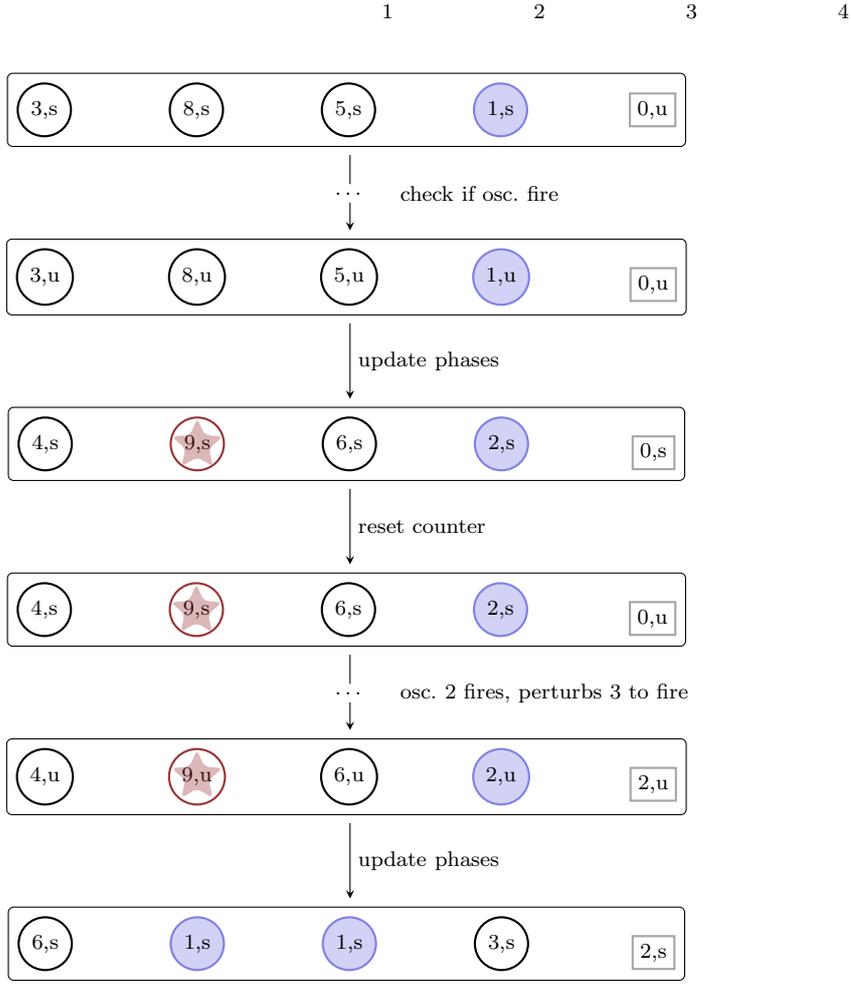

%\pgcomment{Could we perhaps have some intuitive overview of how the concrete model progresses
%from state to state? I don't think it is made to clear here.}{}
%The model is defined as a DTMC, hence we have to define the transition probabilities between states.
We now define the transition probabilities between states. To do this we first distinguish the following cases:
\begin{enumerate}
\item the environment resets its counter; \label{it:env_reset_absorb}
\item no oscillator has a clock value of \(T\); \label{it:no_fire_absorb}
\item an oscillator is in the mode \(\start\), has a clock value lower than \(T\), is perturbed, but not enough to fire; \label{it:in_cycle_no_flash_absorb}
\item an oscillator is in the mode \(\start\), has a clock value lower than \(T\) and is perturbed enough to fire; \label{it:in_cycle_flash_absorb}  
\item an oscillator is in the mode \(\start\), has a clock value of \(T\), and broadcasts its pulse; \label{it:end_cycle_bc_absorb}
\item an oscillator is in the mode \(\start\), has a clock value of \(T\), and fails to broadcast its pulse; \label{it:end_cycle_bc_fail_absorb}
\item all oscillators are in the mode \(\update\), update their clock and reset their state to \(\start\). \label{it:update_absorb}
\end{enumerate}
We will impose an order on certain transitions for two reasons. Firstly, we will restrict
transitions that are only used for bookkeeping purposes. For example, we will require that the reset transition
of the environment is taken before 
%\pgcomment{Not too sure what this sentence is saying, is something missing?}{
any of the transitions for the oscillators within a phase 
are activated. In particular, this means that each computation starts with a transition of the type \ref{it:env_reset_absorb}.
 Secondly, we need
to ensure that, if at least one oscillator fires, the phase response of all oscillators is evaluated starting with
oscillators in the highest phase, down to the lowest phase, as described above.
The cases stated above are reflected in the following definitions for the transition probability between 
two states \(\concstate = (\envstate, \netstate)\) and \(\concstate^\prime = (\envstate^\prime, \netstate^\prime)\).

Case~\ref{it:env_reset_absorb}, where the environment resetting its counter is treated as follows. In the
precondition, we require that the mode of the counter is \emph{start}, and the state of the oscillators
 does not change from \(\concstate\) to \(\concstate^\prime\). Furthermore, the mode of the counter 
changes to \emph{update} in \(\concstate^\prime\), and its value is set to \(0\). Since this 
transition is mandatory at the beginning of each round, its probability is \(1\). 
\begin{align}
\text{If } &  \proj{\location}{\envstate} = \start \land \proj{\location}{\envstate^\prime} = \update \land \proj{c}{\envstate^\prime} = 0 \land  \forall u \colon  \netstate(u) = \netstate^\prime(u),   \label{trans:env_reset}\\
\text{then }& \prmatrix(\concstate, \concstate^\prime) = 1 \nonumber.
\end{align}

Now we turn to the cases \ref{it:end_cycle_bc_absorb} and \ref{it:end_cycle_bc_fail_absorb} where some oscillator is at the end of its cycle. The preconditions of both cases are similar: the counter is required to be in the
\emph{update} mode, and there is an oscillator \(w\), whose phase is \(T\) and mode is \emph{start}. Furthermore,
in \(\concstate^\prime\), the mode of \(w\) is \emph{update}, and the state of all other oscillators does not
change. The difference between the cases is whether the counter is increased, that is, whether the 
oscillator manages to broadcast its signal. The probability of succeeding is \(\frac{1 - \failureprob}{|\initloc{\concstate}{T}|}\), since there may be more than one oscillator in phase \(T\) at state \(\concstate\). Hence
we have to normalise the tranistion probability accordingly. Similarly, the probability of failing
to fire is \(\frac{ \failureprob}{|\initloc{\concstate}{T}|}\).
\begin{align}
\text{If }&\proj{\location}{\envstate} = \update \text{ and there is a } w \text{ s.t.} \label{trans:fire_add}\\
& \proj{\location}{\netstate(w)} = \start \land \proj{\phase}{\netstate(w)} = T \land   \proj{\location}{\netstate^\prime(w)} = \update \nonumber\\
  & \land \proj{\phase}{\netstate(w)} = \proj{\phase}{\netstate^\prime(w)} \land \forall u \colon u \neq w \implies  \netstate(u) = \netstate^\prime(u)     \nonumber \\
& \land \proj{c}{\envstate^\prime} = \proj{c}{\envstate}+1  \nonumber\\
\text{then }& \prmatrix(\concstate, \concstate^\prime) = \frac{1 - \failureprob}{|\initloc{\concstate}{T}|}\nonumber.\\
\text{If }&\proj{\location}{\envstate} = \update \text{ and there is a } w \text{ s.t.} \label{trans:fire_failure_add} \\
& \proj{\location}{\netstate(w)} = \start \land \proj{\phase}{\netstate(w)} = T \land   \proj{\location}{\netstate^\prime(w)} = \update \nonumber\\
  & \land \proj{\phase}{\netstate(w)} = \proj{\phase}{\netstate^\prime(w)} \land \forall u \colon u \neq w \implies  \netstate(u) = \netstate^\prime(u)     \nonumber \\
& \land \proj{c}{\envstate^\prime} = \proj{c}{\envstate}  \nonumber\\
\text{then }& \prmatrix(\concstate, \concstate^\prime) = \frac{\failureprob}{|\initloc{\concstate}{T}|}\nonumber. 
\end{align}

If no oscillator is at the end of its cycle, that is, in case \ref{it:no_fire_absorb}, we define the probability of one oscillator updating
its mode as follows. Observe that we have to normalise the transition probability by the number of
all oscillators that have not transitioned to their update mode yet. This is correct, since 
no oscillator fires, which also means that no oscillator can be activated beyond the maximum phase. This 
implies in particular that the order of oscillator transitions does not matter in this round.
\begin{align}
\text{If }&\proj{\location}{\envstate} = \update \text{ and there is a } w \text{ s.t.} \label{trans:nothing_fires_add}\\
& \proj{\location}{\netstate(w)} = \start \land  \proj{\location}{\netstate^\prime(w)} = \update \land \proj{\phase}{\netstate(w)} = \proj{\phase}{\netstate^\prime(w)} \nonumber\\
& \land \forall u \colon \proj{\phase}{\netstate(u)} < T  \land \forall u \colon u \neq w \implies  \netstate(u) = \netstate^\prime(u)  \land \envstate = \envstate^\prime   \nonumber \\
\text{then }& \prmatrix(\concstate, \concstate^\prime) = \frac{1}{|\initloc{\concstate}{}|}\nonumber.
\end{align}

Now we will consider the cases \ref{it:in_cycle_no_flash_absorb} and \ref{it:in_cycle_flash_absorb}, 
where some oscillator already fired (i.e., \(\proj{c}{\envstate} > 0\)), and other oscillators
are perturbed. We distinguish between two cases: either an oscillator is sufficiently perturbed to also fire
%\pgcomment{I don't think it makes sense to talk about absorption here, since all oscillators
%are considered individually. Also, any chance we can kill some of the 'i.e.s'? ;) }{}
%(i.e., it is
%\emph{absorbed} into firing), 
or the perturbation does not cause the phase to exceed the firing threshold. One complication arises in these cases:
we have to ensure that we only allow the oscillators to update their
mode once all oscillators with a higher phase have been considered. Since the perturbation
function is increasing, a higher phase may result in a higher perturbation. That is, 
oscillators with a higher phase need to be perturbed by fewer firing oscillators before their
phase is increased beyond the threshold and they in turn fire. Hence, if we
did not enforce such an order, oscillators with a lower phases might not be perturbed when oscillators
with a higher phase fire.
Again, observe that we normalise the transition probabilities according to the number of 
oscillators satisfying similar conditions. That is, this time we need to normalise 
on the number of oscillators with the same phase in the \emph{start} mode.

\begin{align}
\text{If }&\proj{\location}{\envstate} = \update \text{ and there is a } w \text{ s.t.} \label{trans:perturbed_but_not_fire_add}\\
& \proj{\location}{\netstate(w)} = \start \land  \proj{\location}{\netstate^\prime(w)} = \update \land \proj{\phase}{\netstate(w)} = \proj{\phase}{\netstate^\prime(w)} \nonumber\\
& \land \proj{\phase}{\netstate(w)} < T \land \exists u \colon \proj{\phase}{\netstate(u)} = T   \nonumber\\
&\land \forall u \colon u \neq w \implies (\proj{\location}{\netstate(u)} = \update \lor \proj{\phase}{\netstate(u)} \leq \proj{\phase}{\netstate(w)})\nonumber\\
& \land \proj{\phase}{\netstate(w)} + \pert( \proj{\phase}{\netstate(w)},  \proj{c}{\envstate},\epsilon) + 1 \leq T \nonumber\\
& \land \forall u \colon u \neq w \implies \netstate(u) = \netstate^\prime(u)\nonumber\\
& \land \envstate^\prime = \envstate \nonumber \\
\text{then }& \prmatrix(\concstate, \concstate^\prime) = \frac{1}{|\initloc{\concstate}{\proj{\phase}{\concstate(w)}}|}\nonumber.
\end{align}

The cases where a perturbed oscillator fires are analogous to oscillators with a maximal
phase, except for the addititional conditions that some other oscillator fired, and that
all oscillators with higher phases have already been considered.
\begin{align}
\text{If }&\proj{\location}{\envstate} = \update \text{ and there is a } w \text{ s.t.} \label{trans:perturbed_fire_failure_add}\\
& \proj{\location}{\netstate(w)} = \start \land  \proj{\location}{\netstate^\prime(w)} = \update \land \proj{\phase}{\netstate(w)} = \proj{\phase}{\netstate^\prime(w)} \nonumber\\
& \land \proj{\phase}{\netstate(w)} < T \land \exists u \colon \proj{\phase}{\netstate(u)} = T   \nonumber\\
&\land \forall u \colon u \neq w \implies (\proj{\location}{\netstate(u)} = \update \lor \proj{\phase}{\netstate(u)} \leq \proj{\phase}{\netstate(w)})\nonumber\\
& \land \proj{\phase}{\netstate(w)} + \pert( \proj{\phase}{\netstate(w)},  \proj{c}{\envstate},\epsilon) + 1 > T \nonumber\\
& \land \forall u \colon u \neq w \implies \netstate(u) = \netstate^\prime(u)\nonumber\\
& \land \envstate^\prime = \envstate \nonumber \\
\text{then }& \prmatrix(\concstate, \concstate^\prime) = \frac{\failureprob}{|\initloc{\concstate}{\proj{\phase}{\concstate(w)}}|}\nonumber 
\end{align}

\begin{align}
\text{If }&\proj{\location}{\envstate} = \update \text{ and there is a } w \text{ s.t.} \label{trans:perturbed_fire_add}\\
& \proj{\location}{\netstate(w)} = \start \land  \proj{\location}{\netstate^\prime(w)} = \update \land \proj{\phase}{\netstate(w)} = \proj{\phase}{\netstate^\prime(w)} \nonumber\\
& \land \proj{\phase}{\netstate(w)} < T \land \exists u \colon \proj{\phase}{\netstate(u)} = T   \nonumber\\
&\land \forall u \colon u \neq w \implies (\proj{\location}{\netstate(u)} = \update \lor \proj{\phase}{\netstate(u)} \leq \proj{\phase}{\netstate(w)})\nonumber\\
& \land \proj{\phase}{\netstate(w)} + \pert( \proj{\phase}{\netstate(w)},  \proj{c}{\envstate},\epsilon) + 1 > T \nonumber\\
& \land \forall u \colon u \neq w \implies \netstate(u) = \netstate^\prime(u)\nonumber\\
          & \land \proj{c}{\envstate^\prime} = \proj{c}{\envstate} +1 \nonumber \\
          & \land \proj{\location}{\envstate^\prime} = \proj{\location}{\envstate}  \nonumber \\  
\text{then }& \prmatrix(\concstate, \concstate^\prime) = \frac{1 - \failureprob}{|\initloc{\concstate}{\proj{\phase}{\concstate(w)}}|}\nonumber.
\end{align}

The final case \ref{it:update_absorb}, where all oscillators update their clock values simultaneously, is given by the following equation. It requires that all oscillators have finished their computation, whether they fire,
and both the counter and the oscillators will reset their mode to \emph{start} after the transition.
\begin{subequations}\label{trans:update_add}
\begin{align}
\text{If }& \proj{\location}{\envstate} = \update \text{ and } \proj{\location}{\envstate^\prime} = \start \text{ and }\tag{\ref{trans:update_add}}\\
& \text{for all } u \text{ we have } \proj{\location}{\netstate(u)} = \update \land  \proj{\location}{\netstate^\prime(u)} = \start  \land F_\update  \nonumber\\
\text{then } &\prmatrix(\concstate,\concstate^\prime) = 1 \nonumber.
\end{align}
The formula \(F_\update\) is an abbreviation for the conjunction of the following four conditions, which
model the update of the phases of the oscillators, according to the phase response function. Observe
that the phases of the oscillators had not been updated by the previously defined transitions. Hence, we
now update the phases of all oscillators at once.
\begin{align}
  &\forall u \colon \proj{\phase}{\netstate(u)} = T  \implies \label{trans:update_end_cycle_add}\\ 
& \hspace{2cm}\proj{\phase}{\netstate^\prime(u)} =1  \nonumber  \\
 & \forall u \colon \proj{\phase}{\netstate(u)} < T \land \proj{\phase}{\netstate(u)} \leq R \implies \label{trans:update_in_cycle_no_pulse_add} \\
& \hspace{2cm} \proj{\phase}{\netstate^\prime(u)} = \proj{\phase}{\netstate(u)} + 1  \nonumber\\
 & \forall u \colon \proj{\phase}{\netstate(u)} < T \land \proj{\phase}{\netstate(u)} > R \;\land      \label{trans:update_in_cycle_pulse_stays_in_cycle_add}\\
 & \hspace{1cm}\proj{\phase}{\netstate(u)} + \pert(\proj{\phase}{\netstate(u)},  \proj{c}{\envstate},\epsilon)  +1 \leq T \implies \nonumber\\
 &\hspace{2cm}\proj{\phase}{\netstate^\prime(u)} = \proj{\phase}{\netstate(u)} + \pert(\proj{\phase}{\netstate(u)},  \proj{c}{\envstate},\epsilon)  +1   \nonumber\\
 & \forall u \colon \proj{\phase}{\netstate(u)} < T \land \proj{\phase}{\netstate(u)} > R \;\land  \label{trans:update_in_cycle_pulse_exceeds_cycle_add}\\
 & \hspace{1cm}\proj{\phase}{\netstate(u)} + \pert(\proj{\phase}{\netstate(u)}, \proj{c}{\envstate}, \epsilon)  +1 > T \implies \nonumber\\ 
 & \hspace{2cm}\proj{\phase}{\netstate^\prime(u)} = 1  \nonumber.
\end{align}
\end{subequations}

In this formula, (\ref{trans:update_end_cycle_add}) handles the simple case of firing oscillators,   while (\ref{trans:update_in_cycle_no_pulse_add}) 
defines the behaviour of oscillators within their refractory period. The formulas (\ref{trans:update_in_cycle_pulse_stays_in_cycle_add}) and
(\ref{trans:update_in_cycle_pulse_exceeds_cycle_add}) reflect the two cases where oscillators are perturbed, either not exceeding their
oscillation cycle, or firing, respectively.  

With this model, we could begin to analyse the synchronisation behaviour with respect
to different phase response functions or broadcast failure probabilities. However, 
the state space of the model increases exponentially with the number of oscillators,
which makes an analysis beyond small numbers of infeasible. To overcome
this restriction, we increase the level of abstraction as presented in the next section.
%%% Local Variables: 
%%% mode: latex
%%% TeX-master: "sync_journal"
%%% End: 
\section{Population Model}
\label{sec:population}

In this section, we define a \emph{population model}
of a network of pulse-coupled oscillators for parameters as defined in Sect.~\ref{sec:oscillator_model} as $\psystemfull$.
Oscillators in our model have
identical dynamics, and two oscillators are indistinguishable if they
share the same phase. That is, we can reason about groups of
oscillators, instead of individuals.
We therefore encode the global state of the model
as a tuple $\langle k_1, \ldots, k_T \rangle$
where each $k_\Phase$ is the number of oscillators sharing a
phase value of $\Phase$.
The population model does not account for the introduction of additional
oscillators to a network, or the loss of existing
coupled oscillators. That is, the population $N$ remains constant. 
\begin{definition}
A global state of a population model $\psystemfull$ is a $T$-tuple
$\state \in \{0, \ldots, N\}^T$, where
$\state = \langle k_1, \ldots, k_T\rangle$
and $\sum_{\Phase = 1}^T k_\Phase = N$.
%We denote 
The set of all global states of $\psystem$ is
$\states(\psystem)$, or simply \(\states\) when \(\psystem\) is clear from the context.
\label{def:globalstates}
\end{definition}

\begin{example}
\label{ex:globalstates}
%\pgcomment{Need an actual instantiation of a population model here
%to justify the calculations in examples~\ref{ex:successor} and~\ref{ex:prob}.}{}
Figure~\ref{fig:example1} shows four global states for an instantiated
population model of $N = 8$ oscillators with $T = 10$ discrete values for
their phase and a refractory period of length $R = 2$. We assume that
the phase response function is linear, that is, \(\pert(\Phase, \alpha, \epsilon) = [\Phase \cdot \alpha \cdot \epsilon]\),
where \([\cdot]\) denotes rounding to the closest integer. Furthermore, let \(\epsilon = 0.115\). 
For example \mbox{$\state_0 = \langle 0, 0, 2, 1, 0, 0, 5, 0, 0, 0 \rangle$}
is the global state where two oscillators have a phase of three,
one oscillator has a phase of four, and five oscillators have a phase of seven.
The starred node indicates the number of oscillators with phase 
ten that will fire in the next moment in time, while the shaded nodes
indicate oscillators with phases that lie within the refractory period
(one and two).
If no oscillators have some phase $\Phase$ then we omit the $0$ in the
corresponding node.
Observe that, while going from \(\state_{i-1}\) to \(\state_{i}\) (\(1 \le i \le 3\)), the oscillator phases increase by one.
%Here the states illustrate the standalone evolution of oscillator phases,
%as each state $\state_i$ is the state where the phases all oscillators in state
%$\state_{i-1}$ increased by one, for $1 \le i \le 3$.
In the next section, we will explain how transitions between these
global states are made.
Note that directional arrows indicate cyclic direction, and do not represent transitions.
\end{example}
\begin{figure}[tb]
\begin{center}
\begin{tabular}{@{}c@{}c@{}c@{}c@{}}
\begin{tikzpicture}
	\runningexampleparams
	\def \statel {\state_0}	
	\dnac{1}{0};\dnac{2}{0};\dnac{3}{2};\dnac{4}{1};\dnac{5}{0};
	\dnac{6}{0};\dnac{7}{5};\dnac{8}{0};\dnac{9}{0};\dnac{10}{0};
\end{tikzpicture} &
\begin{tikzpicture}
	\runningexampleparams
	\def \statel {\state_1}	
	\dnac{1}{0};\dnac{2}{0};\dnac{3}{0};\dnac{4}{2};\dnac{5}{1};
	\dnac{6}{0};\dnac{7}{0};\dnac{8}{5};\dnac{9}{0};\dnac{10}{0};
	\dta{7}{1}{0.8cm}{0.8};	
	\dta{4}{1}{0.6cm}{0.6};	
	\dta{3}{1}{0.8cm}{0.8};	
\end{tikzpicture} &
\begin{tikzpicture}
	\runningexampleparams
	\def \statel {\state_2}	
	\dnac{1}{0};\dnac{2}{0};\dnac{3}{0};\dnac{4}{0};\dnac{5}{2};
	\dnac{6}{1};\dnac{7}{0};\dnac{8}{0};\dnac{9}{5};\dnac{10}{0};
	\dta{8}{1}{0.8cm}{0.8};	
	\dta{5}{1}{0.6cm}{0.6};	
	\dta{4}{1}{0.8cm}{0.8};	
\end{tikzpicture} &
\begin{tikzpicture}
	\runningexampleparams
	\def \statel {\state_3}	
	\dnac{1}{0};\dnac{2}{0};\dnac{3}{0};\dnac{4}{0};\dnac{5}{0};
	\dnac{6}{2};\dnac{7}{1};\dnac{8}{0};\dnac{9}{0};\dnac{10}{5};
	\dta{9}{1}{0.8cm}{0.8};	
	\dta{6}{1}{0.6cm}{0.6};	
	\dta{5}{1}{0.8cm}{0.8};	
\end{tikzpicture} \\
\end{tabular}
\caption{Evolution of the global state over four discrete time steps.}
\label{fig:example1}
\end{center}
\end{figure}
With every state $\state \in \states$ we associate a non-empty set
of \emph{failure vectors}, where each failure vector is a tuple
of broadcast failures that could occur in $\state$.
\begin{definition}
A \emph{failure vector} is a $T$-tuple $\fvec = \langle f_1, \dots, f_T\rangle \in (\{0, \ldots, N\} \cup \{\star\})^T$, where 
 \(f_i = \star\) implies \(f_j =\star\) for all \(1 \leq j \leq i\).
We denote the set of all possible failure vectors by $\failvecs$.
\label{def:failvec}
\end{definition}
Given a failure vector $\fvec = \langle f_1, \ldots, f_T \rangle$,
$f_\Phase \in \{0, \ldots, N\}$
indicates
the number of broadcast failures that occur for all
oscillators with a phase of $\Phase$.
If $f_\Phase = \star$ then
no oscillators with a phase of $\Phase$ fire.
%for all $1 \le \Phase \le T$.
Semantically,
$f_\Phase = 0$ and $f_\Phase = \star$ differ in that
the former indicates that all (if any) oscillators with
phase $\Phi$ fire and no broadcast failures occur,
while the latter indicates that all (if any) oscillators
with a phase of $\Phase$ do not fire.
If no oscillators fire at all in a global state then we
have only one possible failure vector, namely $\{\star\}^T$.

\subsection{Transitions}
In Section~\ref{subsec:fveccalc} we will describe how we can calculate
the set of all possible failure vectors for a
global state, and thereby identify all of its successor states.
However we must first show how we can calculate the single
successor state of a global state $\state$,
given some failure vector $\fvec$.

\paragraph{Absorptions.}
For real deployments of synchronisation protocols it is often the case
that the duration of a single oscillation cycle will be at least
several seconds \cite{christensen2009fireflies,perez2016emergence}.
The perturbation induced by the firing of a group of oscillators
may lead to groups of other oscillators to which they are coupled firing
in turn. The firing of these other oscillators may then
cause further oscillators to fire, and so forth, leading
to a ``chain reaction'', where each group of oscillators triggered
to fire is \emph{absorbed} by the initial group of firing oscillators.
Since the whole chain reaction of absorptions may occur within just a
few milliseconds,
and in our model the oscillation cycle is a sequence of
discrete states, when a chain reaction occurs the phases of all
perturbed oscillators are updated at one single time step.

Since we are considering a fully connected network of oscillators,
two oscillators sharing the same phase 
will have their phase updated to the same value in the next
time step. They will always perceive the same number of other
oscillators firing.
Therefore, for each phase $\Phase$ we define the function
$\alpha^\Phase\colon \states \times \failvecs\to \{0,\dots,N\}$, 
where $\alpha^\Phase(\state, \fvec)$ is the number of 
oscillators with a phase greater than $\Phase$ perceived to be firing by
oscillators with phase $\Phase$, in some global state, incorporating
the broadcast failures defined in the failure vector $\fvec$.
This allows us to encode the aforementioned chain reactions of
firing oscillators. Note that our encoding of chain reactions results in a
global semantics that differs from typical parallelisation operations, 
for example, the construction of the cross product of the individual oscillators.
Observe that, in the concrete model of Sect.~\ref{sec:concrete_absorbtion}, we modelled
such a behaviour by case \ref{it:in_cycle_flash_absorb}.

Given a global state $\state = \langle k_1, \ldots, k_T \rangle$
and a failure vector $\fvec = \langle f_1, \ldots, f_T \rangle$,
the following mutually recursive definitions show how we
calculate the values $\alpha^1(\state, \fvec),\ldots,\alpha^T(\state, \fvec)$, and how 
functions introduced in Sect.~\ref{sec:oscillator_model}
are modified
to indicate the update in phase, and firing, of \emph{all} oscillators
sharing the same phase $\Phase$.
Observe that to calculate
any $\alpha^\Phase(\state, \fvec)$ we only refer to definitions for phases greater
than $\Phase$ and the base case is $\Phase = T$, that is, values are computed
from $T$ down to $1$. The function \(\refr\) is the refractory function as defined in Sect.~\ref{sec:oscillator_model}.
%\pgcomment{Reference to $\refr$ in definition of $\update$ here, this needs to be defined somewhere.}{}
\begin{align}
	\update^\Phase(\state, \fvec) & = 1 + \refr(\Phase, \pert(\Phase, \alpha^\Phase(\state, \fvec), \epsilon)) \\
	\fire^\Phase(\state, \fvec) & = \update^\Phase(\state, \fvec) > T \\
	\alpha^\Phase(\state, \fvec) & =
	\begin{cases}
		0 & \text{if } \Phase{=}T\\
		\alpha^{\Phase{+}1}(\state{,}\fvec){+}k_{\Phase{+}1}{-}f_{\Phase{+}1}		
		& \text{if } \Phase{<}T, f_{\Phase{+}1}{\ne}\star \text{ and } \fire^{\Phase{+}1}(\state{,}\fvec) \\
		\alpha^{\Phase + 1}(\state, \fvec) & \text{otherwise}
	\end{cases} \label{eq:alpha}
\end{align}
%\pgcomment{Some failure vectors are still admissable here that will be disallowed later $\langle 1, \star, 2, \star \rangle$, for instance.}{}

\paragraph{Transition Function.}
We now define the transition function that maps phase values to their updated
values in the next time step.
Note that since we no longer 
distinguish different oscillators with the same phase we
only need to calculate a single value for their evolution and perturbation.
\begin{definition}
The \emph{phase transition function}
$\ptf : \states \times \{1, \ldots, T\} \times \failvecs \to \nat$
maps a global state $\state$, a phase $\Phase$, and some possible failure vector
$\fvec$ for $\state$,
to the updated phase  in the next discrete time step,
with respect to the broadcast
failures defined in $\fvec$, and is defined as
\begin{align}
	\ptf(\state, \Phase, \fvec) =
	\begin{cases}
		1 & \text{if } \fire^\Phase(\state, \fvec) \\
		\update^\Phase(\state, \fvec) & \text{otherwise}. \\
	\end{cases}
\end{align}
\label{def:taufailure}
\end{definition}

%\begin{lemma}
%The range of the function $\ptf$ is bound by $T$. That is, for any
%$\state$, for any possible failure vector $\fvec$ for $\state$,  and for all
%$\Phase \in \{1, \ldots, T\}$, we have that $1 \le \ptf(\state, \Phase, \fvec) \le T$.
%\label{lem:taubound}
%\end{lemma}
%\begin{proof}
%By construction.
%\end{proof}
Let $\phiupdate_\Phase(\state, \fvec)$ be the set of phase values $\PhaseTwo$ where all
oscillators with phase $\PhaseTwo$ in $\state$ will have their phase updated to
$\Phase$ in the next time step,
with respect to the broadcast failures defined in $\fvec$. Formally, 
\begin{align}
	\phiupdate_\Phase(\state, \fvec) = \{\PhaseTwo \mid \PhaseTwo \in \{1,\ldots,T\} \land
		\ptf(\state, \PhaseTwo, \fvec) = \Phase\}.
\end{align}
We can now calculate the successor state of a global state $\state$
and define how the model evolves over time.
\begin{definition}
The \emph{successor function}
$\suc : \states \times \failvecs \to \states$ maps
a global state $\state$ and a failure vector $\fvec$ to a  state $\state^\prime$, and is defined as
$\suc(\langle k_1, \ldots, k_T \rangle, \fvec) = \langle k_1^\prime, \ldots, k_T^\prime \rangle$, 
where $k_\Phase^\prime{=}\sum_{\PhaseTwo \in \phiupdate_\Phase(\state, \fvec)} k_\PhaseTwo$ for $1 \le \Phase \le T$.
\label{def:fsucc}
\end{definition}

\begin{example}
  \label{ex:successor}
  Recall that the perturbation function of our example was
  given as \(\pert(\Phase, \alpha, \epsilon) = [\Phase \cdot \alpha \cdot \epsilon]\),
  where \([\cdot]\) denotes rounding and \(\epsilon = 0.115\).
Consider the global state $\state_2$ of Fig~\ref{fig:example2} where
no oscillators will fire since $k_{10} = 0$. We therefore have one
possible failure vector for $\state_0$, namely $\fvec = \{\star\}^{10}$.
Since no oscillators fire the dynamics of the oscillators are
determined solely by their standalone evolution, and all oscillators
simply increase their phase by $1$ in the next time step.
Now consider the global state $\state_3$ and 
$\fvec = \langle \star, \star, \star, \star, \star, \star, 1, 0, 0, 0 \rangle$,
a possible failure vector for $\state_3$, indicating that oscillators
with phases of $7$ to $10$ will fire and one broadcast failure
will occur for the single oscillator that will fire with phase $7$.
Here a chain reaction occurs as the perturbation induced by the
firing of the $5$ oscillators causes the single oscillator with a 
phase of $7$ to also fire. A broadcast failure occurs when this single
oscillator fires, and the perturbation of the $5$ firing oscillators
is insufficient to cause the $2$ oscillators with a phase of $6$ to
also fire. In the next state the oscillator with phase $7$ has been
absorbed by the group of the $5$ oscillators that had phase $10$.

More explicitly, 
since $\fire^{10}(\state_3, \fvec)$ holds we have that
$\alpha^{9}(\state_3, \fvec)$ = $\alpha^{10}(\state_3, \fvec) + k_{10} - f_{10} = 0 + 5 - 0 = 5$.
Now, since \(\pert(9, 5, 0.14) = [9 \cdot 5 \cdot 0.115] = [5.175] = 5\), we have \(\update^9(\state_3, F) = 15 > 10\), and thus, $\fire^{9}$ holds. Hence, we have that
$\alpha^{8}(\state_3, \fvec) = \alpha^{9}(\state_3, \fvec) + k_{9} - f_{9} = 0 + 5 - 0 = 5$, and similarly, due to \(\pert(8, 5, 0.115) =5\), $\fire^{8}$ holds. That is, we have that
$\alpha^{7}(\state_3, \fvec) = \alpha^{8}(\state_3, \fvec) + k_{8} - f_{8} = 0 + 5 - 0 = 5$. We then continue calculating $\alpha^{\Phase}(\state_3, \fvec)$ for
$6 \ge \Phase \ge 1$, and noting that \(\pert(6,5,0.115) = [3.45] = 3\). Hence   $\fire^{6}(\state_3, \fvec)$
does not hold, and we  obtain
$\alpha^{1}(\state_3, \fvec) = \alpha^{2}(\state_3, \fvec) =
\alpha^{3}(\state_3, \fvec) = \alpha^{4}(\state_3, \fvec) =
\alpha^{5}(\state_3, \fvec) = \alpha^{6}(\state_3, \fvec) =
\alpha^{7}(\state_3, \fvec) = 5$. We conclude that
$\phiupdate_{1}(\state_3, \fvec) = \{10, 9, 8, 7\}$,
$\phiupdate_{10}(\state_3, \fvec) = \{6, 5\}$,
$\phiupdate_{9}(\state_3, \fvec) = \{4, 3\}$,
and $\phiupdate_{\Phase}(\state_3, \fvec) = \emptyset$ for $9 > \Phase > 3$.
Since $R = 2$ we have that
$\phiupdate_3(\state_3, \fvec) = \{2\}$ and $\phiupdate_2(\state_3, \fvec) = \{1\}$.
We calculate the successor of $\state_3$ as
$\state_4 = \suc( \langle \star, \star, \star, \star, \star, \star, 1, 0, 0, 0 \rangle, \fvec) =
\langle k_{10} + k_9 + k_8 + k_7, k_1, k_2, 0, 0, 0, 0, 0, k_4 + k_3, k_6 + k_5 \rangle = \langle 6, 0, 0, 0, 0, 0, 0, 0, 0, 2 \rangle$.
\end{example}

%\pgcomment{Need a better caption for Figure~\ref{fig:example2}}{}
\begin{figure}[bt]
\begin{center}
\begin{tabular}{@{}c@{}c@{}c@{}c@{}}
\begin{tikzpicture}
	\runningexampleparams
	\def \statel {\state_2}
	\dnac{1}{0};\dnac{2}{0};\dnac{3}{0};\dnac{4}{0};\dnac{5}{2};
	\dnac{6}{1};\dnac{7}{0};\dnac{8}{0};\dnac{9}{5};\dnac{10}{0};
\end{tikzpicture} &
\begin{tikzpicture}
	\runningexampleparams
	\def \statel {\state_3}	
	\dnac{1}{0};\dnac{2}{0};\dnac{3}{0};\dnac{4}{0};\dnac{5}{0};
	\dnac{6}{2};\dnac{7}{1};\dnac{8}{0};\dnac{9}{0};\dnac{10}{5};
	\dta{5}{1}{0.8cm}{1};	
	\dta{6}{1}{0.6cm}{0.8};	
	\dta{9}{1}{0.8cm}{0.6};	
\end{tikzpicture} &
\begin{tikzpicture}
	\runningexampleparams
	\def \statel {\state_4}	
	\dnac{1}{6};\dnac{2}{0};\dnac{3}{0};\dnac{4}{0};\dnac{5}{0};
	\dnac{6}{0};\dnac{7}{0};\dnac{8}{0};\dnac{9}{0};\dnac{10}{2};
	\dta{6}{4}{0.8cm}{1};
	\dta{7}{4}{0.6cm}{0.8};
	\dta{10}{1}{0.4cm}{0.6};
\end{tikzpicture} &
\begin{tikzpicture}
	\runningexampleparams
	\def \statel {\state_5}	
	\dnac{1}{2};\dnac{2}{6};\dnac{3}{0};\dnac{4}{0};\dnac{5}{0};
	\dnac{6}{0};\dnac{7}{0};\dnac{8}{0};\dnac{9}{0};\dnac{10}{0};
	\dta{10}{1}{0.8cm}{0.8};
	\dta{1}{1}{0.6cm}{0.6};
\end{tikzpicture} \\
\end{tabular}
\caption{Evolution of the global state over four discrete time steps.}
\label{fig:example2}
\end{center}
\end{figure}
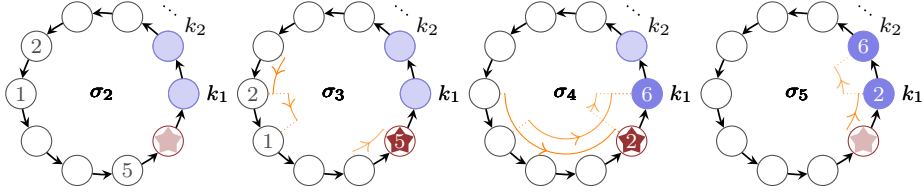

\begin{lemma}
The number of oscillators is invariant during transitions, i.e., the successor
function only creates tuples that are states of the given model. Formally,
let \(\state= \langle k_1, \ldots, k_T \rangle\) and \(\state^\prime = 
\langle k_1^\prime, \ldots, k_T^\prime \rangle\) be two states
of a model \(\psystem\) such that \(\state^\prime = \suc(\state, \fvec)\), where $\fvec$
is some possible failure vector for \(\state\). Then
$\sum_{\Phase = 1}^T k_\Phase = \sum_{\Phase = 1}^T k_\Phase^\prime = N.$
\end{lemma}
\begin{proof}
Observe that the range of the function $\ptf$ is bound by $T$. By construction
we can see that for any $\state$, for any possible failure vector $\fvec$ for
$\state$,  and for all $\Phase \in \{1, \ldots, T\}$, we have that
$1 \le \ptf(\state, \Phase, \fvec) \le T$.
%Observe that for all \(1 \leq \Phase \leq T\), we have that
%\(\ptf(\state, \Phase, \fvec) \in \{1, \ldots , T\}\), 
%by Lemma~\ref{lem:taubound}.
Hence for all \(\PhaseTwo\) with \(1 \leq \PhaseTwo \leq T\),
there is a \(\Phase\) such that 
\(\PhaseTwo \in \phiupdate_\Phase(\state, \fvec)\). 
This implies \(\bigcup_{\Phase=1}^{T} \phiupdate_\Phase(\state, \fvec) = \{1, \ldots, T\}\). 
Furthermore, there cannot be more than one \(\Phase\) such that
\(\PhaseTwo \in \phiupdate_{\Phase}(\state, \fvec)\), since \(\ptf\) is functional.
Now we have 
$\sum_{\Phase=1}^T k^\prime_\Phase = \sum_{\Phase=1}^T \sum_{\PhaseTwo \in \phiupdate_\Phase(\state, \fvec)} k_\PhaseTwo = \sum_{\Phase=1}^T k_\Phase = N.$
\qed
\end{proof}

\subsection{Failure Vector Calculation}
\label{subsec:fveccalc}
We construct all possible failure vectors for
a global state by considering every group
of oscillators in decreasing order of phase. At each
stage we determine if the oscillators would fire.
If they fire then we consider each outcome where any, all, or none
of the firings result in a broadcast failure. We
then add a corresponding value to a partially calculated failure vector and consider the
next group of oscillators with a lower phase. If the oscillators do not fire
then there is nothing left to do, since by Def.~\ref{def:pert}
we know that $\pert$ is increasing,
therefore all oscillators
with a lower phase will also not fire. We can then pad the partial failure
vector with $\star$ appropriately to indicate that no failure could happen
since no oscillator fired.
\begin{table}[b]
	\caption{Construction of a possible failure vector for a global state $\state_3 = \langle 0, 0, 0, 0, 0, 2, 1, 0, 0, 5 \rangle$.}
	\setlength\tabcolsep{0.2cm}
	\centering
	\begin{tabular}{|c|c|r|c|c|}
	\hline
	iteration ($i$) & $\pi_1$ & failure vector $B$ & fired & branches \\
	\hline
	$0$ & $\langle 0, 0, 0, 0, 0, 2, 1, 0, 0, 5 \rangle$ & $\langle \rangle$ & -- & $\mathit{false}$  \\
	$1$ & $\langle 0, 0, 0, 0, 0, 2, 1, 0, 0, \underline{5} \rangle$ & $\langle 0 \rangle$ & $\mathit{true}$ & $\mathit{true}$ \\	
	$2$ & $\langle 0, 0, 0, 0, 0, 2, 1, 0, \underline{0}, 5 \rangle$ & $\langle 0, 0 \rangle$ &
$\mathit{true}$ & $\mathit{false}$ \\	
	$3$ & $\langle 0, 0, 0, 0, 0, 2, 1, \underline{0}, 0, 5 \rangle$ & $\langle 0, 0, 0 \rangle$ & $\mathit{true}$ & $\mathit{false}$ \\	
	$4$ & $\langle 0, 0, 0, 0, 0, 2, \underline{1}, 0, 0, 5 \rangle$ & $\langle 1, 0, 0, 0 \rangle$ & $\mathit{true}$ & $\mathit{true}$ \\	
	$5$ & $\langle 0, 0, 0, 0, 0, \underline{2}, 1, 0, 0, 5 \rangle$ & $\langle \star, \star, \star, \star, \star, \star, 1, 0, 0, 0 \rangle$ & $\mathit{false}$ & -- \\	
	\hline
	\end{tabular}
	\label{tab:fvecexample}
\end{table}

Table~\ref{tab:fvecexample}
illustrates how a possible failure vector for global state
$\state_3$ in Fig.~\ref{fig:example2} is iteratively constructed.
The first three columns respectively indicate
the current iteration $i$, the global state $\state_3$ with
the currently considered oscillators underlined, and the elements of the
failure vector $\fvec$ computed so far. The fourth column is $\mathit{true}$ if the
oscillators with phase $T{+}1{-}i$ would fire given the broadcast failures in the
partial failure vector.
We must consider all outcomes of any or all firings resulting in broadcast
failure. The final column therefore
indicates whether the value added to the partial failure vector in
the current iteration is the only possible value ($\mathit{false}$), or 
a choice from one of several possible values ($\mathit{true}$).

Initially we have an empty partial failure vector.
At the first iteration there are $5$ oscillators with a phase
of $10$. These oscillators will fire so we must consider each case
where $0, 1, 2, 3, 4$ or $5$ broadcast failures occur. Here we 
choose $0$ broadcast failures, which is then added to the partial failure vector.
At iterations $2$ and $3$ the oscillators would have fired, but since there are no
oscillators with a phase of $9$ or $8$ we only have one possible
value to add to the partial failure vector, namely $0$. At iteration~$4$
a single oscillator with a phase of $7$ fires, and we
choose the case where the firing resulted in a broadcast failure. In the final
iteration oscillators with a phase of $6$ do not fire, hence we can conclude
that oscillators with phases less than $6$ also do not fire, and can
fill the partial failure vector appropriately with $\star$.

Formally, we define a family of functions $\fail$ indexed by $\Phase$,
where each $\fail_\Phase$ takes as parameters some global
state $\state$, and $V$, a vector of length $T - \Phase$. $V$
represents all broadcast failures for all oscillators with a phase
greater than $\Phase$. The function $\fail_\Phase$ then computes the set
of all possible failure vectors for $\state$ with suffix $V$.
Here we use the notation $v^\frown v^\prime$ to indicate vector
concatenation.
\begin{definition}
We define 
$\fail_\Phase : \states \times \{0, \ldots, N\}^{T - \Phase} \to \mathbb{P}((\{0, \ldots, N\} \cup \{\star\})^T)$,
for $1 \le \Phase \le T$, as the family of functions indexed by $\Phase$,
where $\state = \langle k_1, \ldots, k_T \rangle$ and
\begin{align*}
\fail_\Phase(\state, V) & =
	\begin{cases}
			\textstyle\bigcup_{k = 0}^{k_\Phase} \fail_{\Phase - 1}(\state, \langle k \rangle^\frown V) & \text{if } 1 < \Phase \le T \text{ and } \fire^\Phase(\state, {\{\star\}^\Phase}^\frown V) \\
			\textstyle\bigcup_{k = 0}^{k_1} \left\{ \langle k \rangle^\frown V \right\} & \text{if } \Phase = 1 \text{ and } \fire^1(\state, \langle \star \rangle^\frown V) \\
		\left\{ {\{\star\}^\Phase}^\frown V \right\} & \text{otherwise}
	\end{cases}
\end{align*}\vspace*{-3ex}
\label{def:failvecompute}
\end{definition}
Observe that the result of $\fail_T$ is always a set of well defined failure vectors, since whenever $\star$ is introduced into a failure vector at index $\Phase$, all preceding indices are also filled with $\star$, as required by Definition~\ref{def:failvec}.

\begin{definition}
Given a global state $\state \in \states$, we define 
$\failvecs_{\state}$, the set of
all possible failure vectors for that state, as
$\failvecs_{\state} = \fail_T(\state, \langle \rangle)$, and define
$\nextstates(\state)$,
the set of all successor states of $\state$, as
$\nextstates(\state) = \{\suc(\state, \fvec) \mid \fvec \in \failvecs_{\state}\}$.
\label{def:next}
\end{definition}
Note that for some global states $|\nextstates(\state)| < |\failvecs_{\state}|$,
since we may have that $\suc(\state, \fvec) = \suc(\state, \fvec^\prime)$ for
some $\fvec, \fvec' \in \failvecs_{\state}$ with $\fvec \not = \fvec'$.

Given a global state $\state$ and a failure vector $\fvec \in \failvecs_{\state}$,
we will now compute the probability of a transition being made
to state $\suc(\state, \fvec)$ in the next time step.
Recall that $\mu$ is the probability with which a
broadcast failure occurs.
Firstly we define the probability mass function
$\pfail : \{1, \ldots, N\}^2 \rightarrow [0, 1]$,
where $\pfail(k, f)$ gives the probability of $f$ broadcast failures
occurring given that $k$ oscillators fire, as
$
\pfail(k, f) =  \mu^{f} (1 - \mu)^{k - f} {{k }\choose{f}}.
$
We then denote by $\pfailvec : \states \times \failvecs_\state \to [0, 1]$ the function
mapping a possible broadcast failure vector $\fvec$ for $\state$, to
the probability of the failures in $\fvec$ occurring. That is,
\begin{align}
\label{eqn:pfailvec}
\pfailvec(\langle k_1, \ldots, k_T \rangle, \langle f_1, \ldots, f_T \rangle) =
		\prod_{\Phase = 1}^T
		\begin{cases}		
			\pfail(k_\Phase, f_\Phase) & \text{if } f_\Phase \ne \star \\
			1 & \text{otherwise}		
		\end{cases}
\end{align}
\begin{lemma}
	For any global state $\state$, $\pfailvec$ is a 
	discrete probability distribution over $\failvecs_\state$. Formally,
	$\sum_{\fvec \in \failvecs_\state} \pfailvec(\state, \fvec) = 1$.
\end{lemma}
\begin{proof}
Given a global state $\state = \langle k_1, \ldots, k_T \rangle$ we can
construct a tree of depth $T$ where each leaf node is labelled with
a possible failure vector for $\state$, and each node $\node$ at depth
$\Phase$ is labelled with a vector of length $\Phase$
corresponding to the last $\Phase$ elements of a possible failure vector
for $\state$. We denote the label of a node $\node$ by $V(\node)$.
We label each node $\node_\omega$ with
$\langle \omega \rangle^\frown V(\node)$.
We iteratively construct the tree, starting with the root node,
$\mathit{root}$,
at depth $0$, which we label with the empty tuple $\langle \rangle$.
For each node $\node$ at depth $0 \le \Phase < T$
we construct the children of $\node$ as follows:
\begin{enumerate}
\item
If oscillators with phase $\Phase$ fire
we define the sample space $\Omega = \{0, \ldots, n_\Phase\}$
to be a set of disjoint events, where each $\omega \in \Omega$
is the event where $\omega$ broadcast failures occur, given that
$k_\Phase$ oscillators fired.
For each $\omega \in \Omega$ there is a child $\node_\omega$ of
$\node$ with label $\langle \omega \rangle^\frown V(\node)$,
and we label the edge from $\node$ to $\node_\omega$ with
$\pfail(k_\Phase, \omega)$.
\item
\label{item:nofire}
If oscillators with phase $\Phase$ do not fire then $\node$ has
a single child $\node_\star$ labelled with
$\langle \star \rangle^\frown V(\node)$, and we label the edge
from $\node$ to $\node_\star$ with $1$.
\end{enumerate}
We denote the label of an edge from a node $\node$ to its
child $\node^\prime$ by $L(\node, \node^\prime)$.
For case~\ref{item:nofire} we can observe that if oscillators with
phase $\Phase$ do not fire then we know that oscillators with any
phase $\PhaseTwo < \Phase$ will also not fire, since from Def.~\ref{def:pert}
we know that $\pert$ is an increasing function.
Hence, all descendants of $\node$ will also have
a single child, with an edge labelled with $1$, and each node
is labelled with the label of its parent, prefixed with $\langle \star \rangle$.

After constructing the tree we have a vector of length $T$
associated with each leaf node, corresponding to a failure
vector for $\state$. The set $\failvecs_\state$ of all possible failure vectors for
$\state$ is therefore the set of all vectors labelling leaf nodes.
We denote by $\ppath(\node)$ the product of all labels on edges
along the path from $\node$ back to the root. Given a global state
$\state = \langle k_1, \ldots, k_T \rangle$ and a failure vector
$\fvec = \langle f_1, \ldots, f_T \rangle \in \failvecs_\state$ labelling some
leaf node $\node$ at depth $T$, we can see that
\begin{align*}
\ppath(\node) = 1 \cdot \prod_{\Phase=1}^T
	\begin{cases}
		\pfail(k_\Phase, f_\Phase) & \text{if } f_\phase \ne \star \\
		1 & \text{otherwise}
	\end{cases}
	= \pfailvec(\state, \fvec).
\end{align*}

Let $D^\Phase$ denote the set of all nodes at depth $\Phase$.
We show $\sum_{d \in D^\Phase} \ppath(d) = 1$ by induction
on $\Phase$. For $\Phase = 0$, i.e., $D^\Phase = \{\mathit{root}\}$,
the property holds by definition. 
Now assume that $\sum_{d \in D^\Phase} \ppath(d) = 1$ holds for
some $0 \le \Phase < T$.
Let $\node$ be some node in $D^\Phase$, and let $C^\node$ be the set of all
children of $\node$. Consider the following two cases:
If oscillators with phase $\Phase$ do not
fire then $|C^\node| = 1$, and for the only $c \in C^\node$ we have that
$L(\node, c) = 1$. If oscillators with phase $\Phase$ 
fire observe that $\pfail$ is a probability mass
function for a random variable defined on the sample space
$\Omega = \{0, \ldots, k_\Phase\}$.
In either case we can see that $\sum_{c \in C^\node} L(\node, c) = 1$.
Note that $D^{\Phase + 1} = \bigcup_{d \in D^\Phase} C^d$, and
recall that $L(d, c) \cdot \ppath(d) = \ppath(c)$.
Therefore,
\begin{align*}
\sum_{d \in D^{\Phase + 1}} \ppath(d) = \sum_{d \in D^\Phase} \sum_{c \in C^d} L(d, c) \cdot \ppath(d) = \sum_{d \in D^\Phase} \left( \ppath(d) \sum_{c \in C^d} L(d, c) \right).
\end{align*}
Since $\sum_{c \in C^d} L(d, c) = 1$ for each $d \in D^\Phase$,
and from the induction hypothesis, we then have that 
\begin{align*}
\sum_{d \in D^\Phase} \left( \ppath(d) \sum_{c \in C^d} L(d, c) \right)
=\sum_{d \in D^\Phase}\ppath(d) = 1.
\end{align*}
We have already shown that $\ppath(\node) = \pfailvec(\state, \fvec)$ for any
leaf node $\node$ labelled with a failure vector $\fvec$, and since the set
of all labels for leaf nodes is $\failvecs_\state$ we can conclude that
\begin{align*}
\sum_{\fvec \in \failvecs_\state} \pfailvec(\state, \fvec) = \sum_{d \in D^T} \ppath(d) = 1.
\end{align*}
This proves the lemma.~\qed
\end{proof}

\begin{example}
\label{ex:prob}
We consider again the global states
$\state_3 = \langle 0, 0, 0, 0, 0, 2, 1, 0, 0, 5 \rangle$
and
$\state_4 = \langle 6, 0, 0, 0, 0, 0, 0, 0, 0, 2 \rangle $, given in
Fig.~\ref{fig:example2}, of the population model instantiated
in Example~\ref{ex:globalstates}, and the failure vector
$\fvec = \langle \star, \star, \star, \star, \star, \star, 1, 0, 0, 0 \rangle$
given in
Example~\ref{ex:successor}, noting that $\fvec \in \failvecs_{\state_3}$,
$\suc(\state_3, \fvec) = \state_4$, and $\mu = 0.1$. We calculate the
probability of a transition being made from
$\state_3$ to $\state_4$ as
\begin{align*}
 & \pfailvec(\langle 0, 0, 0, 0, 0, 2, 1, 0, 0, 5 \rangle, \langle \star, \star, \star, \star, \star, \star, 1, 0, 0, 0 \rangle)\\
&= 1 \cdot 1 \cdot 1 \cdot 1 \cdot 1 \cdot 1 \cdot \pfail(1, 1) \cdot
   \pfail(0, 0) \cdot \pfail(0, 0) \cdot \pfail(5, 0)\\
&   = (0.1^1 \cdot 0.9^0 \cdot 1)  \cdot (1) \cdot (1) \cdot ( 0.1^0 \cdot 0.9^5 \cdot 1) = 
  0.059049
  \end{align*}
%\pgcomment{Finish this when concrete values are given for the population
%model in example~\ref{ex:globalstates}.}{}
\end{example}

We now have everything we need to fully describe the evolution of the
global state of a population model over time.
An execution path of a population model $\psystem$ is an infinite sequence
of global states
$\path = \state_0 \state_1 \state_2 \state_3 \cdots$,
where $\state_0$ is called the \emph{initial state}, and 
$\state_{k + 1}  \in \nextstates(\state)$ for all $k \ge 0$.

\subsection{Synchronisation}
When all oscillators in a population model have the same phase
in a global state we say that the state is \emph{synchronised}.
Formally, a global state $\state = \langle k_1, \ldots, k_T \rangle$
is \emph{synchronised} if, and only if, there is some
$\Phase \in \{1, \ldots, T\}$ such that $k_\Phase = N$, and hence
$k_{\Phase^\prime} = 0$ for all $\Phase^\prime \neq \Phase$.
We will often want to reason about whether some particular run $\path$
of a model leads to a global state that is
synchronised.
We say that a path $\path = \state_0 \state_1 \cdots$
synchronises if, and only if, there exists some $k \ge 0$
such that $\state_k$ is synchronised. Once a synchronised global
state is reached any successor states will also be synchronised.
Finally we can say that a model synchronises if, and only if,
all runs of the model synchronise.
%In Fig.~\ref{fig:example2} global state $\state_2$ is
%synchronised, since $n_1 = N$.

\subsection{Model Construction}
\label{subsec:model}
%We use the probabilistic model checker \prism{}~\cite{kwiatkowska2011prism}
%to formally verify properties of our model. Given a probabilistic model of a
%system, \prism{} can be used to reason about temporal and probabilistic
%properties of the input model, by checking requirements expressed in a suitable
%formalism against all possible runs of the model.

%We define our input models as
%\emph{Discrete Time Markov Chains} (DTMCs).
%A DTMC is a tuple $(Q, \init, \mathbf{P})$ where $Q$ is a set of states,
%$\init \in Q$
%is the initial state, and $\mathbf{P} : Q \times Q \rightarrow [0, 1]$ is the 
%function mapping %ordered 
%pairs of states $(q, q^\prime)$ to the probability with
%which a transition from $q$ to $q^\prime$ occurs, where
%$\sum_{q^\prime \in Q} \mathbf{P}(q, q^\prime) = 1$ for all $q \in Q$.

Given a population model $\psystemfull$ we 
construct a DTMC $D(\psystem) = (Q, \init, \mathbf{P}, L)$ where $L$ ranges
over the singleton $\{\synchlab\}$. We define the set of
states $Q$ to be $\Gamma(\psystem) \cup \{\init\}$, where $\init$ is
the initial state of the DTMC.
For each $\state = \langle k_1, \ldots, k_T \rangle \in \states(S)$,
we set $L(\state) = \{\synchlab\}$ if $k_T = N$. % for some $1 \le i \le T$.

In the initial state all oscillators are
\emph{unconfigured}. That is, oscillators have not yet been assigned a
value for their phase.
For each
$\state = \langle k_1, \ldots, k_T \rangle \in Q \setminus \{\init\}$ we define
\begin{align*}
\label{eq:initialprob}
%	\prmatrix(\mathit{init}, q) = \frac{1}{T^N} \prod_{i = 1}^T {{N - (\sum_{j = 1}^{i  - 1} k_j )}\choose{k_i}}  \\
	\prmatrix(\init, \state) = \frac{1}{T^N} {{N}\choose{k_1, \ldots, k_T}} %\prod_{i = 1}^T {{N - (\sum_{j = 1}^{i  - 1} k_j )}\choose{k_i}} 
\end{align*}
%\begin{equation}
%	\mathbf{P}(\mathit{init}, q) = \prod_{i = 1}^T \binom{N - (\sum_{j = 1}^{i  - 1} n_j )}{n_i}
%\end{equation}
to be the probability of moving from $\init$ to a state where
$k_i$ arbitrary oscillators are configured with the phase value $i$ for
%each $i$,
$1\leq i\leq T$. The multinomial coefficient defines the number of possible
assignments of phases to distinct oscillators that result in the global state
$\state$. The fractional coefficient normalises the multinomial coefficient with
respect to the total number of possible assignments of phases to all oscillators.
In general, given an arbitrary set of initial configurations (global states)
for the oscillators, the total number of possible phase assignments can be
calculated by computing the sum of the multinomial coefficients for each
configuration (global state) in that set.  
Since $\states$
is the set of all possible global states, we have that
\begin{align*}
	\displaystyle\sum_{\langle k_1, \ldots, k_T \rangle \in \states} {{N}\choose{k_1, \ldots, k_T}} = T^N.
\end{align*}
%%%%%to be the probability of moving from $init$ to a state where the oscillators
%%%%%are \emph{configured} with the phase values defined in $q$,
%since there
%are $N$ choose $n_1$ ways to select $n_1$ oscillators to have a phase
%of $1$, then $N - n_1$ choose $n_2$ ways to select $n_2$ oscillators
%to have a phase of $2$, and so forth. 

We assign probabilities to the transitions as follows:
% for every
% $\state \in Q \setminus \{\init\}$ we consider each
% $\state^\prime \in Q \setminus \{\init\}$ where
% $\state^\prime = \suc(\state, F)$ for some
% $F \in \failvecs_\state$, and set
% $\mathbf{P}(\state, \state^\prime) = \pfailvec(\state, F)$. For all other
% $\state\in Q \setminus \{\init\}$ and $\state^\prime \in Q$, where $\state \ne \state^\prime$
% and $\state^\prime \not\in \nextstates(\state)$, we set $\mathbf{P}(\state, \state^\prime) = 0$.
for every \(\state \in Q \setminus \{\init\}\), we consider each \(F \in \failvecs_\state\), and
set \(\prmatrix(\state, \suc(\state,F)) = \pfailvec(\state,F)\). For every combination 
of \(\state\) and \(\state^\prime\) where \(\state^\prime \not\in\nextstates(\state)\)
we set \(\prmatrix(\state, \state^\prime) = 0\).

%To facilitate the analysis of parameterwise-different population models
%we provide a Python script that allows the user to define ranges for parameters.
%The script then automatically generates a model for each set of parameter values,
%checks given properties in the model using \prism{},
%and writes user specified output to a comma separated value file which can be
%used by statistical analysis tools.\footnote{The scripts to create and analyse the data,
%along with the verification results, can be found at \url{https://github.com/PaulGainer/mc-bio-synch/tree/master/energy-analysis}}

\subsection{Model Reduction}

We now describe a reduction of the population model that 
results in a significant decrease in the size of the model, 
but is equivalent to the original model with respect to the
reachability of synchronised states.
We first distinguish between states where one or more oscillators are about
to fire, and states where no oscillators will fire at all. We refer
to these states as \emph{firing states} and \emph{non-firing} states
respectively.
\begin{definition}
\label{def:firingstates}
Given a population model $\psystem$, a global state
$\langle k_1, \ldots, k_T\rangle \in \states$
is a \emph{firing state} if, and only if, $k_T > 0$.
We denote by $\fstates$ the set of all firing states of $\psystem$,
and denote by $\nfstates = \states \setminus \fstates$ the set of
all non-firing states of $\psystem$. We will
again omit $\psystem$ if it is clear from the context
\end{definition}

%Firstly, we distinguish between states where one or more oscillators are about
%to fire, and states where no oscillators will fire at all. We refer
%to these states as \emph{firing states} and \emph{non-firing} states
%respectively.
%%\begin{definition}
%Given a population model $\psystem$, a global state
%$\langle k_1, \ldots, k_T\rangle \in \states$
%is a \emph{firing state} if, and only if, $k_T > 0$.
%We respectively denote the sets of firing and non-firing states of
%$\psystem$ by $\fstates(\psystem)$ and $\nfstates(\psystem)$, and will
%again omit $\psystem$ if it is clear from the context.

Given a DTMC $D = (Q, \init, \prmatrix, L)$
 let 
\(\lvert \prmatrix \rvert =
	\lvert \{(t, t^\prime) \mid t, t^\prime \in Q^2  \text{ and } \prmatrix(t, t^\prime) > 0\} \rvert
\)
 be the number of non-zero transitions in $\prmatrix$, and
$\lvert D \rvert = \lvert Q \rvert + \lvert \prmatrix \rvert$ to be the
total number of states and non-zero transitions in $D$.

\begin{theorem}
\label{theorem:reduction}
For every population model $\psystem$ and its corresponding DTMC
$D(\psystem) = (Q, \init, \prmatrix, L)$, there is a reduced model
$D^\prime(\psystem) = (Q^\prime, \init, \prmatrix^\prime, L^\prime)$ where 
$\lvert D^\prime(\psystem) \rvert < \lvert D(\psystem) \rvert$ and
%$\prmatrix^\prime(q, q^\prime) = 1 \text{ implies } q = q'$,% \text{ and } \synch(q)$,
%\pgcomment{Maybe mention synch and perpetual asynch here, or later?}{}
unbounded-time reachability properties
with respect to synchronised firing states in $D(\psystem)$ are preserved in $D'(\psystem)$.
In particular, the states and transitions in $D(\psystem)$ are reduced in
$D^\prime(\psystem)$ such that $Q^\prime = Q \setminus \nfstates$ and
%\pgcomment{$\lvert \prmatrix^\prime \rvert \le \lvert \prmatrix \rvert -
%2 \lvert \nfstates \rvert$ or $\lvert \prmatrix^\prime \rvert =
%\lvert \prmatrix \rvert -	2 \lvert \nfstates \rvert$ ???}{}
\begin{align*}
	\lvert Q^\prime \rvert & = 1 + \frac{T^{(N-1)}}{(N-1)!}, \\
	\lvert \prmatrix^\prime \rvert & \le \lvert \prmatrix \rvert -
	2 \lvert \nfstates \rvert
%%		\sum_{j=1}^{T-1} \sum_{k=1}^{N} {{N+T-j-k-2}\choose{N-k}},	
%%		\sum_{\delta=1}^{T-1} \sum_{k=1}^{N} {{N+\delta-1}\choose{N-k}}
\end{align*}
where $x^{(n)}$ is the rising factorial.
\end{theorem}

We now proceed to prove this theorem. To that end, we need some preliminary properties
of non-firing states and their relation to firing states.

%\pgcomment{Derivation for state reduction, probably don't need all of this.}{}
%\begin{align}
%\renewcommand*{\arraystretch}{2.0}
%\begin{array}{lll}
%	\displaystyle{{N+T-1}\choose{N}} - {{N+T-2}\choose{N}} & = &
%		\displaystyle \frac{(N+T-1)!}{N!(T-1)!} - \frac{(N+T-2)!}{N!(T-2)!} \\
%	& = & \displaystyle \frac{(N+T-1)!}{N!(T-1)!} - \frac{(T-1)(N+T-2)!}{N!(T-1)!} \\
%	& = & \displaystyle \frac{N(N+T-2)!}{N!(T-1)!} \\	
%	& = & \displaystyle \frac{(N+T-2)!}{(N-1)!(T-1)!} \\
%%	& = & \displaystyle \frac{\displaystyle\prod_{k=T}^{N-1} k}{\displaystyle\prod_{k=1}^{N} k} \\
%	& = & \displaystyle \frac{T^{(N-1)}}{(N-1)!} \\
%\end{array}	
%\end{align}
%Pochhammer symbol (rising factorial)

%\begin{align}
%\renewcommand*{\arraystretch}{2.0}
%\begin{array}{lll}
%	\displaystyle \sum_{j=1}^{T-1} j \sum_{k=1}^{N} {{N}\choose{k}} {{N+T-j-k-2}\choose{N-k}} & = & \\
%	\displaystyle \sum_{j=1}^{T-1} j \sum_{k=1}^{N} \frac{N!}{k!(N-k)!} \cdot \frac{(N+T-j-k-2)!}{(N-k)!(T-j-2k-2)!} & = & \\	
%	\displaystyle \sum_{j=1}^{T-1} j \sum_{k=1}^{N} \frac{N! (N+T-j-k-2)!}{k!((N-k)!)^2(T-j-2k-2)!} & = & \\	
%	\displaystyle \sum_{j=1}^{T-1} j \sum_{k=1}^{N} \frac{N! (T-j-2k-1)^{(N+k)}}{k!((N-k)!)^2} & = & \\	
%	\displaystyle \sum_{j=1}^{T-1} j \sum_{k=1}^{N} \frac{(N-k+1)^{(k)} (T-j-2k-1)^{(N+k)}}{k!(N-k)!} & = & \\	
%\end{array}
%\end{align}

%The reduction is achieved by removing deterministic transitions from
%the model.

%That is, all transitions labelled with probability $1$. 

\begin{lemma}
\label{lem:onesuccessor}
Every non-firing state $\state \in \nfstates$ has exactly one successor state,
and in that state all oscillator phases have increased by $1$.
\end{lemma}
\begin{proof}
Given a non-firing state $\state = \langle k_1, \ldots, k_T \rangle$
observe that as $k_T = 0$
there is only one possible failure vector for $\state$, namely $\{\star\}^T$.
The set of all successor states of $\state$ is then the singleton
$\{\suc(\state, \{\star\}^T)\}$.
By construction we can then see that $update^\Phase(\state, \{\star\}^T) = 1$
and $\phiupdate_\Phase(\state, \{\star\}^T) = \{\Phase - 1\}$
for $1 \le \Phase \le T$. The single successor state is then given
by $\suc(\state, \{\star\}^T) = \langle 0, k_1, \ldots, k_{T - 1} \rangle$.
\qed
\end{proof}
\begin{corollary}
\label{cor:determinstictransition}
An immediate corollary of Lemma~\ref{lem:onesuccessor} is that a
transition from any non-firing state is taken deterministically, since
for any $\state \in \nfstates$ we have $\pfailvec(\state, \{\star\}^T) = 1$.
\end{corollary}

\paragraph{Reachable State Reduction.}

Given a path
$\path = \state_0 \cdots \state_{n-1} \state_n$
where $\state_i \in \nfstates$ for $0 < i < n$ and $\state_0, \state_n \in \fstates$,
we omit transitions $(\state_i, \state_{i + 1})$ for $0 \le i < n$, and
instead introduce a direct transition from $\state_0$, the first firing state,
to $\state_n$, the next firing state in the sequence.
For any $\state = \langle k_1, \ldots, k_T \rangle \in \states$
let $\delta_\state = \max\{\Phase \mid k_\Phase > 0 \text{ and } 1 \le \Phase \le T \}$
be the highest phase of any oscillator in $\state$. The successor state
of a non-firing state is then the state where all phases have increased by
$T - \delta_\state$. Observe that
$T - \delta_\state = 0$ for any $\state \in \fstates$.

%\slcomment{Mention that $\delta$ is always less than $T$.}{}
\begin{definition}
The \emph{deterministic successor function} $\sucskip : \states \to \fstates$, given by
\begin{align*}
\sucskip(\langle k_1, \ldots, k_T\rangle) = {\{0\}^{T-\delta_\state}}^\frown \langle k_1, \ldots, k_{\delta_\state} \rangle ,
\end{align*}
maps a state $\state \in \states$ to the next firing state reachable by taking
$T-\delta_\state$ deterministic transitions. Observe that for any firing state
$\state$ we have $\delta_\state = T$, and hence that $\sucskip(\state) = \state$.
\label{def:dsucc}
\end{definition}
We now update the definition for the set of all successor states
for some global state $\state \in \states$ to incorporate the deterministic
successor function. 
\begin{definition}
Given a global state $\state \in \states$, we define $\nextstatesskip(\state)$
to be the set of all successor states of $\state$, where
\begin{align*}
	\nextstatesskip(\state) =
	\{\sucskip(\suc(\state, \fvec)) \mid \fvec \in \failvecs_{\state}\}.
\end{align*}
%We denote a transition from a state $\state$ to a successor state
%$\state^\prime \in \nextstatesskip(\state)$ by $\state \transskip \state^\prime$.
\end{definition}

%We can generalise this reduction of the set of all initial configurations for
%oscillators, and hence show how we can reduce the set of all reachable
%states to $\fstates$.
\begin{definition}
\label{def:predecessors}
Given a firing state $\state \in \fstates$ let $\predstates(\state)$ be
the set of all non-firing predecessors of $\state$, where $\state$ is reachable from the
predecessor by taking some positive number of transitions deterministically. Formally,
\[\predstates(\state) = \{\state^\prime \mid \state^\prime \in \nfstates \text{ and } \sucskip(\state^\prime) = \state \}.\]
We refer to all states $\state^\prime \in \predstates(\state)$ as
\emph{deterministic predecessors} of $\state$.
\end{definition}
Then given $D = (Q, \init, \prmatrix, L)$ with
$Q = \{\init\} \cup \states$, we define
$Q^\prime = Q \setminus \bigcup_{\state \in \fstates} \predstates(\state)$
to be the reduction of $Q$ where all non-firing states from which a
firing state can be reached deterministically are removed.

\begin{lemma}
\label{lem:qprimeisfiring}
For any $D(\psystem) = (Q, \init, \prmatrix, L)$ with
$Q = \states \cup \{\init\}$,
the reduction $Q^\prime$ is equal to $\fstates \cup \{\init\}$.
\end{lemma}
\begin{proof}
Let $P = \bigcup_{\state \in \fstates} \predstates(\state)$ be the set
of all predecessors of firing states in $\fstates$.
Since $Q = \states \cup \{\init\}$ and 
$Q^\prime = Q \setminus P$ we can see that $Q^\prime = \fstates \cup \{\init\}$
if, and only if, $P = \nfstates$.
From Definition~\ref{def:predecessors} it follows that
$P \subseteq \nfstates$.
In addition, for any $\state \in \nfstates$ there is some state
$\state^\prime$ such that $\state \in \predstates(\state^\prime)$ and
$\state^\prime = \sucskip(\state) \in \fstates$,
%In addition, $\state \in \predstates(\state^\prime)$
%for any $\state \in \nfstates$ since
%$\state^\prime = \sucskip(\state) \in \fstates$,
hence $\nfstates \subseteq P$
and the lemma is proved.
%Since $\sucskip(\state) \in \fstates$ for any $\state \in \nfstates$
%there  must be some $\state^\prime$ such that
%$\state \in \predstates(\state^\prime)$, hence
%$\nfstates \subseteq P$ and $P = \nfstates$ which proves the lemma.
\qed
\end{proof}

\begin{lemma}
\label{lem:statereduction}
For a population model $\psystemfull$ and its corresponding DTMC
$D = (Q, \init, \prmatrix, L)$ with $Q = \states \cup \{\init\}$,
the number of states in the reduction of $Q$ is given by
$
\lvert Q^\prime \rvert = 1 + \frac{T^{(N-1)}}{(N-1)!},
$
where $x^{(n)}$ is the rising factorial.
\end{lemma}
\begin{proof}
Observe that there are ${{N + T - 1}\choose{N}}$ ways to assign
$T$ distinguishable phases to $N$ indistinguishable oscillators~\cite{feller1968introduction}.
Since $Q = \states \cup \{\init\}$ and $\states$ is the set of all
possible configurations for oscillators we can see that
$\lvert Q \rvert = {{N + T - 1}\choose{N}} + 1$.
For any non-firing state $\state = \langle k_1, \ldots k_T \rangle \in \nfstates$
we know from Definition~\ref{def:globalstates} that
$\sum_{\Phase=1}^{T} k_\Phase = N$
and from Definition~\ref{def:firingstates} that $k_T = 0$, so it must be the
case that $\sum_{\Phase=1}^{T-1} k_\Phase = N$. That is, there must be
${{N + T - 2}\choose{N}}$ ways to assign $T - 1$ distinguishable phases to
$N$ indistinguishable oscillators, and so 
$\lvert \nfstates \rvert = {{N + T - 2}\choose{N}}$.
From Lemma~\ref{lem:qprimeisfiring} we know that
$Q^\prime = Q \setminus \nfstates$ so it must be the case that
$ %\begin{align*}
	\lvert Q^\prime \rvert = \lvert Q \vert - \lvert \nfstates \rvert 
	= 1 + {{N+T-1}\choose{N}} - {{N+T-2}\choose{N}} 
	= 1 + \frac{T^{(N-1)}}{(N-1)!}.
$ %\end{align*}
%\begin{align*}
%\renewcommand*{\arraystretch}{2.0}
%\begin{array}{rllll}
%	\lvert Q^\prime \rvert & = & \lvert Q \vert - \lvert \nfstates \rvert 
%	& = & 1 + \displaystyle{{N+T-1}\choose{N}} - {{N+T-2}\choose{N}} \\
%	& & & = & 1 + \displaystyle \frac{(N+T-1)!}{N!(T-1)!} - \frac{(N+T-2)!}{N!(T-2)!} \\
%%	& & & = & 1 + \displaystyle \frac{(N+T-1)!}{N!(T-1)!} - \frac{(T-1)(N+T-2)!}{N!(T-1)!} \\
%%	& & & = & 1 + \displaystyle \frac{N(N+T-2)!}{N!(T-1)!} \\	
%	& & & = & 1 + \displaystyle \frac{(N+T-2)!}{(N-1)!(T-1)!} \\
%%	& & & = & 1 + \displaystyle \frac{\displaystyle\prod_{k=T}^{N-1} k}{\displaystyle\prod_{k=1}^{N} k} \\
%	& & & = & 1 + \displaystyle \frac{T^{(N-1)}}{(N-1)!}. \\
%\end{array}	
%\end{align*}
\qed
\end{proof}

\paragraph{Transition Matrix Reduction.}

Here we describe the reduction in the number of non-zero transitions
in the model. We ilustrate how initial transitions to non-firing
states are removed by using a simple example, and then describe how
we remove transitions from firing states to any successor non-firing
states.. 

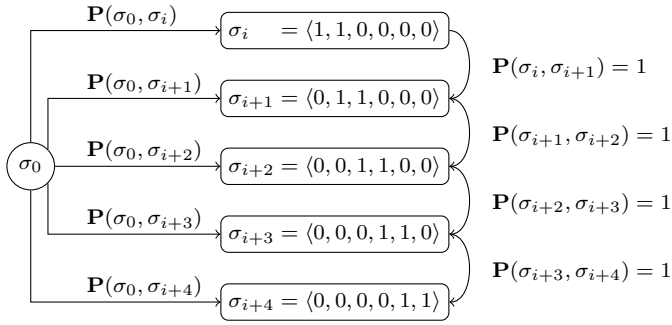
\begin{figure}[bt]
	\begin{center}
		\begin{tikzpicture}
			\tikzstyle{arrow} = [->];
			\tikzstyle{dasharrow} = [->, dashed];
			\def \dist {4.0cm}
			\def \sep {0.9cm};
			\def \phpos {-2.5cm};
			\def \pvpos {0.2cm};
			\def \pphpos {3.25cm};
			\def \ppvpos {-0.5cm};
			\def \sone {$\state_{i\phantom{+1}} = \langle 1, 1, 0, 0, 0, 0 \rangle$};
			\def \stwo {$\state_{i+1} = \langle 0, 1, 1, 0, 0, 0 \rangle$};
			\def \sthree {$\state_{i+2} = \langle 0, 0, 1, 1, 0, 0 \rangle$};
			\def \sfour {$\state_{i+3} = \langle 0, 0, 0, 1, 1, 0 \rangle$};
			\def \sfive {$\state_{i+4} = \langle 0, 0, 0, 0, 1, 1 \rangle$};
			
			\node (init) [draw, circle] at (0, 0) {$\init$};
			\node (s1) [rounded corners = 0.1cm, draw] at (\dist, {\sep * 2}) {\sone};
			\node (s2) [rounded corners = 0.1cm, draw] at (\dist, \sep) {\stwo};
			\node (s3) [rounded corners = 0.1cm, draw] at (\dist, 0) {\sthree};
			\node (s4) [rounded corners = 0.1cm, draw] at (\dist, {-\sep}) {\sfour};
			\node (s5) [rounded corners = 0.1cm, draw] at (\dist, {-(\sep * 2)}) {\sfive};
			\draw [arrow] (init.north)|-(s1.west);
			\draw [arrow] (init.north east)|-(s2.west);
			\draw [arrow] (init.east)|-(s3.west);
			\draw [arrow] (init.south east)|-(s4.west);
			\draw [arrow] (init.south)|-(s5.west);
			\node (s1p) at ([shift={(\phpos, \pvpos)}]s1) {$\prmatrix(\init, \state_{i})\phantom{+l}$};
			\node (s2p) at ([shift={(\phpos, \pvpos)}]s2) {$\prmatrix(\init, \state_{i+1})$};
			\node (s3p) at ([shift={(\phpos, \pvpos)}]s3) {$\prmatrix(\init, \state_{i+2})$};
			\node (s4p) at ([shift={(\phpos, \pvpos)}]s4) {$\prmatrix(\init, \state_{i+3})$};
			\node (s5p) at ([shift={(\phpos, \pvpos)}]s5) {$\prmatrix(\init, \state_{i+4})$};
			\draw [arrow] (s1.east) to [out = 0, in = 0] (s2.east);
			\draw [arrow] (s2.east) to [out = 0, in = 0] (s3.east);
			\draw [arrow] (s3.east) to [out = 0, in = 0] (s4.east);
			\draw [arrow] (s4.east) to [out = 0, in = 0] (s5.east);
			\node (s12p) at ([shift={(\pphpos, \ppvpos)}]s1) {$\prmatrix(\state_{i}, \state_{i+1})=1\phantom{+l}$};
			\node (s23p) at ([shift={(\pphpos, \ppvpos)}]s2) {$\prmatrix(\state_{i+1}, \state_{i+2})=1$};
			\node (s34p) at ([shift={(\pphpos, \ppvpos)}]s3) {$\prmatrix(\state_{i+2}, \state_{i+3})=1$};
			\node (s45p) at ([shift={(\pphpos, \ppvpos)}]s4) {$\prmatrix(\state_{i+3}, \state_{i+4})=1$};		
		\end{tikzpicture}
	\end{center}
	\caption{Five possible initial configurations in $Q$ for $N = 2$, $T = 6$.}
	\label{fig:initialstates}	
\end{figure}

Figure~\ref{fig:initialstates} shows five possible initial
configurations $\state_i, \ldots, \state_{i+4} \in Q$ for $N = 2$ oscillators with
$T = 6$ values for phase, where a transition is taken from $\init$ to each $\state_k$
with probability $\prmatrix(\init, \state_k)$.
Any infinite run of $D$ where a transition is taken
from $\init$ to one of the configured states $\state_i, \ldots, \state_{i+3}$ will
pass through $\state_{i+4}$, since all
transitions $(\state_{i+k}, \state_{i+k+1})$ for $0 \le k \le 3$
are taken deterministically.
Also, observe that states $\state_i, \ldots, \state_{i+3}$ are not in
$Q^\prime$, since $\state_{i+4}$ is reachable from each by taking some
number of deterministic transitions. We therefore set the probability
of moving from $\init$ to $\state_{i+4}$ in $\prmatrix^\prime$ to be
the sum of the probabilities of moving from $\init$ to $\state_{i+4}$
and each of its predecessors in $\prmatrix$. Generally, given a
state $\state \in Q^\prime$ where $\state \ne \init$, we set
$\prmatrix^\prime(\init, \state) =
	\prmatrix(\init, \state) +
	\sum_{\state^\prime \in \predstates(\state)}
		\prmatrix(\init, \state^\prime)$.

We now define how we calculate the probability with which a transition is
taken from a firing state to each of its
possible successors. For each firing state $\state \in Q^\prime$ we
consider each possible successor $\state^\prime \in \nextstatesskip(\state)$
of $\state$ and define $\failvecs_{\state \to \state^\prime}$ to be the set
of all possible failure vectors for $\state$ for which the successor of $\state$
is $\state^\prime$, given by
$\failvecs_{\state \to \state^\prime} = 
	\{\fvec \in \failvecs_\state \mid \sucskip(\suc(\state, \fvec)) = \state^\prime\}$.
We then set the probability with which a transition from $\state$ to $\state^\prime$
is taken to
$
\prmatrix^\prime(\state, \state^\prime) = \sum_{\fvec \in \failvecs_{\state \to \state^\prime}} \pfailvec(\state, \fvec).
$

\begin{lemma}
\label{lem:transitionreduction}
For a population model $\psystemfull$, the corresponding DTMC
$D = (Q, \init, \prmatrix, L)$ with
$Q = \{\init\} \cup \states$,
and its reduction
$D^\prime(\psystem) = (Q^\prime, \init, \prmatrix^\prime, L^\prime)$,
the transitions in $\prmatrix$ are reduced in $\prmatrix^\prime$ such that 
$\lvert \prmatrix^\prime \rvert \le \lvert \prmatrix \rvert - 2 \lvert \nfstates \rvert$	
%%	\sum_{\delta=1}^{T-1} \sum_{k=1}^{N} {{N+\delta-1}\choose{N-k}}
%%		\sum_{j=1}^{T-1} \sum_{k=1}^{N} {{N+T-j-k-2}\choose{N-k}}.
\end{lemma}
\begin{proof}
From Lemma~\ref{lem:qprimeisfiring} we know that
$\lvert Q^\prime \rvert = \lvert Q \setminus \nfstates \rvert$, and hence
that $\lvert \nfstates \rvert$ transitions from $\init$ to non-firing
states are not in $\prmatrix^\prime$, and from Lemma~\ref{lem:onesuccessor}
we also know that there is one transition from each non-firing
state to its unique successor state that is not in $\prmatrix^\prime$.
Since no additional transitions are introduced in the reduction it is clear that
$\lvert \prmatrix^\prime \rvert \le \lvert \prmatrix \rvert - 2 \lvert \nfstates \rvert$.
\qed
\end{proof}

\begin{lemma}
\label{lem:reachabilityproperties}
For every population model DTMC $D = (Q, \init, \prmatrix, L)$,
unbounded-time reachability
properties with respect to synchronised firing states in $D$ are preserved
in its reduction $D^\prime$.
%We denote the set of all synchronised states of $D$ by
%$\synch(D) = \{\state \mid \state \in Q \text{ and }
%		L(\state) = \synchlab\}$, and similarly for $\synch(D^\prime)$.
\end{lemma}
\begin{proof}
We want to show that for every
$\bowtie \ \in \{<,\leqslant,\geqslant,>\}$ and every $\lambda \in [0, 1]$,
if
$ %\begin{align*}
\init \models \pctlp_{\bowtie \lambda}[\lsometime \ \synchlab]
$ %\end{align*}
holds in $D$ then it also holds in $D^\prime$.
From the semantics of PCTL over a DTMC we have
\begin{align*}
\init \models \pctlp_{\bowtie \lambda}[\lsometime \ \synchlab]
	\quad \Leftrightarrow \quad
		\pmeasure\{\path \in \paths^D \mid
			\path \models \lsometime \ \synchlab\}
			\bowtie \lambda.
\end{align*}
Therefore we need to show that
\begin{align*}
		\pmeasure^D\{\path \in \paths^D \mid
			\path \models \lsometime \ \synchlab\}
		=_\state
		\pmeasure^{D^\prime}\{\path^\prime \in \paths^{D^\prime} \mid
			\path^\prime \models \lsometime \ \synchlab\},			
\end{align*}
where $\pmeasure^D$ and $\pmeasure^{D^\prime}$
denote the probability measures with respect to the sets of infinite
paths from $\init$ in $D$ and $D^\prime$ respectively.

Given a firing state $\fstate \in Q$ we denote by
$\paths^{D}_{\fstate}$ the set of all infinite paths of $D$ starting in
$\init$ where the first firing state reached   along that path is
$\fstate$. All such sets for all firing states in 
$Q$ form a partition, such that
$\bigcup_{\fstate \in \fstates} \paths^{D}_{\fstate} = \paths^{D}$.
That is, for all firing states $\fstate, \fstateprime \in Q$ where
$\fstate \ne \fstateprime$ we have that
$\paths^{D}_{\fstate} \cap \paths^{D}_{\fstateprime} = \emptyset$.
%First observe that we can partition the set of infinite paths starting
%in $\init$ in $D$ by considering all firing states $\fstate \in \states$,
%and then grouping paths where the first reached firing state along
%that path is $\fstate$.

Now observe that any infinite path $\path$ of $D$
can be written in the form
$\path = \init \nfpath_{1} \fstate_1 \nfpath_{2} \fstate_2 \cdots$
where $\fstate_i$ is the $i^{th}$ firing state in the path and each
$\nfpath_i = \state_i^1 \state_i^2 \cdots \state_i^{k_i}$ is a possibly
empty sequence of $k_i$ non-firing states. Then for every such path
in $D$ there is a corresponding path $\path^\prime$ of $D^\prime$ without
non-firing states, and of the form 
$\path^\prime = \init \fstate_1 \fstate_2 \fstate_3 \cdots$,
as for any $i$ we have $\state_i^j \in \predstates(\fstate_i)$ for all
$1 \le j \le k^i$.
As only deterministic transitions have been removed
in $D^\prime$ we can see that
$\pmeasure^{D}\{\fstate_1 \nfpath_{2} \fstate_2 \cdots\} =
		\pmeasure^{D^\prime}\{\fstate_1 \fstate_2 \fstate_3 \cdots\}$.
%for any $\path = \init, \nfpath_{1}, \fstate_1, \nfpath_{2}, \fstate_2, \ldots$.
%in $D$ and corresponding path
%$\path^\prime = \init, \fstate_1, \fstate_2, \fstate_3, \ldots$ in $D^\prime$.
Hence, we only have to consider the finite paths from $\init$ to $\fstate_1$.
To that end, observe that there are $\left|\predstates(\fstate_1)\right|$
possible prefixes for each path from $\init$ to $\fstate_1$
where the initial transition
is taken from $\init$ to some non-firing predecessor of $\fstate_1$, plus
the single prefix where the initial transition is taken to $\fstate_1$
itself. Overall there are exactly $\left|\predstates(\fstate_1)\right| + 1$ distinct
finite prefixes that have $\path^\prime$ as their corresponding path in $D^\prime$.
We denote the set of these prefixes for a path $\path^\prime$ in $D^\prime$
by $\pref(\path^\prime)$.
Since the measure of each finite prefix extends to a measure over the
set of infinite paths sharing that prefix, it is sufficient
to show that the sum of the probabilities for these finite prefixes
is equal to the probability of the unique prefix $\state_0, \fstate_1$
of $\path^\prime$, that is
$\pmeasure^{D} \pref(\path^\prime) = \pmeasure^{D^\prime} \{\init, \fstate_1\}$.
We can then write
\begin{align*}
	\pmeasure^{D} \pref(\path^\prime)
	& =	\prmatrix(\init, \fstate_1) + 
		\sum_{\state^\prime \in \predstates(\fstate_1)} \prmatrix(\init, \state^\prime)
			\cdot 1^{k_{\state^\prime}} \\
	& = \prmatrix(\init, \fstate_1) + 
		\sum_{\state^\prime \in \predstates(\fstate_1)} \prmatrix(\init, \state^\prime),
\end{align*}
where $k_{\state^\prime}$ is the number of deterministic transitions that lead
from $\state^\prime$ to $\fstate_1$ in $D$.
Now recall that for any $\state \in Q^\prime \setminus \{\init\}$ we have
\begin{align*}
\prmatrix^\prime(\init, \state) = 
	\prmatrix(\init, \state) +
	\sum_{\state^\prime \in \predstates(\state)} \prmatrix(\init, \state).
\end{align*}
So we have shown that
$\pmeasure^{D} \pref(\path^\prime)  =
	\pmeasure^{D^\prime} \{\init, \fstate_1\}$
and the lemma is proved.
\qed
\end{proof}

\begin{proof}[of Theorem~\ref{theorem:reduction}]
Follows from Lemmas~\ref{lem:statereduction} and~\ref{lem:transitionreduction}
for the reduction of states and
transitions respectively, and from Lemma~\ref{lem:reachabilityproperties}
for the preservation of unbounded time reachability properties.
\end{proof}

\subsection{Empirical Analysis}

Table~\ref{tab:reductionstates} shows the number of reachable states
and transitions of the DTMC, and corresponding reduction,
for different population sizes ($N$) and oscillation cycle
lengths ($T$), using the Mirollo and Strogatz model of
synchronisation~\cite{mirollo1990synchronization}.
The number of reachable
states is stable under changes to the parameters $R$,
$\epsilon$, and $\mu$, since every
possible firing state is always reachable from the initial state.
For the results shown here the parameters were arbitrarily set to
$R = 1$, $\epsilon = 0.1$. The underlying
graph of the DTMC, and hence the number of transitions, is stable
under changes to the parameter $\mu$, and is not if interest
here.

\begin{table}[tb]
	\begin{center}
		\caption{Reduction in state space and transitions for different 
			model instances.}
		\label{tab:reductionstates}
		\setlength{\tabcolsep}{8.5pt}
		\begin{tabular}{llrrrrrr}
			\hline
			& & \multicolumn{2}{c}{$D$} & \multicolumn{2}{c}{$D^\prime$}
			& \multicolumn{2}{c}{Reduction (\%)} \\
			$N$ & $T$ & States & Transitions & States & Transitions & States & Transitions \\
			\hline
%			3 & 6 & 113 & 22 & 80.5 \\
%			5 & 6 & 505 & 127 & 74.9 \\
%			8 & 6 & 2575 & 793 & 69.2 \\
%			3 & 8 & 241 & 37 & 84.6 \\
%			5 & 8 & 1585 & 331 & 79.1\\
%			8 & 8 & 12871 & 3433 & 73.3 \\
%			3 & 10 & 441 & 56 & 87.3 \\
%			5 & 10 & 4005 & 716 & 82.1\\
%			8 & 10 & 48621 & 11441 & 76.5 \\
			3 & 6 & 113 & 188 & 22 & 52 & 80.5 & 72.3 \\
			5 & 6 & 505 & 1030 & 127 & 389 & 74.9 & 62.2 \\
			8 & 6 & 2575 & 7001 & 793 & 3154 & 69.2 & 54.9 \\
			3 & 8 & 241 & 410 & 37 & 97 & 84.6 & 76.3 \\
			5 & 8 & 1585 & 3250 & 331 & 1097 & 79.1 & 66.2\\
			8 & 8 & 12871 & 34615 & 3433 & 14519 & 73.3 & 58.1\\
			3 & 10 & 441 & 752 & 56 & 156 & 87.3 & 79.3 \\
			5 & 10 & 4005 & 8114 & 716 & 2484 & 82.1 & 69.4 \\
			8 & 10 & 48621 & 128936 & 11441 & 50883 & 76.5 & 60.5\\
			\hline
		\end{tabular}
	\end{center}
\end{table}

Table~\ref{tab:reductiontrans} shows the number of transitions of the DTMC,
and corresponding reduction, for various population model instances, and again
uses the Mirollo and Strogatz model of synchronisation. 
Increasing the length of the refractory period ($R$)
results in an increase in the
reduction of transitions in the model. A longer refractory period leads
to more firing states where the firing of a group of oscillators
is ignored. This results in successor states having oscillators with lower
values for phase, and hence a longer sequence of deterministic transitions
(later removed in the reduction) leading to the next firing state.
Conversely, increasing the strength of the coupling between oscillators
($\epsilon$) results in a decrease in the reduction of transitions
in the model. For the Mirollo and Strogatz model of synchronisation used here,
increasing the coupling strength results in a linear increase in the pertubation
to phase induced by the firing of an oscillator. This results in successor
states of firing states having oscillators with higher values for phase, and
hence a shorter sequence of deterministic transitions leading to the next
firing state.  

\begin{table}[tb]
	\begin{center}
		\caption{Reduction in transitions for different population
			model instances.}
		\label{tab:reductiontrans}
		\setlength{\tabcolsep}{15.5pt}
		\begin{tabular}{llllllr}
			\hline
			& & & & \multicolumn{2}{c}{Transitions} \\
			$N$ & $T$ & $R$ & $\epsilon$ & $D$ & $D^\prime$ & Reduction (\%)  \\
			\hline
			\color{gray} 5 & \color{gray} 10 & 1 & \color{gray} 0.1 & 8114 & 2484 & 69.4 \\
			\color{gray} 5 & \color{gray} 10 & 3 & \color{gray} 0.1 & 7928 & 2391 & 69.8 \\
			\color{gray} 5 & \color{gray} 10 & 5 & \color{gray} 0.1 & 7568 & 2211 & 70.8 \\
			\color{gray} 5 & \color{gray} 10 & 7 & \color{gray} 0.1 & 6976 & 1915 & 72.5 \\
			\color{gray} 5 & \color{gray} 10 & 9 & \color{gray} 0.1 & 6006 & 1430 & 76.2 \\
			\color{gray} 5 & \color{gray} 10 & \color{gray} 1 & 0.01 & 6006 & 1430 & 76.2 \\
			\color{gray} 5 & \color{gray} 10 & \color{gray} 1 & 0.05 & 6426 & 1640 & 74.5 \\
			\color{gray} 5 & \color{gray} 10 & \color{gray} 1 & 0.1 & 8114 & 2484 & 69.4 \\
			\color{gray} 5 & \color{gray} 10 & \color{gray} 1 & 0.25 & 8950 & 2902 & 67.6 \\
			\color{gray} 5 & \color{gray} 10 & \color{gray} 1 & 0.5 & 9382 & 3118 & 66.7 \\
			\hline
		\end{tabular}
	\end{center}
\end{table}

\subsection{Reward Structures for Reductions}

While probabilistic reachability properties allow us to quantitatively
analyse models with respect to the likelihood of reaching a synchronised
state, they do not allow us to reason about other properties of interest, for
instance the expected time taken for the network to
synchronise~\cite{gainer2017investigating}, or the expected
energy consumption of the network~\cite{gainer2017power}. Therefore, we will
often want to augment the DTMC corresponding to a population model with rewards. 
We do this by annotating states and transitions with real-valued rewards 
(respectively costs, should values be negative) that are
awarded when states are visited, or transitions taken.

\begin{definition}
Given a DTMC $D = (Q, \init, \prmatrix, L)$ a \emph{reward structure}
for $D$ is a pair $\rstruct = (\rstate, \rtrans)$ where
$\rstate : Q \to \reals$ and $\rtrans : Q \times Q \to \reals$ are the
\emph{state reward} and \emph{state transition} functions that
respectively map real valued rewards to states and transitions in $D$.
\end{definition}

For any finite path $\path = \state_0 \cdots \state_k$  of $D$ we define the total
reward accumulated along that path up to, but not including, $\state_k$ as
\begin{align}
\label{eq:tot}
\rtotal_{\rstruct}(\state_0 \cdots \state_k) = 
	\sum_{i=0}^{k-1}
		\left(
			\rstate(\state_i) +
			\rtrans(\state_i, \state_{i+1})
		\right).
\end{align}

Given a DTMC $D = (Q, \init, \prmatrix, L)$ augmented with a reward structure
$\rstruct$, and some state $\state \in Q$,
we will often want to reason about the reward that is accumulated
along a path $\path = \init \state_1 \state_2 \cdots \in \paths$
that eventually passes
through some set of target states $\targets \subset Q$. We first define a
random variable over the set of infinite paths
$\revar_\targets : \paths \to \reals \cup \{\infty\}$. Given the
set $\path_\targets = \{j \mid \state_j \in \targets\}$ of indices of states in
$\path$ that are in $\targets$ we define the random variable
%\begin{align*}
%\revar_\targets(\path) =
%\begin{cases}
%	\begin{array}{ll}
%	\infty	&
%		\text{   if } \path_\targets = \emptyset \\
%%		\text{   if } \path[j] \not \in \targets \text{ for all } j \geqslant 0 \\	
%	\sum_{i=0}^{k-1} \rstate(\path[i]) + \rtrans(\path[i], \path[i+1]) &
%		\text{   otherwise, where } k = \min \path_\targets,
%	\end{array}	
%\end{cases}
%\end{align*}
\begin{align*}
\revar_\targets(\path) =
\begin{cases}
	\begin{array}{ll}
	\infty	&
		\text{   if } \path_\targets = \emptyset \\
	\rtotal_{\rstruct}(\state_0 \cdots \state_k)	
		& \text{   otherwise, where } k = \min \path_\targets,
	\end{array}	
\end{cases}
\end{align*}
and define the expectation of $\revar_\targets$ with respect to $\pmeasure_\state$ by
\begin{align*}
\expect[\revar_\targets] = \int_{\path \in \paths} \revar_\targets(\path)
	\ d \pmeasure
=
\sum_{\path \in \paths} \revar_\targets(\path) \pmeasure\{\path\}.
\end{align*}

The logic of PCTL can be extended to include reward properties by introducing
the state formula $\pctlr_{\bowtie r}[\pctlsometime \ \Psi]$, where
$\bowtie \in \{<,\leqslant,\geqslant,>\}$ and
$r \in \reals$~\cite{kwiatkowska2007stochastic}. Given a state
$\state \in Q$, a real value $r$, and a PCTL path formula $\Psi$, the semantics
of this formula is given by
\[
\state \models \pctlr_{\bowtie r}[\pctlsometime \ \Psi]
		\Leftrightarrow \expect[\revar_{\sat({\Psi})}] \bowtie r,
\]
where $\sat(\Phi)$ denotes the set of states in $Q$ that satisfy $\Phi$.

\begin{theorem}
For every population model $\psystem$ with corresponding DTMC
$D = (Q, \init, \prmatrix, L)$ and a reduction
$D^\prime = (Q^\prime, \init, \prmatrix^\prime, L^\prime)$ of $D$,
and for every reward structure
$\rstruct = (\rstate, \rtrans)$ for $D$, there is a reward structure
$\rstruct^\prime = (\rstate^\prime, \rtrans^\prime)$ for $D^\prime$
such that unbounded-time reachability reward properties with respect
to synchronised firing states in $D$ are preserved in
$D^\prime$.
\label{theorem:reductionreward}
\end{theorem}

Given a reward structure $\rstruct = (\rstate, \rtrans)$ for $D$ we construct
the corresponding reward structure \(\rstruct^\prime = (\rstate^\prime, \rtrans^\prime)\) as follows:
\begin{itemize}

\item There is no reward for the initial state and we set $\rstate(\init) = 0$.
\item For every firing state $\fstate$ in $Q$ with $\rstate(\fstate) = r$ we set
$\rstate^\prime(\fstate) = r$.

\item For every pair of distinct firing states $\fstate_1, \fstate_2 \in Q^\prime$,
where there is a non-zero transition from $\fstate_1$ to $\fstate_2$ in $D^\prime$,
there is a (possibly empty) sequence $\nfstate_1 \cdots \nfstate_k$ of
$k$ deterministic predecessors of $\fstate_2$ in $Q$ such that $k > 0$ implies
$\prmatrix(\fstate_1, \nfstate_1) > 0$, $\prmatrix(\nfstate_k, \fstate_2) = 1$,
and $\prmatrix(\nfstate_i, \nfstate_{i+1}) = 1$ for $1 \leqslant i < k$.
We set the reward for taking the transition from $\fstate_1$ to
$\fstate_2$ in $D^\prime$ to be
the sum of the rewards that would be accumulated across that sequence by
a path in $D$, formally 
\begin{align*}
\rtrans^\prime(\fstate_1, \fstate_2) =
	\rtotal_{\rstruct}(\fstate_1 \nfstate_1 \cdots \nfstate_k \fstate_2).
\end{align*}
%To simplify the following definition
%we will write the sequence in the form
%$\state_0, \state_1, \ldots, \state_{k}, \state_{k+1}$, where
%$\state_0 = \fstate_1$, $\state_{k+1} = \fstate_2$, and
%$\state_j = \nfstate_j$ for $1 \leqslant j \leqslant k$. The
%corresponding reward in $\rstruct^\prime$ is then defined as 
%\begin{align*}
%\rtrans^\prime(\state_0, \state_{k+1}) =
%	\sum_{n=0}^{k} \rstate(\state_n)
%	+ \sum_{n=0}^{k} \rtrans(\state_n, \state_{n+1}).
%\end{align*}

\item For every firing state $\fstate$ in $Q^\prime$ there is a non-zero
transition from the initial state $\init$ to $\fstate$ in $\prmatrix^\prime$.
Therefore, all paths of $D^\prime$ where $\fstate$ is the first firing state
along that path share the same prefix, namely $\init, \fstate$. For paths
of $D$ this is not necessarily the case, since $\fstate$ is the first firing
state not only along the path where the initial transition is taken to
$\fstate$ itself, but also along any path where the initial transition is
taken to a non-firing state from which a sequence of deterministic transitions
leads to $\fstate$ (that state is a deterministic predecessor of $\fstate$).
We therefore set the reward along a path
$\path^\prime = \init \fstate_1 \fstate_2 \cdots$ for taking the initial
transition to $\fstate$ in $D^\prime$ to be the sum of the total rewards
accumulated along all distinct path prefixes of the form
$\init \nfpath \fstate$, normalised by the total probabilitiy of taking any
of these paths, where $\nfpath$ is a possibly empty sequence
of deterministic predecessors of $\fstate$, and where the total reward for
each prefix is weighted by the probability of taking the transitions
along that sequence, 
\begin{align}
\label{eq:rstructprimeinit}
\rtrans^\prime(\init, \fstate) = &
	\frac{
		\sum_{\path_{\mathit{pre}} \in \pref(\path^\prime)}
			\rtotal_{\rstruct}(\path_{\mathit{pre}})	
			\pmeasure^{D}\{\path_{\mathit{pre}}\}
	}
	{\pmeasure^{D^\prime}\{\init \fstate_1\}}
\end{align}
\end{itemize}

\begin{proof}[of Theorem~\ref{theorem:reductionreward}]
We want to show that for every reward structure $\rstruct$
for $D$ and corresponding reward structure
$\rstruct^\prime$ for $D^\prime$, 
every $\bowtie \ \in \{<,\leqslant,\geqslant,>\}$ and every
$r \in \reals$, if
$\init \models \pctlr_{\bowtie r}[\lsometime \ \synchlab]$
holds in $D$ then it also holds in $D^\prime$.
Let $\revar_{\sat({\pctlsometime \synchlab})}$ and
$\revar^\prime_{\sat({\pctlsometime \synchlab})}$ respectively
denote the random variables over $\paths^D(\init)$ and
$\paths^{D^\prime}(\init)$ whose expectations correspond to
$\rstruct$ and $\rstruct^\prime$.
From the semantics of PCTL over a DTMC we have
\begin{align*}
\init \models \pctlr_{\bowtie r}[\lsometime \ \synchlab]
	& \quad \Leftrightarrow 
\expect[\revar_{\sat({\synchlab})}] \bowtie r \\
%\init \models \pctlr_{\bowtie r}[\lsometime \ \synchlab]
	& \quad \Leftrightarrow 
\sum_{\path \in \paths} \revar_{\sat({\synchlab})}
\pmeasure^D_{\init}\{\path\} \bowtie r.
\end{align*}
Therefore, we need to show that
\begin{align}
	\label{eq:theorem2equationone}
	\sum_{\path \in \paths^D} \revar_{\sat({\synchlab})}(\path)
	\pmeasure^D\{\path\}
		=
	\sum_{\path^\prime \in \paths^{D^\prime}} \revar^\prime_{\sat({\synchlab})}(\path^\prime)
	\pmeasure^{D^\prime}\{\path^\prime\},
\end{align}
where $\pmeasure^D$ and $\pmeasure^{D^\prime}$
denote the probability measures with respect to the sets of infinite
paths from $\init$ in $D$ and $D^\prime$ respectively.
There are two cases:

Firstly, if there exists some path of $D$ that does not synchronise then by definition
$\revar_{\sat({\synchlab})} = \infty$. Also,
from Lemma~\ref{lem:reachabilityproperties} we know that there is a corresponding path
of $D^\prime$ that does not synchronise,
and hence that $\revar^\prime_{\sat({\synchlab})} = \infty$.
By definition the probability measure of all paths of $D$ and $D^\prime$
are strictly positive. Therefore, all summands
of Equation~\ref{eq:theorem2equationone} are defined,
and the expectation of both 
$\revar_{\sat({\synchlab})}$ and
$\revar^\prime_{\sat({\synchlab})}$ is $\infty$.
%\item If there exists some path
%$\path = \init, \nfpath_{1}, \fstate_1, \nfpath_{2}, \fstate_2, \ldots \in \paths^D$
%such that $\fstate_i \not \models \synchlab$ for all $i \ge 1$ ($\path$
%does not synchronise), then by definition
%$\revar_{\sat({\pctlsometime \synchlab})} = \infty$. Also,
%from Lemma~\ref{lem:reachabilityproperties} we know that there exists some corresponding path
%$\path^\prime = \init, \fstate_1, \fstate_2, \fstate_3, \ldots \in \paths^{D^\prime}(\init)$
%that does not synchronise, and hence that $\revar^\prime_{\sat({\pctlsometime \synchlab})} = \infty$.
%By definition the probability measure of all paths of $D$ and $D^\prime$
%are strictly positive. Therefore, all summands of Equation~\ref{eq:theorem2needtoshow} are defined,
%and the expectation of both 
%$\revar_{\sat({\pctlsometime \synchlab})}$ and
%$\revar^\prime_{\sat({\pctlsometime \synchlab})}$ is $\infty$.

Secondly, we consider the case where all possible paths of $D$ and $D^\prime$
synchronise. First we define the function
$\reduce : \paths^D \to \paths^{D^\prime}$ that
maps paths of $D$ to their corresponding path in the reduction $D^\prime$,
\begin{align*}
	\reduce(\state_0 \nfpath_1 \fstate_1 \nfpath_2 \fstate_2 \cdots) =
		\state_0 \fstate_1 \fstate_2 \cdots,
\end{align*}
where $\nfpath_i$ is the (possibly empty) sequence of deterministic
predecessors of the firing state $\fstate_i$.
%Recall that $\paths^{D}_{\fstate}(\init)$ denotes the set of all
%infinite paths of $D$ starting in $\init$, where the first firing state
%reached along that path is $\fstate$, and that all such sets for all firing states in
%$Q$ form a partition of $\paths^{D}(\init)$.
%Now let $\mathit{PrePaths}(\path)$ denote the set of all paths of the form
%$\init, \state^i, \state^{i+1}, \ldots, \state^j, \fstate_1, \ldots$ of $D$, where
%$0 \le i \le j$ and $\state^k \in \predstates(\fstate_1)$ for $i \le k \le j$.
%That is,
%\begin{align*}
%	\mathit{PrePaths}(\path) = \bigcup_{\state^i \in \predstates(\fstate_1)}
%		\{\init, \state^i, \state^{i+1}, \ldots, \state^{j}, \fstate_1, \ldots \}.
%\end{align*}
Let $\reduce^{-1}(\path)$ denote the preimage of $\path$ under $\reduce$.
Then, we can
rewrite the left side of (\ref{eq:theorem2equationone}) to
\begin{align*}
	\sum_{\path^\prime \in \paths^{D^\prime}}
	\sum_{\path \in \reduce^{-1}(\path^\prime)}
		\revar_{\sat({\synchlab})}(\path)
		\pmeasure^{D}\{\path\}.
\end{align*}

For any path $\path$ of $D$ or $D^\prime$ let $\synchprefix(\path)$
%be the first firing state along that path that is in the set
%$\sat(\synchlab)$.
be the prefix of that path whose last state is the first firing
state along that path that is in the set $\sat(\synchlab)$.
So we want to show that the following holds for any
path $\path^\prime$ of $D^\prime$, 
%for all paths $\path^\prime \in\paths^{D^\prime}$,
\begin{align}
		%\sum_{\path^\prime \in \paths^{D^\prime}}
		\sum_{\path \in \reduce^{-1}(\path^\prime)}
			\revar_{\sat({\synchlab})}(\path)
			\pmeasure^{D}\{\path\}
		& = 
		%\sum_{\path^\prime \in \paths^{D^\prime}}
		\revar^\prime_{\sat({\synchlab})}(\path^\prime)
			\pmeasure^{D^\prime}\{\path^\prime\} \nonumber \\
		\label{eq:expandvardef}
		%\sum_{\path^\prime \in \paths^{D^\prime}}			
		\sum_{\path \in \reduce^{-1}(\path^\prime)}
			\rtotal_{\rstruct}(\synchprefix(\path))
			\pmeasure^{D}\{\path\}
		& = 
		%\sum_{\path^\prime \in \paths^{D^\prime}}
		\rtotal_{\rstruct^\prime}(\synchprefix(\path^\prime))
			\pmeasure^{D^\prime}\{\path^\prime\}. 
\end{align}

Given some path $\path$ let $\path[i:j]$ denote the sequence
of states in $\path$ from the $i^{th}$ firing state to the
$j^{th}$ firing state along that path (inclusively). The notation
$\path[-:j]$ indicates that no states are removed from the start of
the path i.e. the first state is $\init$, and the notation
$\path[i:-]$ indicates that no states
are removed from the end of the path.
By recalling that
$\pmeasure(\state_0 \state_1 \cdots \state_n) =
	\prod_{i=1}^n \prmatrix(\state_{i-1},\state_i)$ 
we can see that
$\pmeasure(\state_0 \state_1 \cdots \state_n) =
	\pmeasure(\state_0 \cdots \state_i)
	\pmeasure(\state_i \cdots \state_n)$
for any $0 < i < n$.
Also from (\ref{eq:tot}) it is clear that for any reward structure $\rstruct$,
$\rtotal_\rstruct(\init \cdots \state_n) =
	\rtotal_\rstruct(\init \cdots \state_i) +
	\rtotal_\rstruct(\state_i \cdots \state_n)$
holds for all $0 < i < n$.
Now we can rewrite (\ref{eq:expandvardef}) to 
%\begin{align*}
%%	\label{eq:splitpmeasure}
%	\sum_{\path \in \reduce^{-1}(\path^\prime)} &
%		\rtotal_{\rstruct}(\synchprefix(\path))
%		\pmeasure^{D}\{\path[-:1]\}
%		\pmeasure^{D}\{\path[1:-]\}
%	= \\
%	&	
%		\rtotal_{\rstruct^\prime}(\synchprefix(\path^\prime))
%		\pmeasure^{D^\prime}\{\path^\prime[-:1]\}
%		\pmeasure^{D^\prime}\{\path^\prime[1:-]\}.
%\end{align*}
%From Lemma~\ref{lem:reachabilityproperties} we know that
%$\pmeasure^{D}\{\path[1:-]\} =
%	\pmeasure^{D^\prime}\{\path^\prime[1:-]\}$
%since only deterministic transitions are removed in $D^\prime$, and we can simplify to
%\begin{align*}
%%	\sum_{\path^\prime \in \paths^{D^\prime}} 
%	\sum_{\path \in \reduce^{-1}(\path^\prime)}
%		\rtotal_{\rstruct}(\synchprefix(\path))
%		\pmeasure^{D}\{\path[-:1]\}
%	= 
%%	\sum_{\path^\prime \in \paths^{D^\prime}}
%		\rtotal_{\rstruct^\prime}(\synchprefix(\path^\prime))
%		\pmeasure^{D^\prime}\{\path^\prime[-:1]\}.
%\end{align*}
\begin{align}
	\label{eq:splitprandtot}
	\begin{split}
		\sum_{\path \in \reduce^{-1}(\path^\prime)} &
			\left(
				\rtotal_{\rstruct}(\synchprefix(\path)[-:1]) + 
				\rtotal_{\rstruct}(\synchprefix(\path)[1:-])
			\right)		
			\pmeasure^{D}\{\path[-:1]\}
		= \\
		&
			\left(		
				\rtotal_{\rstruct^\prime}(\synchprefix(\path^\prime)[-:1]) + 
				\rtotal_{\rstruct^\prime}(\synchprefix(\path^\prime)[1:-])
			\right)		
			\pmeasure^{D^\prime}\{\path^\prime[-:1]\}.
	\end{split}
\end{align}

By the definition of $\rstruct^\prime$ we can write the right hand
side of (\ref{eq:splitprandtot}) as
\begin{align*}
%	\label{eq:rprimesubst}
	\begin{split}
		\left(
			\left(
				\frac
				{	
					\sum_{\path_{\mathit{pre}} \in \pref(\path^\prime)}
						\rtotal_{\rstruct}(\path_{\mathit{pre}})	
					\pmeasure^{D}\{\path_{\mathit{pre}}\}
				}
				{\pmeasure^{D^\prime}\{\path^\prime[-:1]\}}
			\right)
			+ 
			\rtotal_{\rstruct^\prime}(\synchprefix(\path^\prime)[1:-])
		\right)		
		\pmeasure^{D^\prime}\{\path^\prime[-:1]\}
		= \\
		\sum_{\path_{\mathit{pre}} \in \pref(\path^\prime)}
		\left(		
			\rtotal_{\rstruct}(\path_{\mathit{pre}})	
			\pmeasure^{D}\{\path_{\mathit{pre}}\}
		\right)
		+	
		\rtotal_{\rstruct^\prime}(\synchprefix(\path^\prime)[1:-])
		\pmeasure^{D^\prime}\{\path^\prime[-:1]\}.
	\end{split}
\end{align*}
From Lemma~\ref{lem:reachabilityproperties} we know that
\begin{align*}
	\pmeasure^{D^\prime}\{\path^\prime[-:1]\} = 
	\pmeasure^D \pref(\path^\prime) =
	\sum_{\path_{\mathit{pre}} \in \pref(\path^\prime)}
		\pmeasure^D\{\path_{\mathit{pre}}\},
\end{align*}
and hence obtain
\begin{align}
	\label{eq:factorise}
	\begin{split}
		\sum_{\path_{\mathit{pre}} \in \pref(\path^\prime)} &
		\left(		
			\rtotal_{\rstruct}(\path_{\mathit{pre}})	
			\pmeasure^{D}\{\path_{\mathit{pre}}\}
		\right)
		+
		\sum_{\path_{\mathit{pre}} \in \pref(\path^\prime)}
			\rtotal_{\rstruct^\prime}(\synchprefix(\path^\prime)[1:-])
		\pmeasure^{D}\{\path_{\mathit{pre}}\}
		= \\
%		\sum_{\path_{\mathit{pre}} \in \pref(\path^\prime)} &
%		\left(		
%			\rtotal_{\rstruct}(\path_{\mathit{pre}})	
%			\pmeasure^{D}\{\path_{\mathit{pre}}\}
%			+
%			\rtotal_{\rstruct^\prime}(\synchprefix(\path^\prime)[1:-])
%			\pmeasure^{D}\{\path_{\mathit{pre}}\}
%		\right)
%		= \\
		\sum_{\path_{\mathit{pre}} \in \pref(\path^\prime)} &
		\left(		
			\rtotal_{\rstruct}(\path_{\mathit{pre}})	
			+
			\rtotal_{\rstruct^\prime}(\synchprefix(\path^\prime)[1:-])
		\right)
		\pmeasure^{D}\{\path_{\mathit{pre}}\}.
	\end{split}
\end{align}

Since $\pref(\path^\prime)$ is the set of all possible finite prefixes from the
initial state $\init$ to the first firing state $\fstate_1$, and since
$\path[-:1] = \synchprefix(\path)[-:1]$ clearly holds, we know that
\begin{align*}
	\bigcup_{\path_{\mathit{pre}} \in \pref(\path^\prime)}
	\{\path_{\mathit{pre}}\}
	=
	\bigcup_{\path \in \reduce^{-1}(\path^\prime)}
	\{\path[-:1]\}.	
	=
	\bigcup_{\path \in \reduce^{-1}(\path^\prime)}
	\{\synchprefix(\path)[-:1]\}.
\end{align*}
Using this fact, and by observing that by definition
\begin{align*}
	\rtotal_{\rstruct^\prime}(\synchprefix(\path^\prime)[1:-]) =
	\rtotal_{\rstruct}(\synchprefix(\path)[1:-]),
\end{align*}
we can write (\ref{eq:factorise}) as
\begin{align*}
	\sum_{\path \in \reduce^{-1}(\path^\prime)} &
		\left(		
			\rtotal_{\rstruct}(\synchprefix(\path)[-:1]) +
			\rtotal_{\rstruct^\prime}(\synchprefix(\path^\prime)[1:-])
		\right)
		\pmeasure^{D}\{\path[-:1]\}.
\end{align*}
This is the same as the left hand side of (\ref{eq:splitprandtot}) and
the theorem is proved.
\qed
\end{proof}

%%% Local Variables: 
%%% mode: latex
%%% TeX-master: "sync_journal"
%%% End: 
\section{Connecting the Concrete Model and the Population Model}
\label{sec:connection}
In this section, we define the abstraction function to connect a concrete model with a
population model. To that end, let \(D_c = (\concstates, \concstate_0, \prmatrix_c)\) be a concrete
model of a network of \(N\) PCOs with a clock cycle length \(T\), a refractory period
 \(R\), a phase response function \(\pert\), a coupling
 \(\epsilon\) and broadcast failure probability of \(\failureprob\). Furthermore,
 let
 \(D_p = (\popstates, \popstate_0, \prmatrix_p)\) be the DTMC of a population model for the same
 parameters.
For simpler notation, we introduce some general notation for transitions in DTMCs.
If there is a possible transition between two states \(q\) and \(q^\prime\) in a
DTMC,
that is \(\prmatrix(q,q^\prime) > 0\), then we also write \(q \transition q^\prime\). Observe that
for this simplification, \(q\) and \(q^\prime\) are either both in \(\concstates\) or
both in \(\popstates\).
We also denote the reflexive, transitive closure of \(\transition\) by \(\transseq\).

\subsection{Proving the Correspondence between Concrete and Population Models}
\label{sec:correspondence} 
We need to associate states in \(D_c\) to states in \(D_p\).
In general, several concrete states will be mapped to a single
population state, since we do not distinguish between different orders
of oscillators in the latter, while we do in the former.

Furthermore, 
we want to abstract from different modes of the oscillators. However,
it is not sensible to associate all modes within a phase to the same
population state, since in the transitions from one mode to the
next the system chooses, whether an oscillator fails to broadcast
its pulse or not. If we want to be able to define a simulation relation,
we need  to represent the failures described by the
transitions in the population model.
To have an exact correspondence, 
we first collect all the concrete states where the counter and all oscillators are 
at the start mode into a single set.
\begin{align*}
\concstates^\prime &= \{\concstate \in \concstates \mid s = (\envstate, \netstate) \land \proj{\location}{\envstate} = \start \land \forall u \colon \proj{\location}{\netstate(u)} = \start\}  
\end{align*}

The abstraction function \(  \abstSing \colon \concstates^\prime \to \popstates\)
takes a concrete state \(\concstate\) and counts the number of oscillators sharing the
same phase, mapping \(\concstate = (\envstate, \netstate)\) to the corresponding state of the
population model,
\begin{align*}
  \abst{\concstate} & = \seq{ \size{\{u \mid \proj{\Phi}{\netstate(u)} = 1\}},
  		\dots,  \size{\{u \mid \proj{\Phi}{\netstate(u)} = T\}} }.
\end{align*}
To show that this abstraction is sensibly defined, we need to show that 
the concrete model can weakly simulate the transitions allowed by the population model, and
vice versa.
That is, if the abstraction \(\popstate_1\) of a concrete state \(\concstate_1\)
allows a transition to another population state \(\popstate_2\), then there is
a sequence of transitions from \(\concstate_1\) leading to \(\concstate_2\), whose
abstraction is \(\popstate_2\). Furthermore, if there is a transition
sequence from one concrete state \(\concstate_1\) to \(\concstate_2\),
where both statescan be abstracted
to  population states \(\popstate_1\) and \(\popstate_2\), respectively,
then there is also a sequence of transitions connecting \(\popstate_1\) with
\(\popstate_2\). This situation is visualised in Fig.~\ref{fig:abstract-simulate}.

\begin{figure}
\begin{center}
\begin{tikzpicture}
  \node (c1) {\(\concstate_1\)};
  \node[below =1cm of c1] (c2) {\(\concstate_2\)};
  \node[left =1cm of c1] (p1) {\(\popstate_1\)};
  \node[below =1cm of p1] (p2) {\(\popstate_2\)};

  \draw[->,double,double distance=2pt, >=implies] (p1) to (p2);
  \draw[->, double,double distance=2pt, >=implies] (c1) to (c2);
  \draw[-> ] (c1) to node[above] {\(\abstSing\)} (p1);
  \draw[->] (c2) to node[above] {\(\abstSing\)} (p2);

\end{tikzpicture}
\end{center}
\caption{Weak Simulation Relation of Concrete States and Populations}
\label{fig:abstract-simulate}
\end{figure}
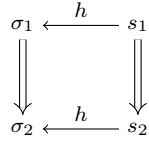

For the first direction, we actually show this condition for a single transition
in the population model. However, this result can be straightforwardly extended to transition sequences.
\begin{lemma}
  \label{lem:conc-sim-pop}
Let \(\concstate_1 \in \concstates^\prime\) and \(\popstate_1, \popstate_2 \in \popstates\) such that \(\abst{\concstate_1} = \popstate_1\) and
 \(\popstate_1 \poptrans \popstate_2\). Then there is a \(\concstate_2\in\concstates^\prime\) such that \(\concstate_1 \conctransseq \concstate_2\)
and \(\abst{\concstate_2} = \popstate_2\). Furthermore, the sum of the probabilities of transition sequences from \(\concstate_1\) 
to an instantiation \(\concstate_2\)
of \(\popstate_2\) is equal to the probability of the transition from \(\popstate_1\) to \(\popstate_2\).
\end{lemma}
\begin{proof}
  Let \(\concstate_1 = (\envstate_1, \netstate_1) \in \concstates^\prime\), i.e.,  a concrete state where \(\proj{\location}{\envstate_1} = \start\) and
  \(\proj{\location}{\netstate_1(u)} = \start\) for all \(1 \leq u \leq N\). 
Then \(\abst{\concstate_1} = \seq{ \size{\{u \mid \proj{\phase}{\netstate_1(u)} = 1\}}, \dots,  \size{\{u \mid \proj{\phase}{\netstate_1(u)} = T\}} } = \popstate_1\).
Let \(\popstate_1 \poptrans \popstate_2\). Note that there is only a single outgoing transition from \(\concstate_1\) according to
condition~(\ref{trans:env_reset}). That is, in the successor state of \(\concstate_1\), we have \(\proj{c}{\envstate} = 0\) and
\(\proj{\location}{\envstate} = \update\), while the oscillator states are not changed. To keep the notation tidy, we identify this successor state
with \(\concstate_1\) in the following.

Now consider two cases. If \(\size{\{u \mid \proj{\phase}{\netstate_1(u)} = T\}} = 0\), 
then \(\popstate_2 =  \seq{0, \size{\{u \mid \proj{\phase}{\netstate_1(u)} = 1\}}, \dots,  \size{\{u \mid \proj{\phase}{\netstate_1(u)} = T-1\}} }\), 
since no set of oscillators in \(\popstate_1\) is perturbed by a firing oscillator. In particular, there is no \(u\) such that
\(\proj{\phase}{\netstate_1(u)} = T\). Hence, for all possible successors \(\concstate^\prime\) of \(\concstate_1\), we have
that only condition~(\ref{trans:nothing_fires_add}) is satisfied. Furthermore, this is the case until all oscillators changed their
mode to \(\update\). Let us call this state \(\concstate_1^\prime = (\envstate_1^\prime, \netstate_1^\prime)\). 
Now, the environment was not changed from \(\concstate_1\) to \(\concstate_1^\prime\), i.e., 
\(\proj{c}{\envstate} = \proj{c}{\envstate^\prime} = 0\).

Hence, condition~(\ref{trans:update_in_cycle_no_pulse_add}) is satisfied for all oscillators.  
Since \(\pert(\Phase,  0, \epsilon) = 0\) for all \(\Phase\), the phase of
each oscillator is increased by one. 
This implies that there  is a single successor of \(\concstate_1^\prime\), which we call
\(\concstate_2 = (\envstate_2, \netstate_2)\) and that  
 \(\proj{\phase}{\netstate_2(u)} = \proj{\phase}{\netstate_1^\prime(u)} + 1\).  In particular,
we have that for all oscillators \(u\), \(\proj{\phase}{\netstate_2(u)} > 0\), and 
\(\{u \mid \proj{\phase}{\netstate_2(u)} = \Phase\} = \{u \mid \proj{\phase}{\netstate^\prime_1(u)} = \Phase - 1\}\)
for all \(0 < \Phase  \leq T\). Hence \(\concstate_2\) is the required state of the concrete model.
 
 Now let \(\size{\{u \mid \proj{\phase}{\netstate_1(u)} = T\}} > 0\). Then, each transition in the population
 model is induced by a failure vector \(F=\langle f_1, \dots, f_T\rangle\). In particular, there is a maximal number \(k\), such 
 that for all \(l < k\), we have \(f_{l} = \star\). That is, \(k\) denotes the lowest phase in which oscillators possibly
 fire.

 First, we introduce some notation, where \(\Phase > R\) and \(\concstate = (\envstate, \netstate)\). 
 \begin{align*}
   \concstatesFire{\concstate}& =  \{u \mid \proj{\phase}{\netstate(u)} = T \}\\
%   \concstatesPert{\concstate}&= \outref{s} \cap \{i \mid \proj{\phase}{\concstate(i)} < T \land \pert(\proj{\phase}{\concstate(i)}, \alpha^{\proj{\phase}{\concstate(i)}}(\popstate_1, F), \epsilon) +1 \leq  T \}\\
%   \concstatesPertFire{\concstate}&= \outref{s} \cap \{i \mid \proj{\phase}{\concstate(i)} < T \land \pert(\proj{\phase}{\concstate(i)},\alpha^{\proj{\phase}{\concstate(i)}}(\popstate_1, F), \epsilon) +1 > T \}\\
   \concstatesPertFirePhase{\concstate}{\Phase} & = %\outref{s} \cap
                                                  \{u \mid \proj{\phase}{\netstate(u)} = \Phase \land \pert(\proj{\phase}{\netstate(u)},\alpha^{\proj{\phase}{\netstate(u)}}(\popstate_1, F), \epsilon) +1 > T \} \\
   \concstatesPertPhase{\concstate}{\Phase}&= \{u \mid \proj{\phase}{\netstate(u)} = \Phase \land \pert(\proj{\phase}{\netstate(u)}, \alpha^{\proj{\phase}{\netstate(u)}}(\popstate_1, F), \epsilon) +1 \leq  T \}
 \end{align*}
 That is, \(\concstatesFire{\concstate}\) denotes the set of oscillators possibly firing in \(\concstate\). The
 sets \(\concstatesPert{\concstate}\) and \(\concstatesPertFire{\concstate}\) denote the sets of oscillators being perturbed
  but not firing (since the perturbation is not sufficient for the oscillators to reach the end of their cycle), and
 possibly firing, respectively. We can only say that elements of \(\concstatesFire{\concstate}\) and \(\concstatesPertFire{\concstate}\) 
 \emph{possibly} fire, since they may be affected by a broadcast failure.
%\slcomment{maybe move this notation outside of the proof, if we need it somewhere else}

We now have to construct a sequence of transitions, where we draw the firing oscillators from the sets \(\concstatesFire{\concstate_1}\) and
\(\concstatesPertFirePhase{\concstate_1}{\Phase}\), according to the broadcast failure vector \(F\). Furthermore, all elements of \(\concstatesFire{\concstate_1}\)
and the sets \(\concstatesPertFirePhase{\concstate_1}{\Phase}\) have to take transitions such that their phase value in the next
iteration is \(1\).

%  We have the following facts.
%  \begin{fact}
%    \(\sum_{\phase = k}^T b_\phase \leq \size{\concstatesFire{\concstate_1} \cup \concstatesPertFire{\concstate_1}}\).
%  \end{fact}
%  \begin{fact}
%  \( \outref{\concstate_1} \subseteq \concstatesFire{\concstate_1} \cup \concstatesPert{\concstate_1} \cup \concstatesPertFire{\concstate_1}   \) 
%  \end{fact}
%  \begin{fact}
%    The sets \(\concstatesFire{\concstate_1}\), \(\concstatesPert{\concstate_1}\) and \(\concstatesPertFire{\concstate_1}\) are all
%  mutually disjoint.
%  \end{fact}
%  \begin{fact}
%    For each oscillator \(u \in \concstatesPertFire{\concstate}\), and each oscillator \(v \in \concstatesPert{\concstate}\), 
%  we have that \(\proj{\phase}{\concstate(u)} > \proj{\phase}{\concstate(v)}\).
%  \end{fact}
% \begin{fact}
% For a failure vector \(F\), in which \(k\) is the minimal phase in which oscillators are pertubed to fire, we have \(k > R\). 
% \end{fact}
%  \begin{fact}
% \label{fact:k_greater_R}
%    For each oscillator \(u \in \concstatesPertFire{\concstate}\), we have \(\proj{\Phase}{\concstate_1(u)} \geq k\), where \(k\)
% is the minimal phase in which oscillators are perturbed to fire.
%  \end{fact}

 Let \(\popstate_1 = \langle k_1, k_2, \dots, k_T\rangle\).
 Now consider an arbitrary sequence \(u_1, \dots, u_{k_T}\) of all \(k_{T}\)  elements from \(\concstatesFire{\concstate_1}\).
Additionally, let \(\failstates{T} \subseteq \concstatesFire{\concstate_1} \) be the set of oscillators in phase \(T\) with a broadcast failure, i.e.,
\(|\failstates{T}| = f_T\). Observe that \(\proj{\phase}{\netstate_1(u_j)} = T\) for all \(1 \leq j \leq k_T\).
 Furthermore, let \(r_0 = \concstate_1\). Then we define a sequence of successors of \(r_0= (\envstate_0, \netstate_0)\) as follows,
where \(1 \leq j \leq k_T\). If \(k_i \not\in \failstates{T}\), then
% where \(1 \leq j \leq k_T - f_T\).
 \begin{align*}
   \netstate_j & = \override{\netstate_{j-1}}{u_j}{(T, \update)}\\
   \envstate_j & = (\update, \proj{c}{\envstate_{j-1}} + 1)
 \end{align*}
otherwise
 \begin{align*}
   \netstate_j & = \override{\netstate_{j-1}}{u_j}{(T, \update)}\\
   \envstate_j & = (\update, \proj{c}{\envstate_{j-1}})
 \end{align*}
 Observe %that we exhausted \(\concstatesFire{\concstate_1}\) with these two sequences, and 
that these states
define a sequence of transitions from \(r_0\) to \(r_{k_{T}}\) according to conditions (\ref{trans:fire_add}) and (\ref{trans:fire_failure_add}).

Now, for each phase \(\Phase\), with \(k \leq \Phase < T\), we proceed similarly. That is, we first choose a sequence
 \(u^\Phase_1, \dots, u^\Phase_{k_\Phase} \) of oscillators and a set \(\failstates{\Phase} \subseteq\concstatesPertFirePhase{\concstate_1}{\Phase}\) with
\(|\failstates{\Phase}| = f_\Phase\). 

Subsequently, we define each \(r^\Phase_j\) to be 
 \begin{align*}
   \netstate^\Phase_j & = \override{\netstate_{j-1}}{u^\Phase_j}{(\Phase, \update)}\\
   \envstate^\Phase_j & = \left\{\begin{array}{ll} 
                     (\update, \proj{c}{\envstate^\Phase_{j-1}} + 1) & \text{, if } j \not\in \failstates{\Phase}\\
                     (\update, \proj{c}{\envstate^\Phase_{j-1}}) & \text{, otherwise}
                   \end{array} \right.                                                      
 \end{align*}
where \(r^\Phase_0 = r^{\Phase + 1}_{k_\Phase}\). Observe again, that these sequences exhaust \(\concstatesPertFirePhase{\concstate_1}{\Phase}\) for
each phase \(\Phase\). Furthermore, we claim that the number of firing oscillators that are not inhibited by a broadcast failure in the concrete model
coincides with the number of perceived firing oscillators in the population model in this phase.
\begin{claim}
\label{clm:flashing}
  For each \(\Phase\) with \(k \leq \Phase \leq T\), we have \(\proj{c}{\envstate^\Phase_0} = \alpha^{\Phase}(\popstate_1, F)\).
\end{claim}
\begin{proof}
For \(\Phase = T\), we have \(\proj{c}{\envstate^T_0} = 0 = \alpha^T(\popstate_1,F)\). Now let \(\Phase < T\) and assume
\(\proj{c}{\envstate^{\Phase+1}_0} = \alpha^{\Phase+1}(\popstate_1,F)\). 
By definition, we have 
\begin{align*}
\alpha^{\Phase}(\popstate_1, F)&= \alpha^{\Phase + 1}(\popstate,F) +k_{\Phase +1} - f_{\Phase+1}, 
\end{align*}
since \(\Phase < T\) and \(f_{\Phase+1} \neq \star\). 
%By assumption, that is
%\begin{align*}
%\alpha^{\Phase}(\popstate_1, F)& = \proj{c}{\envstate^{\Phase+1}_0} +k_{\Phase +1} - f_{\Phase+1}. 
%\end{align*}
Now, in the sequence  \(r^{\Phase+1}_0, \dots, r^{\Phase+1}_{k_{\Phase+1}}\),
we increase \(\proj{c}{\envstate^{\Phase+1}_0}\) exactly \(k_{\Phase+1} - f_{\Phase+1}\) times, i.e,
\begin{align*}
\proj{c}{\envstate^\Phase}_0 & = \proj{c}{\envstate^{\Phase+1}_0} +k_{\Phase +1} - f_{\Phase+1}. 
\end{align*}
By assumption, we then get 
\begin{align*}
\proj{c}{\envstate^\Phase_0} & = \alpha^{\Phase + 1}(\popstate_1,F) +k_{\Phase +1} - f_{\Phase+1} = \alpha^{\Phase}(\popstate_1,F). 
\end{align*}
This proves the claim. \qed  
\end{proof}

This claim in particular states that the perturbation within the population model and the
concrete model is the same.

Since \(\pert\) is a monotonically increasing function in \(\alpha\), every oscillator in \(\concstatesPertFirePhase{\concstate_1}{\Phase}\)
is still perturbed to firing after other oscillators in the same phase fired. Hence, for 
each pair of states \(u^\Phase_{j-1}\) and \(u^\Phase_j\) with \(1 \leq j \leq k_\Phase - f_\Phase\),
a transition according to condition (\ref{trans:perturbed_fire_add}) is well-defined. Similarly,
for 
 oscillators that should fire, but are affected by a broadcast failure, \(u^\Phase_{j-1}\) and \(u^\Phase_j\) with \( k_\Phase - f_\Phase+1 \leq j \leq k_\Phase \),
the transition is defined according to condition (\ref{trans:perturbed_fire_failure_add}). 

Now, for every \(\Phase < k\), we know that \(\Phase + \pert(\Phase, \alpha^{\Phase}(\popstate_1, F), \epsilon)+1 \leq T\) and
\(\alpha^{\Phase -1}(\popstate_1, F) = \alpha^{\Phase}(\popstate_1,F)\), according to 
equation~(\ref{eq:alpha}).  Hence, for every phase \(\Phase < k\), we arbitrarily enumerate the oscillators of \(\concstatesPertPhase{\concstate_1}{\Phase} =
u^\Phase_1, \dots, u^\Phase_{k_\Phase}\) and define the following sequence of states \(r_j^\Phase\), where \(r_0^\Phase = r^{\Phase+1}_{k_{\Phase+1}}\).
 \begin{align*}
   \netstate^\Phase_j & = \override{\netstate_{j-1}}{u^\Phase_j}{(\Phase, \update)}\\
   \envstate^\Phase_j & = (\update, \proj{c}{\envstate^\Phase_{j-1}}) 
 \end{align*}
For each \(\Phase\) and  pair of states \(r^\Phase_j\) and \(r^\Phase_{j+1}\), there is a transition according
to condition (\ref{trans:perturbed_but_not_fire_add}). So, all in all, we have a sequence of transitions
from \(\concstate_1\) to \(r^0_{k_0}\).

Now, in \(r^0_{k_0} = ( \envstate^0_{k_0}, \netstate^0_{k_0})\), we have that 
 \(\proj{\location}{\envstate^0_{k_0}} = \update\) and for all \(u\), \(\proj{\location}{\netstate^0_{k_0}(u)} = \update\). 
Then let \(\concstate_2 = ( \envstate, \netstate)\) be defined by the following formulas. 
\begin{align}
&\proj{\location}{\envstate} = \start \text{ and } \proj{c}{\envstate} = \proj{c}{\envstate^0_{k_0}} \label{eq:reset_env}\\
&\forall u \colon \proj{\location}{\netstate(u)} = \start \label{eq:reset_location}\\
&\forall u \colon \proj{\phase}{\netstate^0_{k_0}(u)} = T \implies \proj{\phase}{\netstate(u)} = 1 \label{eq:flashing_osc}\\
&\forall u \colon \proj{\phase}{\netstate^0_{k_0}(u)} < T \land \proj{\phase}{\netstate^0_{k_0}(u)} \leq R  \implies \proj{\phase}{\netstate(u)} = \proj{\phase}{\netstate^0_{k_0}(u)} +1 
\label{eq:refractory_osc}\\
&\forall u \colon \proj{\phase}{\netstate^0_{_0}(u)} < T \land \proj{\phase}{\netstate^0_{k_0}(u)} > R\; \land \nonumber \\
&\hspace{1cm}\proj{\phase}{\netstate^0_{k_0}(u)} + \pert(\proj{\phase}{\netstate^0_{k_0}(u)}, \proj{c}{\envstate^0_{k_0}}, \epsilon) + 1 \leq T \nonumber \\
& \hspace{2cm} \implies \proj{\phase}{\netstate(u)} = \proj{\phase}{\netstate^0_{k_0}(u)}  + \pert(\proj{\phase}{\netstate^0_{k_0}(u)}, \proj{c}{\envstate^0_{k_0}}, \epsilon)  +1 \label{eq:perturbed_not_flashing_osc}\\ 
&\forall u \colon \proj{\phase}{\netstate^0_{k_0}(u)} < T \land \proj{\phase}{\netstate^0_{k_0}(u)} > R\; \land \nonumber\\
&\hspace{1cm}\proj{\phase}{\netstate^0_{k_0}(u)} + \pert(\proj{\phase}{\netstate^0_{k_0}(u)}, \proj{c}{\envstate^0_{k_0}},\epsilon) + 1 > T \nonumber\\
& \hspace{2cm} \implies \proj{\phase}{\netstate(u)} = 1  \label{eq:perturbed_flashing_osc}
\end{align}
Then \(r^0_{k_0}\) and \(\concstate_2\) satisfy all parts of condition (\ref{trans:update_add}). Hence, we have a sequence
of transitions from \(\concstate_1\) to \(\concstate_2\).
To prove \(h(\concstate_2) = \popstate_2\), we need to show that the number of oscillators possessing
a phase \(\Phase\) in \(\concstate_2\) matches the \(\Phase\)-th entry of \(\popstate_2 = \langle k^\prime_1, \dots, k^\prime_T\rangle \). To that end, recall
that by Def.~\ref{def:fsucc}, each \(k^\prime_\Phase =\sum_{\PhaseTwo \in \phiupdate_\Phase(\state, \fvec)} k_\PhaseTwo \), where \(\popstate_1 = \langle k_1, \dots, k_T\rangle\)
and \(	\phiupdate_\Phase(\state, \fvec) = \{\PhaseTwo \mid \PhaseTwo \in \{1,\ldots,T\} \land \ptf(\state, \PhaseTwo, \fvec) = \Phase\} \). 
Observe that both the concrete model and the population model use the same
perturbation function \(\pert\) and that \(\ptf\) is defined in terms of \(\pert\).
In particular, we have
\begin{align*}
  	\ptf(\state, \Phase, \fvec) =
	\begin{cases}
		1 & \text{if } \fire^\Phase(\state, \fvec) \\
		\update^\Phase(\state, \fvec) & \text{otherwise}. \\
	\end{cases}
\end{align*}
Now let us distinguish three cases for \(\Phase\). In the following, let 
 \(\concstate_1 = (\envstate_1, \netstate_1) \), \(\concstate_2 = (\envstate_2, \netstate_2)\), 
\(\popstate_1 = \langle k_1, \dots, k_T \rangle\) and \(\popstate_2 = \langle k^\prime_1, \dots, k^\prime_T \rangle\).
\begin{enumerate}
\item If \(\Phase \leq R\), then \(\update^\Phase(\popstate_1, \fvec) = \Phase + 1 \), due to the definition
of the refractory function \(\refr\). Similarly, for all \(u\) such that \( \proj{\phase}{\netstate_1(u)} = \Phase\), we get that 
\(\proj{\phase}{\netstate_2(u)} = \Phase +1 \).  Hence, for all \(\Phi \leq R\), we have that 
\(|\{u \mid \proj{\phase}{\netstate_2(u)} = \Phase+1\}| = |\{u \mid \proj{\phase}{\netstate_1(u)} = \Phase\}|\).  

\item If \(\Phase > R\) and \(\update^\Phase(\popstate_1, \fvec) = \PhaseTwo\), with \(\PhaseTwo \leq T\). Then
\(\proj{\phase}{\netstate_2(u)} = \PhaseTwo\), by formula~(\ref{eq:perturbed_not_flashing_osc}). Observe 
that the number of oscillators in \(\concstate_1\) with a phase of \(\Phase\) is \(k_\Phase\).
So, the number of oscillators that get perturbed to be in \(\PhaseTwo\) is the union 
of the oscillators \(u\) in phases \(\Phase\),  where \(\pert(\Phase,\proj{c}{\envstate_2}, \epsilon)+1 = \PhaseTwo\).
That is, \(\{u \mid \proj{\Phase}{\netstate_2(u)} = \PhaseTwo\} = \{u \mid \pert(\proj{\phase}{\netstate_1(u)},\proj{c}{\envstate_2}, \epsilon)+1 =\PhaseTwo \}\).
By the definition of \(\ptf\) and claim~\ref{clm:flashing}, we get that \(\ptf(\popstate_1, \Phase, \fvec) = \PhaseTwo\). 
That is, for a specific \(\PhaseTwo\), we have that the phases \(\Phase\) of oscillators perturbed to \(\PhaseTwo\) are in \(\phiupdate_\PhaseTwo(\popstate_1,\fvec)\).
Hence, since the sets of oscillators in each phase are disjoint, 
\(|\{u \mid \proj{\phase}{\netstate_2(u)} = \PhaseTwo\}| = \sum_{\Phase \in \phiupdate_\PhaseTwo(\popstate_1, \fvec)}k_\Phase\). % \slcomment{this not sufficient! it lacks an argument}

\item Finally, let \(\update^\Phase(\popstate_1, \fvec) = \PhaseTwo\) and \(\PhaseTwo > T\). Then \(\ptf(\popstate_1, \Phase,\fvec) = 1\). Furthermore,
by formulas (\ref{eq:flashing_osc}) and (\ref{eq:perturbed_flashing_osc}), we have \(\proj{\phase}{\netstate_2(u)} = 1\) for
all \(u\) with phase \(\Phase\). With similar reasoning as above, we get that 
\(|\{u \mid \proj{\phase}{\netstate_2(u)} = 1\}| = \sum_{\Phase \in \phiupdate_1(\popstate_1, \fvec)}k_\Phase\). 
\end{enumerate}
Hence, we get 
\(h(\concstate_2) = \popstate_2\), and we are done.\qed
\end{proof}

Now we turn our attention to the other direction. That is, if we have a sequence of transitions
in the concrete model, we can find a corresponding transition sequence in the population model.

\begin{lemma}
  \label{lem:pop-sim-conc}
  Let \(D_c = (\concstates, \concstate_0, \prmatrix_c)\) be a concrete network of oscillators and 
  \(D_p = (\popstates, \popstate_0, \prmatrix_p)\) be its abstraction as a population model. Furthermore,
let \(\concstate_1,\concstate_2 \in \concstates^\prime\) and \(\popstate_1 \in \popstates\) such that \(\abst{\concstate_1} = \popstate_1\) and
\(\concstate_1 \transseq \concstate_2\).
Then there is a \(\popstate_2\in\popstates\) such that
\(\abst{\concstate_2} = \popstate_2\) and \(\popstate_1 \transseq \popstate_2\).
\end{lemma}
\begin{proof}
  If \(\concstate_1 = \concstate_2\), then the lemma holds trivially. Otherwise,
  assume that for all \(\concstate^\prime = (\envstate, \netstate)\) different
  from \(\concstate_1\) and \(\concstate_2\), such that \(\concstate_1 \conctransseq s^\prime \conctransseq \concstate_2\),
  we have \(\proj{\location}{\envstate} = \update\). Furthermore, let \(\concstate_1 = (\envstate_1, \netstate_1 )\) and \(\concstate_2 = (\envstate_2, \netstate_2 )\).  
  By definition of the abstraction function, we have
  \begin{align*}
    \popstate_1 & = \langle \size{\{u \mid \proj{\phase}{\netstate_1(u) = 1}\}}, \dots, \size{\{u \mid \proj{\phase}{\netstate_1(u) = T}\}}\rangle\\
    \popstate_2 & = \langle \size{\{u \mid \proj{\phase}{\netstate_2(u) = 1}\}}, \dots, \size{\{u \mid \proj{\phase}{\netstate_2(u) = T}\}}\rangle
  \end{align*}
  We now distinguish two cases. First, assume that \(\{u \mid \proj{\phase}{\netstate_1(u) = T}\} = \emptyset\), and let \(\concstate = (\envstate_s, \netstate_s)\) be such that
  \(\concstate_1 \transseq \concstate\) and \(\proj{\location}{\netstate_s(u)} = \update\) for all \(u\). Then there is exactly one transition
  \(\concstate \transition \concstate_2\), which is defined according to equation (\ref{trans:update_add}). Furthermore, due to the assumption that no oscillator fires,
  we have \(\proj{c}{\envstate_s} = 0\), which implies \(\pert(\Phase, \proj{c}{\envstate_s}, \epsilon)=0\) for all \(\Phase\) by Definition~\ref{def:pert}. Hence, for all \(u\), we have
  \(\proj{\phase}{\netstate_2(u)} = \proj{\phase}{\netstate_s(u)} + 1 = \proj{\phase}{\netstate_1(u)}+1\). That is,
  \begin{align*}
    \popstate_2 & =\langle 0, \size{\{u \mid \proj{\phase}{\netstate_1(u) +1 = 2}\}}, \dots,\size{\{u \mid \proj{\phase}{\netstate_1(u) +1 = T}\}}\rangle\\
    &= \langle 0, \size{\{u \mid \proj{\phase}{\netstate_1(u) = 1}\}}, \dots,\size{\{u \mid \proj{\phase}{\netstate_1(u)  = T -1}\}}\rangle \enspace .
  \end{align*}
  That is, we have \(\popstate_1 \transition \popstate_2\) due to a deterministic transition, which, in
  particular, implies \(\popstate_1 \transseq \popstate_2\).

  The second case is more involved. Let us assume \(\{u \mid \proj{\phase}{\netstate_1(u) = T}\} \neq \emptyset\), that is, at least one oscillator fires. Hence, due to the preconditions of the
  transitions, we can divide the transition sequence from \(\concstate_1\) to \(\concstate_2\) as follows:
  \begin{align*}
    \concstate_1 \conctransseq r_T \conctransseq r_{T-1} \conctransseq \dots \conctransseq r_1 \conctransseq \concstate_2 \enspace, 
  \end{align*}
  where \(r_\Phase = (\envstate_{r\Phase}, \netstate_{r\Phase})\) denotes the state where all oscillators with
  phase \(\Phase\) changed their mode to \(\update\). Our goal now is to find a broadcast
  failure vector \(\fvec\), such that \(\suc(\popstate_1, \fvec) = \popstate_2\). To that end,
  let
  \begin{align*}
    f_\Phase & = \size{\{ u \mid \proj{\phase}{\netstate_1(u)} \geq \Phase\}} - \left( \proj{c}{\envstate_{r\Phase}} + \sum_{\PhaseTwo = \Phase+1}^{T} f_\PhaseTwo \right)
  \end{align*}
  for all \(\Phase\) where \(\Phase + \pert(\Phase, \proj{c}{\envstate_{r\Phase}}, \epsilon) +1 \ge T\). For the remaining \(\Phase\), set \(f_\Phase = \star\). Then \(\fvec = \langle f_1, \dots, f_T\ \rangle\). With this broadcast failure vector at hand, we now have to show that
  \begin{align*}
    \sum_{\PhaseTwo \in \phiupdate_\Phase(\popstate_1, \fvec) }\size{\{u \mid \proj{\phase}{\netstate_1(u)} = \PhaseTwo\}} = \size{\{u \mid \proj{\phase}{\netstate_2(u)} = \Phase\}} \enspace . 
  \end{align*}
  Recall that \(\phiupdate_\Phase(\popstate_1, \fvec) = \{\PhaseTwo \mid \ptf(\popstate_1, \PhaseTwo, \fvec) = \Phase\}\). Together with the condition we want to prove,
  this implies, that we  need to show \(\proj{\phase}{\netstate_2(u)} =
  \ptf(\popstate_1, \proj{\phase}{\netstate_1(u)}, \fvec)\) for all oscillators \(u\).
  We now need again to distinguish several cases, according to the different cases
  of the transition defined by condition (\ref{trans:update_add}).

  First, let \(u\) be such that \(\proj{\phase}{\netstate_1(u)} \leq R\), i.e., oscillator
  \(u\) is within its refractory period. If \(\proj{\phase}{\netstate_1(u)} = T\), then we have
  \begin{align*}
    \proj{\phase}{\netstate_2(u)} &= 1 &  \{\text{Cond.}~(\ref{trans:update_end_cycle_add})\}\\
    &=\ptf(\popstate_1, \proj{\phase}{\netstate_1(u)}, \fvec) & \{\text{Def.}~\ref{def:taufailure}\}
    \end{align*}
  Otherwise, if \(\proj{\phase}{\netstate_1(u)} < T\), we have
  \begin{align*}
    \proj{\phase}{\netstate_2(u)} &= \proj{\phase}{\netstate_1(u)} + 1 &  \{\text{Cond.}~(\ref{trans:update_in_cycle_no_pulse_add})\}\\
    &=\ptf(\popstate_1, \proj{\phase}{\netstate_1(u)}, \fvec) & \{\text{Def.}~\ref{def:taufailure}\}
  \end{align*}
  Now assume that \(\proj{\phase}{\netstate_1(u)} \ge R\), i.e., oscillator
  \(u\) is outside of its refractory period and thus will be perturbed by firing oscillators.
  If \(\proj{\phase}{\netstate_1(u)} = T\), then we proceed as in the previous case. So, let us
  assume \(\proj{\phase}{\netstate_1(u)} < T\). To show that the transition function of
  the population model coincides with the result within the concrete model, we need to ensure
  that the perceived firing oscillators are equal in both models for each oscillator.

  \begin{claim}
    Within each phase, the perceived oscillators in the population model coincide
    with the oscillators that fired up to the next higher phase in the concrete model.
    Formally, for each \(1 \leq \Phase < T\), we have 
    \(\proj{c}{\envstate_{r\Phase+1}} = \alpha^\Phase(\popstate_1,\fvec)\).
  \end{claim}
  \begin{proof}
    In the following, we use the notation \(\popstate_1 = \langle k_1, \dots, k_T\rangle\).
    Let \(f_\Phase \neq \star\). Then for \(\Phase = T-1\), we have
    \begin{align*}
      \proj{c}{\envstate_{rT}} &= \size{\{ u \mid \proj{\phase}{\netstate_1(u)} = T\}} - f_T & \{\text{Def. of } f_\Phase\}\\
                               & = 0 + k_T - f_T& \{\text{Def. of } \popstate_1\}\\
                               & = \alpha^T(\popstate_1,\fvec) + k_T - f_T& \{\text{Def. of } \alpha^\Phase\text{, Eq.}~(\ref{eq:alpha})\}\\
      & = \alpha^{T-1}(\popstate_1, \fvec) &
    \end{align*}
    Now assume that \(\proj{c}{\envstate_{r\Phase+1}} = \alpha^\Phase(\popstate_1,\fvec)\) for some
    \(1 < \Phase < T\) with \(f_\Phase \neq \star\). Then
    \begin{align*}      
      \proj{c}{\envstate_{r\Phase}}  & = \size{\{ u \mid \proj{\phase}{\netstate_1(u)} \geq \Phase\}} - \left(  \sum_{\PhaseTwo = \Phase}^{T} f_\PhaseTwo \right) & \{\text{Def of } f_\Phase \}\\
                                     & = \left( \sum_{\PhaseTwo = \Phase}^{T}k_\PhaseTwo\right) - \left(  \sum_{\PhaseTwo = \Phase}^{T} f_\PhaseTwo \right) &  \{\text{Def. of } \popstate_1\}\\
                                     &=  \left( \sum_{\PhaseTwo = \Phase+1}^{T}k_\PhaseTwo\right) - \left(  \sum_{\PhaseTwo = \Phase+1}^{T} f_\PhaseTwo \right) + k_\Phase - f_\Phase &\\
                                     & = \proj{c}{\envstate_{r\Phase+1}}  + k_\Phase - f_\Phase & \{\text{Def. of }f_\Phase\}\\
                                     & = \alpha^\Phase(\popstate_1, \fvec) + k_\Phase - f_\Phase & \{\text{Ass.}\}\\
      & = \alpha^{\Phase -1}(\popstate_1, \fvec) & \{\text{Def. of }\alpha^\Phase, \text{Eq.}~(\ref{eq:alpha})\}
    \end{align*}
    If \(f_\Phase = \star\), then the claim immediately holds.\qed
  \end{proof}
  Furthermore, observe that 
  for all \(\Phase\), if \(f_\Phase = \star\), then we have \(\proj{c}{\envstate_{r\Phase+1}} = \proj{c}{\envstate_{r1}}\). 
  We can now proceed to prove the final two cases.

  First, let \(\proj{\phase}{\netstate_1(u)} + \pert(\proj{\phase}{\netstate_1(u)},\proj{c}{\envstate_{r1}}, \epsilon) + 1 \leq T \). This implies
  \(f_{\proj{\phase}{\netstate_1(u)}} = \star\). Then, by the observation and the claim above, we
  also get
  \begin{align*}
    &\ptf(\popstate_1,\proj{\phase}{ \netstate_1(u)} , \fvec) \\
    &= 1 + \refr(\proj{\phase}{ \netstate_1(u)}, \pert(\proj{\phase}{ \netstate_1(u)}, \alpha^{\proj{\phase}{ \netstate_1(u)}}(\popstate_1, \fvec), \epsilon))&\{\proj{\phase}{ \netstate_1(u)} > R\} \\
    & = 1 + \proj{\phase}{ \netstate_1(u)} + \pert(\proj{\phase}{ \netstate_1(u)}, \alpha^{\proj{\phase}{ \netstate_1(u)}}(\popstate_1, \fvec), \epsilon) & \{\text{Claim}\}\\
    & = 1 + \proj{\phase}{ \netstate_1(u)} + \pert(\proj{\phase}{ \netstate_1(u)}, \proj{c}{ \envstate_{r\proj{\phase}{ \netstate_1(u)}+1}}), \epsilon) &\{\text{Obs.}\}\\
    &=\proj{\phase}{\netstate_1(u)} + \pert(\proj{\phase}{\netstate_1(u)},\proj{c}{\envstate_{r1}}, \epsilon) + 1& \{\text{Eq.}~(\ref{trans:update_in_cycle_pulse_stays_in_cycle_add})\}\\
    & = \proj{\phase}{\netstate_2(u)}\enspace. & 
  \end{align*}
  Finally, let \(\proj{\phase}{\netstate_1(u)} + \pert(\proj{\phase}{\netstate_1(u)},\proj{c}{\envstate_{r1}}, \epsilon) + 1 > T \), i.e. \(\proj{\phase}{\netstate_2(u)} = 1\). Then it has also to be the case that \(\proj{\phase}{\netstate_1(u)} + \pert(\proj{\phase}{\netstate_1(u)},\proj{c}{\envstate_{r\proj{\phase}{\netstate_1(u)}+1}}, \epsilon) + 1 > T\). By the claim above, this means
  \(\proj{\phase}{\netstate_1(u)} + \pert(\proj{\phase}{\netstate_1(u)},\alpha^{\proj{\phase}{\netstate_1(u)}}(\popstate_1,\fvec), \epsilon) + 1 > T\).  Hence, \(\ptf(\popstate_1,\proj{\phase}{ \netstate_1(u)} , \fvec) = 1\) as well.

  Now recall that we assumed initially that for all intermediate states \(\concstate = (\envstate, \netstate)\)  of the transition sequence, we have \(\proj{\location}{\envstate} = \update\).
  If this is not the case, we can partition the sequence into distinct subsequences, where this assumption holds for each subsequence, and apply the arguments above. 
 This proves the lemma. \qed
\end{proof}

Now we compare the probabilities of transition sequences in the different models.
\begin{lemma}
  \label{lem:prob-sum}
  Let \(D_c = (\concstates, \concstate_0, \prmatrix_c)\) be a concrete network of oscillators and 
  \(D_p = (\popstates, \popstate_0, \prmatrix_p)\) be its abstraction as a population model, as well
  as
 \(\concstate_1, \concstates^\prime\) and \(\popstate_1, \popstate_2 \in \popstates\), with \(h(\concstate_1) = \popstate_1\).
  Then, the sum of the probabilities of transition sequences from \(\concstate_1\) 
to all instantiations \(\concstate_2\)
with  \(h(\concstate_2) = \popstate_2\) is equal to the probability of the transition from \(\popstate_1\) to \(\popstate_2\).
\end{lemma}
\begin{proof}
%Now we turn to the probabilities of the transitions. To that end, we have to take into consideration that
%the one transition of the population model corresponds to multiple transitions in the concrete model.
% :
% \begin{enumerate}
% % \item Each state \(\popstate\) in the population model corresponds to a set of concrete states \(\concstate_i\).
% % The difference between states \(\concstate_i\) and \(\concstate_j\) is that, while the number of oscillators in
% % the phases is similar, different concrete oscillators are in the corresponding phases. \label{it:multiple_conc_states}
% \item For a given state \(\concstate_i\), the order in which the oscillators move from the 
% location \(\start\) to \(\update\) may be different. In particular, different oscillators
% may fail to broadcast during their firing transition. \label{it:multiple_orders}
% \end{enumerate}

Let \(\popstate = \langle k_1, \dots, k_T\rangle\). Furthermore, let \(N = \sum_{i=1}^T k_i\). 
%The number of concrete states corresponding
%to \(\popstate\) is given as follows.
%\begin{align*}
%  \mybinom{N}{k_1} \cdot \mybinom{N-k_1}{k_2} \cdot \dots \cdot \mybinom{N-(\sum_{i=1}^{T-1} k_i)}{k_T} = \frac{N!}{k_1! \cdot k_2! \cdot \dots \cdot k_T!} = \mybinom{N}{k_1, k_2, \cdots, k_T}
%\end{align*}
Now, let \(\concstate\) be an arbitrary state corresponding to \(\popstate\). If no oscillator fires, we have 
\(N!\) possibilities to create an transition sequence, each of which has a probability of \(\frac{1}{N!}\) to happen.
Hence, we get that the probability that one of these transitions happen is \(N! \cdot \frac{1}{N!} = 1\), which coincides
with the definition in the population model.

For the case that at least one oscillator fires and thus perturbs the other oscillators, we consider the
construction in the proof of Lemma~\ref{lem:conc-sim-pop} with respect to a failure vector \(F = \langle f_1, \dots, f_T\rangle\) for \(\popstate\). 
During each phase \(\Phase\), we have to choose the particular order of the \(k_\Phase\) and in addition,
we have to choose the set \(\failstates{\Phase}\). That is,
we have \(k_\Phase!\) possible orders, and \(\mybinom{k_\Phase}{f_\Phase}\) possibilities for the choice
of \(\failstates{\Phase}\). Furthermore, the combined  probability for the transitions of 
the  oscillators that should fire but are inhibited by a broadcast failure is
\begin{align*}
  \frac{1}{|\initloc{\concstate}{\Phase}|!} \cdot (1-\failureprob)^{k_\Phase - f_\Phase} \cdot \failureprob^{f_\Phase} \enspace .
\end{align*}
Observe that at the start of the construction of each phase, \(|\initloc{\concstate}{\Phase}| = k_\Phase\).
Hence the probability above simplifies to 
\begin{align*}
  \frac{1}{k_\Phase!} \cdot (1-\failureprob)^{k_\Phase - f_\Phase} \cdot \failureprob^{f_\Phase} \enspace .
\end{align*}
Due to the possible choices during the construction of the transition sequence, we have
that the probability of \emph{one} of these sequences to happen is
\begin{align*}
 \mybinom{k_\Phase}{f_\Phase} \cdot {k_\Phase}! \cdot \frac{1}{k_\Phase!} \cdot (1-\failureprob)^{k_\Phase - f_\Phase} \cdot \failureprob^{f_\Phase} =  \mybinom{k_\Phase}{f_\Phase} \cdot (1-\failureprob)^{k_\Phase - f_\Phase} \cdot \failureprob^{f_\Phase}  \enspace ,
\end{align*}
which is exactly the function \(\pfail(k_\Phase, f_\Phase)\) as in the population model.
Furthermore, with similar reasoning as above, the transition probability for the sequences, where no oscillator
is perturbed anymore, is \(1\).
Hence, we get that the combined probability of the set of paths from one instantiation \(\concstate_1\) of a population model state
\(\popstate_1\) to an instantiation of a successor \(\popstate_2\) of \(\popstate_1\) is equal to the probability of
the the transition from \(\popstate_1\) to \(\popstate_2\). \qed
\end{proof}

From Lemmas~\ref{lem:conc-sim-pop}, \ref{lem:pop-sim-conc} and \ref{lem:prob-sum}, we immediately get that the same weak simulation relation holds between a population model and 
the concrete network of oscillators it represents.

\begin{theorem}
  \label{thm:simulation}
  Let \(D_c = (\concstates, \concstate_0, \prmatrix_c)\) be a concrete network of oscillators and 
  \(D_p = (\popstates, \popstate_0, \prmatrix_p)\) be its abstraction as a population model.
  If we have \(\concstate_1, \concstate_2 \in \concstates^\prime\) and \(\popstate_1, \popstate_2 \in \popstates\), with \(h(\concstate_i) = \popstate_i\), then 
  \(
  %\begin{align*}
    \popstate_1 \poptrans^\ast \popstate_2  \text{ if, and only if, } \concstate_1 \conctransseq \concstate_2\).
 % \end{align*}
  Furthermore, the probabilities over all paths in both models coincide.
In particular, we have  \(h(\concstate_0) = \popstate_0\), and that \(\concstate_0\) weakly bisimulates \(\popstate_0\). 
\end{theorem}

Hence, we can use population models to analyse the global properties of a network
of pulse-coupled oscillators following the concrete model as defined in Sect.~\ref{sec:concrete} without
loss of precision. In particular, this allows us to increase the size of the network to check such properties,
while still giving us the opportunity to analyse the internal behaviour of nodes, if we restrict the
network size.
%\pgcomment{Add some tables here showing that the models for the concrete models get the same results as the models for the population model.}{}

\subsection{Experimental Validation} 
%\slcomment{Title is bad. Up for suggestions}

As Theorem~\ref{thm:simulation} implies, the synchronisation probabilities for a concrete model and its corresponding
population model coincide. However, we have to keep in mind that the PCTL formulas describing synchronisation are
 of course different. For a concrete model with four nodes and a cycle length \(T=10\), the synchronisation probability can be queried with  the following 
formula.
\begin{align*}
  \mathsf{sync}_c  & \equiv \pctlp_{=?}[\pctlsometime \proj{\phase}{1} = \proj{\phase}{2} \land \proj{\phase}{1} = \proj{\phase}{3} \land \proj{\phase}{1} = \proj{\phase}{4}  ]
\end{align*}
For a population model, the corresponding property is
\begin{align*}
  \mathsf{sync}_p  &\equiv \pctlp_{=?}[\pctlsometime \bigvee_{1 \leq i \leq 10} k_i = 4] 
\end{align*}
For both types of models, we defined a suitable input for the model checker \prism{}, and compared the results for different
values of \(R\), \(\epsilon\) and \(\failureprob\). As expected, the model checking results were matching exactly.   
%show the results of model checking the corresponding property for different values of \(R\), \(\epsilon\) and 
%\(\failureprob\) in Fig.~\ref{fig:comparison}. As expected, the two figures are exactly alike. 
% \begin{figure}[t]
%   \centering
% \begin{subfigure}{.49\linewidth}
%   \includegraphics[width=\textwidth]{figures/{concrete_n4_t10_e01_prob_vs_rp}.pdf}
%  % \caption{Power Consumption per Node to Achieve Synchronisation}
%   \label{fig:concrete_prob}
% \end{subfigure}
% % \subfloat[]{
% %   \includegraphics[width=.49\linewidth]{figures/{ms_n8_t10_e0.100000_ml0.200000_sync_vs_ord}.pdf}
% % %  \caption{Time in Cycles to Achieve Synchronisation}
% %   \label{fig:sync_vs_ord}
% % }
% \begin{subfigure}{.49\linewidth}
%   \includegraphics[width=\textwidth]{figures/{population_n4_t10_e01_prob_vs_rp}.pdf}
%  % \caption{Power Consumption per Node to Achieve Synchronisation}
%   \label{fig:population_prob}
% \end{subfigure}
%  \caption{Probability to Achieve Synchronisation}
% \label{fig:comparison}
%  \end{figure}
% \slcomment{is there really any value in this comparison???}
Table~\ref{tab:comparison_time} shows the model construction and checking times for some
exemplary parameter combinations of the models as reported by \prism\footnote{The experiments were run on a computer equipped with an Intel Core i7-7700 CPU
at 3.6 GHz and with 16GB of RAM. The version of \prism{} used was 4.4 beta.}. In the concrete model, the bulk of time 
is spent in the model checking phase, while the construction is much faster. For the analysis of 
the population model, however, the situation is reversed. The model construction phase is an order of magnitude
 longer than the model checking phase. As expected, the model checker needs less time for the analysis of the population model,
 if we add the time needed  for model construction and checking. 

% The table also indicates a trend we observed in our analysis: the time needed is the highest for refractory periods roughly
% half of the cycle length. While we did not conduct a formal investigation, intuitively this happens for the following
% reasons. 

% Assume that the refractory period is close to the length of the cycle. Then, an oscillator can only be 
% perturbed by other oscillators if its close to the end of its cycle. In particular this means, that the oscillators can
% only get ``more synchronised'' \slcomment{do we mention the order parameter in this paper?}, if they are already
% in similar phases in the initial state. Otherwise, they will never perturb each other, and hence the model checker will
% find that the run re-visits a state very early. Similarly, if oscillators are affected at the beginning of their cycle,
% which is only possible for short refractory periods, then they are only slightly perturbed, which may result in no
% perturbation at all, due to rounding the perturbation to discrete values. Hence, in such situations the oscillators
% do not get closer to synchronisation.

% However, if oscillators are always rather far within their cycle, as is the case if the refractory period is around
% half the length of the cycle, then the perturbation will affect the oscillators more strongly. Similarly, if the oscillators
% are not already synchronised, at least one firing of an oscillator will affect the others. Hence, 

\begin{table}
  \centering
  \caption{Model Construction Times and Model Checking Times for Both the Concrete and Population Model with \(N=4\), \(T=10\), \(\failureprob = 0.2\) and \(\epsilon = 0.1\) (in seconds).}
  \label{tab:comparison_time}
  \begin{tabular}{l|r|r|r|r}
  &\multicolumn{2}{c|}{Concrete Model}& \multicolumn{2}{c}{Population Model}\\
\hline
  \(R\) & Constr. & Check. & Constr. & Check\\
\hline
1 & 0.014 & 0.89  & 0.388  &0.017\\
5 & 0.069 & 10.38 &0.420 &0.047\\
8 & 0.056 & 1.46  &0.356 &0.013 \\
  \end{tabular}
\end{table}
%%% Local Variables: 
%%% mode: latex
%%% TeX-master: "sync_journal"
%%% End: 
\section{Conclusion}
\label{sec:conclusion}

In this paper we have introduced a formal concrete model for a network of nodes synchronising
their clocks over a set of discrete values. Furthermore, we developed a population model that
can alleviate state-space explosion when reasoning about significantly larger networks. We encoded both models as
discrete-time Markov chains, and formally connected them by showing that a concrete model of a
network weakly simulates a population model of that same network. We then showed that these two
models are equivalent with respect to the reachability of distinguished states, namely those
where all nodes in the network have synchronised their clocks.

%MORE ADVANTAGES HERE...

%While formal concrete models of synchronising networks give insight into the internal behaviour of
%nodes, the analysis is of course inhibited by the naivety of the approach. %\slcomment{this sounds a bit too negative by my liking}{}
%\pgcomment{open to suggestions for an alternative to 'naivety' then?}{}

Formalising the individual nodes of a network allows for the analysis of their internal properties.
 However,  this internal structure also inhibits the verification of global network properties.
Modelling the
whole network as the product of the models for the individual nodes quickly, and unsurprisingly, results in a model that is too large to analyse with existing tools and techniques.
While the use of appropriate collective abstractions, such as population models, allow for the analysis of
larger networks, they often impose restrictions on the topologies of the network that can be
considered. We could, of course, simply take the product of individual population models to
represent network structures more specialised than the fully-connected graphs considered here,
but again we face the consequences of this approach when trying to analyse the resulting model.
In addition, when using population models we lose the possibility to distinguish between nodes
having the same internal state. However, this does not restrict our analysis when considering
networks of homogeneous nodes where the properties of interest relate to global behaviours of the network itself.

% \begin{itemize}
% %\item something about non-positive phase response functions
% \item Is the fact that the 'inbetween' states in the concrete models cannot be reasoned about in the population model a disadvantage? What we show is useful because we are interested in the reachability of global states which are present in both the concrete and population models. 
% \end{itemize}

%xFUTURE WORK ETC...

Our current definition of pulse-coupled oscillators only allows for non-negative results of
the phase response function. However, there are also oscillator definitions with
phase response functions with possibly negative values~\cite{wang2012energy}. That is, instead
of shifting the state of an oscillator towards the end of the cycle, the perturbation may
reduce the value of the oscillator's state. It would be interesting to study the impact
of negative-valued phase response functions in the setting of discrete clock values.

%In typical definitions of PCO synchronisations \slcomment{citation}, phase response functions might be allowed
%delay the clock of
%the nodes. That is, if a node is closer to the beginning of its cycle, it reduces its
%current phase according to the received perturbance. Even though our current definition does not allow for
%negative phase responses, it might be interesting to see how negative responses 
%relate to discrete clock values.

While a concrete model can be instantiated to incorporate different topologies by explicit
encoding of possible perturbances in the nodes' transitions, it is by no means obvious how
to encorporate topologies into a population model. By design, the nodes in the latter are
indistinguishable, hence the differences in the connections between nodes are lost. We could
alleviate this restriction slightly, by modelling the connection between networks of strongly connected
components. That is, each component can be modelled by a different population model, and the firings
within one model can perturb different models. However, this would mean to compute the cross-product
of the population model, and hence we are back at the state-space explosion problem. Furthermore, our abstract
relation would need to take the mapping of single nodes into different components into account.

Deductive approaches might serve as an additional way to verify larger systems. In particular, 
due to the regularity of population models, we conjecture the existence of an inductive
invariant that holds from a certain size of models onwards. That is, as soon as the population
grows to a size to be treated as a single entity, we can increase this size by one node, and
guarantee that synchronisation still occurrs. For the population model sizes below this threshold, 
we could still use our proposed model-checking technique as the induction base. 
However, is is not clear what such an invariant should be, and how it can be verified.

% \begin{itemize}
% %\item topologies....
% %\item inductive proofs for networks... need deductive reasoning for this?
% \item something linking this to incremental analysis papers
% \end{itemize}

%%% Local Variables: 
%%% mode: latex
%%% TeX-master: "sync_journal"
%%% End: 

\bibliographystyle{spbasic}
\bibliography{lit,figo}

\end{document}